\documentclass{article}
\usepackage[margin=1in]{geometry}

\usepackage[utf8]{inputenc}
\usepackage[T1]{fontenc}
\usepackage{lmodern}

\usepackage{amsmath,amssymb, amsfonts, amsthm, mathtools, bm} 
\usepackage{empheq}
\newtheorem{theorem}{Theorem}
\newtheorem{definition}[theorem]{Definition}

\newtheorem{proposition}[theorem]{Proposition}
\newtheorem{lemma}[theorem]{Lemma}

\usepackage{tabularx}
\usepackage{enumerate}
\usepackage{graphicx}

\usepackage{euscript}

\usepackage{algorithm}
\usepackage{algorithmicx}
\usepackage{algpseudocode}

\usepackage{calc}
\usepackage{tikz}
\usetikzlibrary{calc}
\usetikzlibrary{decorations.markings}
\usetikzlibrary{arrows,calc,decorations.pathmorphing}
\usetikzlibrary{decorations.pathreplacing,shapes.misc}
\usetikzlibrary{positioning}

\usepackage{xcolor}
\definecolor{orange}{rgb}{1.0,0.5,0}
\definecolor{violet}{rgb}{0.6,0,0.8}
\definecolor{darkgreen}{rgb}{0,0.5,0}
\definecolor{verydarkgreen}{rgb}{0,0.3,0}
\definecolor{darkblue}{rgb}{0,0,0.6}
\definecolor{darkred}{rgb}{0.75,0,0}
\definecolor{grey}{rgb}{0.35,0.35,0.35}

\DeclareMathOperator{\detour}{detour}
\DeclareMathOperator{\peak}{peak}
\DeclareMathOperator{\parent}{parent}
\DeclareMathOperator{\pred}{pred}

\newcommand\OPT{\ensuremath{\mathrm{OPT}}}

\def\cupp{\stackrel{.}{\cup}}
\def\bigcupp{\stackrel{.}{\bigcup}}
\renewcommand{\epsilon}{\varepsilon}

\newcommand{\sfrac}[2]{{\textstyle\frac{#1}{#2}}} 

\def\Ieuscr{\EuScript{I}}
\def\Jeuscr{\EuScript{J}}

\def\Hscr{\mathcal{H}}
\def\Iscr{\mathcal{I}}

\def\Pscr{\mathcal{P}}
\def\Qscr{\mathcal{Q}}

\def\Tscr{\mathcal{T}}

\def\Yscr{\mathcal{Y}}
\def\Zscr{\mathcal{Z}}

\def\Tg{T}

\usepackage[%
 pdfpagelabels,
 hyperindex,
 bookmarksopen=true,
 bookmarksopenlevel=1,
]{hyperref}

\hypersetup{
    colorlinks=true,
    linkcolor=darkblue,
    citecolor = darkgreen,
}

\makeatletter
\let\@fnsymbol\@arabic
\makeatother

\author{Jannis Blauth\thanks{
Research Institute for Discrete Mathematics and Hausdorff Center for Mathematics, University of Bonn, Germany.
Email: {blauth@or.uni-bonn.de}.
}
 \and 
Vera Traub\thanks{
Department of Mathematics, ETH Zurich, Switzerland.
Email: {vera.traub@ifor.math.ethz.ch}.
Supported by Swiss National Science Foundation grant 200021\_184622.
}
\and 
Jens Vygen\thanks{
Research Institute for Discrete Mathematics and Hausdorff Center for Mathematics, University of Bonn, Germany.
Email: {vygen@or.uni-bonn.de}.
}
}

\date{}

\title{Improving the Approximation Ratio for \\ Capacitated Vehicle Routing}

\begin{document}

\maketitle

\begin{abstract}
We devise a new approximation algorithm for capacitated vehicle routing.
Our algorithm yields a better approximation ratio for general capacitated vehicle routing
as well as for the unit-demand case and the splittable variant.
Our results hold in arbitrary metric spaces.
This is the first improvement upon the classical tour partitioning algorithm by 
Haimovich and Rinnooy Kan \cite{haimovich} and Altinkemer and Gavish \cite{altinkemer}.
\end{abstract}

\section{Introduction}

In the \textsc{Capacitated Vehicle Routing} problem, we are given a metric space with a depot and customers,
each with a positive demand between $0$ and $1$. The goal is to design tours of minimum total length such that each tour contains the depot,
each customer is served by some tour, and the total demand of the customers in one tour does not exceed 1 (after scaling,
this is the vehicle capacity).
\textsc{Capacitated Vehicle Routing} generalizes the famous traveling salesman problem and has obvious applications in logistics.
There is a huge body of literature studying heuristics, mixed-integer programming models, and application scenarios.

The so far best known approximation algorithm is more than 30 years old and quite simple:
it first computes a traveling salesman tour (ignoring the capacity constraint)
and then partitions the tour optimally into segments of total demand at most 1, 
each of which is then served by a separate tour from the depot. 
The approximation ratio of this algorithm is $\alpha + 2$, where $\alpha$ is the approximation ratio of an algorithm computing the traveling salesman tour.
Essentially the same algorithm has been the best known for the unit-demand special case (where all customers have the same demand), and also for the variant 
where a customer's demand can be split and served by more than one tour.
For these special cases, the approximation ratio is $\alpha + 1$.

These algorithms have been proposed and analyzed in the 1980s by 
Altinkemer and Gavish \cite{altinkemer} and Haimovich and Rinnooy Kan \cite{haimovich}.
Despite many efforts and progress in special cases (cf.\ Section~\ref{sec:relatedwork}), 
they have not been improved, except that the traveling salesman tour can now be computed by the Karlin--Klein--Oveis Gharan algorithm \cite{karlin}
instead of the Christofides--Serdjukov algorithm \cite{christofides,serdjukov},
which improves $\alpha$ to slightly less than $\frac{3}{2}$ if one allows randomization.

In this paper we improve upon the classical algorithms of \cite{altinkemer} and \cite{haimovich}. 
Our result is a better black-box reduction to the traveling salesman problem. 
Therefore, our new algorithm has a better approximation ratio than the classical algorithms of \cite{altinkemer} and \cite{haimovich}, 
and this will remain true if the approximation ratio for the traveling salesman problem will be improved further. 
Here are our main results:

\begin{theorem}\label{thm:main_CVR}
For every $\alpha>1$ there is an $\epsilon>0$ such that the following holds.
If there is an $\alpha$-approximation algorithm for the traveling salesman problem, 
then there is an $(\alpha + 2\cdot (1-\epsilon))$-approximation algorithm for
\textsc{Capacitated Vehicle Routing}.
For $\alpha=\frac{3}{2}$ we have $\epsilon>\frac{1}{3000}$.
\end{theorem}

\begin{theorem}\label{thm:main_UD_and_S}
For every $\alpha>1$ there is an $\epsilon>0$ such that the following holds.
If there is an $\alpha$-approximation algorithm for the traveling salesman problem, 
then there is an $(\alpha + 1-\epsilon)$-approximation algorithm for
\textsc{Unit-Demand Capacitated Vehicle Routing} and \textsc{Splittable Capacitated Vehicle Routing}.
For $\alpha=\frac{3}{2}$ we have $\epsilon>\frac{1}{3000}$.
\end{theorem}

\subsection{Outline}

To obtain our results, we analyze instances for which the approximation guarantees of \cite{altinkemer} and \cite{haimovich} 
are almost tight and exploit their structure to design better solutions. 
We will call such instances \emph{difficult}.

We view every tour in a solution to a \textsc{Capacitated Vehicle Routing} instance
as the union of two paths from the depot to the \emph{peak} of the tour: the point farthest away from the depot. 
Our first observation is that the performance of the classical algorithms can be close to the worst case guarantee
only if, for most tours, these two paths have small detour, i.e., they are approximately shortest paths from the depot to the peak.

We will compute an even number of paths that all start at the depot such that all customers are visited by some path.
The total length of these paths will not be much more than the length of an optimum solution to the \textsc{Capacitated Vehicle Routing} instance.
Then we combine pairs of these paths to tours by adding an edge between their endpoints.

If there exists a set of paths with small total detour (like the one induced by an optimum solution to a difficult instance), 
then we can find a set of paths that is not much longer in polynomial time.
In fact, this problem is closely related to regret-bounded vehicle routing, 
a problem that has been studied by Friggstad and Swamy \cite{friggstad,friggstad2017compact}.
Here, one asks for a minimum number of paths serving all customers such that the detour of any path is bounded.

However, combining pairs of paths to tours can be too expensive. 
We need to ensure that a relatively cheap matching of the endpoints of the paths exists.
Ideally, two paths end at the peak of each tour in an optimum solution, then the matching would not cost anything.
But of course we do not know these peaks.
Therefore we try to ``guess'' them, by exploiting another property of difficult instances:
in almost all tours of an optimum solution, the total demand of customers near the peak is almost 1 (the vehicle capacity).
Consequently, we can assume that ``most'' customers are clustered, and we can force two paths to end in each cluster. 

However, another difficulty arises because the clusters are not necessarily clearly separated from one another. 
Still we can identify groups of nearby clusters, and estimate the number of tours whose peak is in that group.
Instead of prescribing the endpoints of the paths, we only specify the total number of paths that must end in each group. 
This number will always be even, in order to ensure that we can find a matching within each group. 
Although customers in the same group can be far away from each other if there is a chain of pairwise overlapping clusters, 
we will be able to prove that a relatively cheap matching exists. 

The key subproblem therefore asks to find an appropriate number of paths that begin at the depot 
and end in these target groups, such that all customers 
(including those that do not belong to any group) are served by some tour.
We call this problem \textsc{Vehicle Routing with Target Groups}.
The instance of \textsc{Vehicle Routing with Target Groups} that we compute
has the property that it has a solution that is cheap and has small total detour. 
This will enable us to find a cheap set of paths in polynomial time:
either by a simple and fast combinatorial algorithm, or alternatively by leveraging an LP-based approach suggested
for regret-bounded vehicle routing by Friggstad and Swamy \cite{friggstad2017compact}.\footnote{As a by-product, 
we will also improve their approximation ratio for regret-bounded vehicle routing from 15 to 10.}

Once we have these paths, we compute the cheapest matching of their endpoints and combine them to tours.
These tours will generally still not meet the capacity constraint, but we can simply concatenate all these tours
(and shortcut) to obtain a traveling salesman tour. Since this tour will be not much more expensive than an optimum
solution to our \textsc{Capacitated Vehicle Routing} instance, applying the classical tour partitioning algorithm finishes the job.

\subsection{Formal problem description}

Given a depot $s$ and a set $V$ of customers, we want to design tours serving all customers.
For now, a \emph{tour} is a cycle that contains $s$ and a subset of customers 
(later we will also consider tours that begin in $s$ but do not end in $s$).
To measure the cost of a tour, we have a semi-metric $c:(\{s\}\cup V) \times (\{s\}\cup V) \rightarrow \mathbb{R}_{\geq 0}$,
i.e., $c$ is symmetric and satisfies the triangle inequality. 
We will interpret a tour $Q$ as an undirected graph with vertex set $V(Q)$ and edge set $E(Q)$.
We write $c(Q)=\sum_{\{v,w\}\in E(Q)}c(v,w)$ for the cost (or total distance) of $Q$.  
Moreover, each customer has a demand, and the total demand of the customers served by a tour must
not exceed the vehicle capacity, which we can assume to be 1 (by scaling).
Then the problem can be described as follows.

\begin{definition}[\textsc{Capacitated Vehicle Routing}]
An instance of \textsc{Capacitated Vehicle Routing} consists of 
\begin{itemize}\itemsep0pt
\item a finite set $V$ (of \emph{customers})
\item a \emph{depot} $s$, not belonging to $V$,
\item a semi-metric $c$ on $\{s\}\cup V$, defining \emph{distances} (or \emph{cost}),
\item a \emph{demand}  $d(v) \in [0,1]$ for each customer $v\in V$. 
\end{itemize}
A feasible solution is a set $\Qscr$ of \emph{tours} such that
\begin{itemize}\itemsep0pt
\item every tour $Q\in \Qscr$ is a cycle that contains $s$,
\item every customer belongs to exactly one tour, and
\item $\sum_{v\in V(Q)\setminus\{s\}} d(v) \le 1$ for all $Q\in\Qscr$.
\end{itemize}
The task is to find a feasible solution such that the total cost (or distance) $c(\Qscr):= \sum_{Q\in \Qscr} c(Q)$ is minimum. 
\end{definition}

Throughout, we will denote by $\OPT(\Ieuscr)$ or simply $\OPT$ 
the minimum cost of a feasible solution to a given instance $\Ieuscr$. 
We note the following well-known lower bound:

\begin{proposition}\label{prop:trivial_lower_bound}
$\OPT\ge\sum_{v\in V}2d(v)c(s,v)$.
\end{proposition}

\begin{proof}
Let $\Qscr$ be a feasible solution.
For each $v\in V$ we obtain two $s$-$v$-paths by splitting the tour $Q\in\Qscr$ that contains $v$.
By the triangle inequality, each of these paths has length at least $c(s,v)$, and hence 
$2c(s,v)\le c(Q)$. Summation yields 
$\sum_{v\in V}2d(v)c(s,v) \le \sum_{Q\in\Qscr} \sum_{v\in V(Q)\setminus\{s\}} d(v) c(Q) \le c(\Qscr)$.
\end{proof}

If $d(v)=\frac{1}{k}$ for all $v\in V$, where $k$ is some positive integer, 
we speak of \textsc{Unit-Demand Capacitated Vehicle Routing} (then every tour can serve up to $k$ customers). 
This is closely related to \textsc{Splittable Capacitated Vehicle Routing}:
here the demand of a customer is arbitrary but can be split into several parts, each of which is served by a different tour.

All variants include the traveling salesman problem as special case and are thus APX-hard.
\textsc{Capacitated Vehicle Routing} also includes bin packing; hence there is no approximation algorithm with ratio less than $\frac{3}{2}$ unless P=NP.

\subsection{Related work\label{sec:relatedwork}}

Despite of a huge amount of research on vehicle routing, the best known approximation ratio 
for \textsc{Capacitated Vehicle Routing} (as well as for the unit-demand and splittable variants) 
has not been improved in more than 30 years.
However, there has been progress on several special cases.
First of all, the tour partitioning algorithms by 
Haimovich and Rinnooy Kan \cite{haimovich} and by
Altinkemer and Gavish \cite{altinkemer} 
(cf.\ Section~\ref{sect:classical_algo}) already yield a slightly better approximation guarantee  
if the least common denominator $k$ of all demands is bounded (this is often called the \emph{bounded capacity} case). 
Compared to $\alpha +2$ and $\alpha +1$, the approximation ratios reduce by $\frac{2\alpha}{k}$ (for general demands) 
and $\frac{\alpha}{k}$ (for unit demands). 
Bompadre, Dror and Orlin \cite{bompadre} gain another $\Omega(\frac{1}{k^3})$.

There are also several results for geometric instances for the unit demand case (in which $d(v) =\frac{1}{k}$ for all $v\in V$).
In the Euclidean plane a PTAS is known for constant $k$ (Haimovich and Rinnooy Kan \cite{haimovich}), 
for $k = O(\log(n)/\log\log(n))$, where $n=|V|$ (Asano, Katoh, Tamaki, and Tokuyama \cite{Asano2}),
and for $k \le 2^{\log^{f(\epsilon)}(n)}$ (Adamszek, Czumaj, and Lingas \cite{Adamszek}).
The latter uses a result by Das and Mathieu \cite{das}, who provided a quasi-polynomial time approximation scheme for the Euclidean plane and unbounded $k$.
For higher dimensional Euclidean metrics, Khachay and Dubinin \cite{khachay} found a PTAS
for fixed dimension $l$ and $k = O(\log^{\frac{1}{l}}(n))$. 

Better approximation ratios have also been found for graph metrics arising from graphs with a special structure. 
For the unit-demand case with constant $k$,
Becker, Klein and Schild \cite{becker3} devised a PTAS in planar graphs, 
and Becker, Klein and Saulpic \cite{becker2} found a PTAS
in graphs with bounded highway dimension.
Becker \cite{becker} designed a $\frac{4}{3}$-approximation algorithm for \textsc{Splittable Capacitated Vehicle Routing} in tree metrics, improving on results by Hamaguchi and Katoh \cite{Hamaguchi} and Asano, Kawashima and Katoh \cite{Asano1}.

For general \textsc{Capacitated Vehicle Routing}, no improvement on the classical approximation algorithm \cite{altinkemer} has been found except for tree metrics. 
For tree metrics, Labbé, Laporte and Mercure \cite{Labbe} gave a $2$-approximation algorithm.
If the tree is a path (i.e., on the line), 
Wu and Lu \cite{wu} described a $\frac{5}{3}$-approximation algorithm for \textsc{Capacitated Vehicle Routing}.
Note that the unit-demand case is polynomially solvable on the line.

\medskip

One part of our proof is leveraging an LP relaxation that was proposed by Friggstad and Swamy \cite{friggstad2017compact} for regret-bounded vehicle routing. 
In (additive) regret-bounded vehicle routing, the goal is to find a minimum number of paths starting at the depot and covering all customers such that none of the tours has a detour more than a given bound. 
Here, the detour (or regret) of a path from $s$ to $t$ is its length minus $c(s,t)$. 
Friggstad and Swamy \cite{friggstad} provided the first constant-factor approximation algorithm for this problem. 
They improved the approximation ratio from 31 to 15 in \cite{friggstad2017compact}. 

Regret-bounded vehicle routing is a special case of the school bus problem. 
In the school bus problem, there is the additional constraint that no tour can serve
more customers than a given bound (the vehicle capacity).
Bock, Grant, K\"onemann and Sanit\`a \cite{bock} 
observed that any $\mu$-approximation algorithm for regret-bounded vehicle routing
implies a $(\mu+1)$-approximation algorithm for the school bus problem.
They provided a 3-approximation algorithm for regret-bounded vehicle routing on trees and thus obtain a 4-approximation algorithm for the school bus problem on trees. 
The later results by Friggstad and Swamy mentioned above imply a 16-approximation algorithm
for the general school bus problem.\footnote{As a by-product of our work, we will improve on this in Section~\ref{sect:improve_FS}.}

\subsection{Review of the classical algorithms}\label{sect:classical_algo}

In this section we review the classical algorithms by 
Altinkemer and Gavish \cite{altinkemer} and Haimovich and Rinnooy Kan \cite{haimovich}, and we do this for two reasons. 
First, we will exploit properties of instances in which their analysis is tight.
Second, the final step of our new algorithm will be identical to these classical algorithms.

The classical algorithms \cite{altinkemer,haimovich} consist of two steps.
The first step simply runs an approximation algorithm for the traveling salesman problem,
and we denote by $\alpha$ the approximation ratio of this algorithm.
The classical Christofides--Serdjukov algorithm \cite{christofides,serdjukov} obtains $\alpha=\frac{3}{2}$,
the new randomized algorithm by Karlin, Klein, and Oveis Gharan \cite{karlin} improves on this by a tiny constant.

A \emph{traveling salesman tour} (for a given instance $(V,s,c,d)$ of \textsc{Capacitated Vehicle Routing}) 
is a cycle with vertex set $\{s\}\cup V$. Note that the demands are ignored.
The minimum length of a traveling salesman tour is a lower bound on $\OPT$ because
we can concatenate any set of tours (they all contain $s$) and shortcut in order to visit $s$ only once;
here we use the triangle inequality.
So the first step yields a traveling salesman tour of length $\alpha\cdot\OPT$.

In the second step, this traveling salesman tour is partitioned into several tours in order to meet the capacity constraint.
This is achieved by the following theorem:

\begin{theorem}[\cite{altinkemer,haimovich}]\label{thm:tour_splitting}
Given an instance $(V,s,c,d)$ of \textsc{Capacitated Vehicle Routing} and a traveling salesman tour $Q$,
one can compute a feasible solution of cost at most $c(Q)+\sum_{v\in V}4d(v)c(s,v)$ in $O(n^2)$ time, where $n=|V|$.
For \textsc{Unit-Demand Capacitated Vehicle Routing} and \textsc{Splittable Capacitated Vehicle Routing} the bound
improves to $c(Q)+\sum_{v\in V}2d(v)c(s,v)$.
\end{theorem} 

\begin{proof}
Number the customers $V=\{v_1,\ldots,v_n\}$ in the order they are visited by $Q$.
Choose $\theta \in (0,1)$ uniformly at random. 
Then there are unique indices $1\le j_1 < ... <  j_{k_{\max}} \le n$, where $k_{\max} = 1+\lfloor d(V)-\theta\rfloor$, such that
\[
\sum\limits_{i=1}^{j_{l}-1} d(v_i) < \theta + l-1 \leq \sum\limits_{i=1}^{j_l} d(v_i) \text{  } \text{ for all } l \in \{1,\ldots,k_{\max} \}.
\]

Partition the traveling salesman tour $Q$ by constructing for each $l \in \{j_1,\dots, j_{k_{\max}}\}$ a single tour visiting only the customer $v_l$ 
and tours visiting the resulting segments of $Q$ in between these customers. 
Note that by construction all these tours meet the capacity constraint.

By using the triangle inequality, the total cost of the resulting tours can be bounded by
\[
c(Q) + 4 \cdot \sum\limits_{l=1}^{k_{\max}} c(s,v_{j_l}).
\]

Note that the probability that a customer $v$ is contained in $\{v_{j_1},..,v_{j_{k_{\max}}}\}$ equals $d(v)$. Hence, the expected total cost of the resulting tours can be bounded by 
\[
c(Q) + \sum\limits_{v \in V} 4 d(v) c(s,v).
\]

Thus, there exists at least one such partition of $Q$ with the required total cost. 
Since there are at most $n$ values of $\theta$ leading to different partitions, 
all these values can be tried and the best among the resulting partitions can be chosen. 
Clearly, this can be done in $O(n^2)$ time.

For unit demands and for splittable demands one does not need a separate tour serving $v_l$ only (for $l=1,\ldots,k_{\max}$),
but can include $v_l$ into the preceding segment (and in the splittable case partly into the succeding segment).
Then the total cost of the resulting tours can be bounded by
\[
c(Q) + \sum\limits_{v \in V} 2 d(v) c(s,v). 
\]
\end{proof}
Note that one can even compute an optimum partition of the given traveling salesman tour in $O(n^2)$ time by a simple dynamic program. 
However, this does not improve the above upper bound on the total cost of the resulting set of tours.

Together with Proposition~\ref{prop:trivial_lower_bound}, Theorem~\ref{thm:tour_splitting} immediately implies the approximation ratio $\alpha+2$
for \textsc{Capacitated Vehicle Routing} \cite{altinkemer}, and the approximation ratio $\alpha+1$
for \textsc{Unit-Demand Capacitated Vehicle Routing} and \textsc{Splittable Capacitated Vehicle Routing} \cite{haimovich}.
In this paper, we present the first improvement on these more-than-thirty-year-old results.

\section{Overview of our new algorithm\label{sec:overview}}

In this section we provide an outline of our approach and introduce some important definitions.

Our algorithm will compute two solutions to the \textsc{Capacitated Vehicle Routing} instance and return the better of the two. 
The first solution results from the classical algorithm that we discussed in Section~\ref{sect:classical_algo}.
The second algorithm is new and yields a better approximation ratio on instances 
where Proposition~\ref{prop:trivial_lower_bound} and thus the analysis of the classical algorithms is almost tight.
We will call such instances \emph{difficult}.
In order to give a formal definition of difficult instances we fix a constant
$0<\epsilon<1$.\footnote{We will set this and other constants in Section~\ref{sect:final_result}.}

\begin{definition}
An instance $(V,s,c,d)$ of \textsc{Capacitated Vehicle Routing} is called \emph{difficult} if
\[
2 \cdot \sum\limits_{v \in V} d(v) \cdot c(s,v) \ > \ (1 - \epsilon) \cdot \OPT.
\]
\end{definition}

If an instance of \textsc{Capacitated Vehicle Routing} is not difficult, 
Theorem~\ref{thm:tour_splitting} says that
the classical algorithm described in Section~\ref{sect:classical_algo} computes a solution of cost at most 
$(\alpha+2 \cdot (1 -\epsilon))\cdot \OPT$, and at most $(\alpha+1 -\epsilon)\cdot \OPT$ for the \textsc{Unit-Demand} and \textsc{Splittable} variants. 
For difficult instances we will compute a traveling salesman tour that is shorter than $\alpha\cdot\OPT$;
then the better approximation ratios follow from calling Theorem~\ref{thm:tour_splitting}.

Informally speaking, in an optimum solution to a difficult instance almost every tour must have the following two properties:
\begin{enumerate}[(1)]
    \item The tour $Q$ is not much longer than $2\cdot c(s,\peak(Q))$, where $\peak(Q)$ is the vertex in $Q$ that is farthest away from the depot $s$.
    \label{item:informal_no_detour_to_peak}
    \item The total demand served by the tour is almost 1, and almost all of it is close to the peak.
    \label{item:clustered_informal}
\end{enumerate}
More precisely, the total length of the tours that don't have the above properties is very small compared to $\OPT$. We will show this in Section~\ref{sect:difficult_instances_clustered}.

Property~\eqref{item:clustered_informal} implies that most vertices can be partitioned into clusters with demand approximately $1$.
The first step of our algorithm aims at identifying these clusters.
Since the clusters around the peaks of different tours can be close to each other and might be difficult to distinguish, 
we merge such clusters into larger ones.
A detailed description of our clustering algorithm will be given in Section~\ref{sect:guessing_clusters}.

Having identified clusters, we now want to find paths starting at the depot and ending in the clusters, 
such that all vertices are visited by one such path, regardless of whether the vertex is part of a cluster or not.
To find such paths we compute a solution to an instance of \textsc{Vehicle Routing with Target Groups}, which is a new problem we introduce.
The \emph{targets} will be some vertices inside the clusters, where targets in the same cluster belong to the same \emph{target group}.

\begin{definition}[\textsc{Vehicle Routing with Target Groups}]
An instance of \textsc{Vehicle Routing with Target Groups} consists of 
\begin{itemize}\itemsep0pt
\item disjoint finite sets $V$ (of \emph{customers}) and $\bar T$ (of \emph{targets}),
\item a \emph{depot} $s$, not belonging to $V\cup \bar T$,
\item a semi-metric $c$ on $\{s\}\cup V\cup \bar T$,
\item a partition  $\mathcal{T}$ of the target set $\bar T$ into \emph{target groups}, and
\item numbers $b: \mathcal{T} \rightarrow \mathbb{Z}_{>0}$ that specify how many tours must end in each target group.
\end{itemize}
A feasible solution is a set $\Pscr$ of \emph{tours} such that
\begin{itemize}\itemsep0pt
\item every tour $P\in \Pscr$ is either an $s$-$t$-path for some target $t\in \bar T$ or a cycle containing $s$, 
and all other vertices of $P$ belong to $V$,
\item every element of $V$ belongs to at least one of these tours, and
\item for every target group $\Tg  \in \Tscr$, exactly $b(\Tg  )$ of these tours end in an element of $\Tg  $.
\end{itemize}
The task is to find a feasible solution whose total cost
$c(\Pscr):=\sum_{P\in \Pscr} c(P)$ is minimum, where we again write
$c(P)=\sum_{\{v,w\}\in E(P)}c(v,w)$.
\end{definition}
Again, we write $\OPT(\Jeuscr)$ for the cost of an optimum solution to instance $\Jeuscr$.

For a target group $\Tg  $, we will set the number $b(\Tg  )$ to roughly twice the demand of the cluster containing~$\Tg  $. 
The number $b(\Tg  )$ of paths ending in $\Tg  $ will always be even.
Therefore, we can turn a solution to our instance of \textsc{Vehicle Routing with Target Groups} into a solution of the \textsc{Capacitated Vehicle Routing} instance as follows.
For each target group $\Tg  $ we pair up the paths and complete a pair of paths to a tour by adding an edge between their endpoints.
Computing a pairing that minimizes the cost of the added edges is a minimum-weight perfect matching problem.
In Section~\ref{sect:matching} we prove that this matching problem has a cheap solution.
In the end we concatenate all these tours and shortcut to obtain a traveling salesman tour,
to which we apply Theorem~\ref{thm:tour_splitting}.

A key part of our proof is to show that we can compute a cheap solution to our instance of \textsc{Vehicle Routing with Target Groups}. 
First, we want to show that if our \textsc{Capacitated Vehicle Routing} instance $\Ieuscr$ is difficult, the instance of \textsc{Vehicle Routing with Target Groups} that we construct has a solution that
\begin{enumerate}[(i)]
    \item is not much more expensive than the cost $\OPT (\Ieuscr)$ of an optimum solution to our \textsc{Capacitated Vehicle Routing} instance, and \label{item:cheap_solution_informal}
    \item has small total \emph{detour}, i.e.\ almost all tours are approximately shortest paths from the depot to the target they end in. \label{item:small_detour_informal}
\end{enumerate}

It will be convenient to consider a solution to a \textsc{Vehicle Routing with Target Groups} instance as a set of walks in a digraph.
For an instance $\Jeuscr$ of \textsc{Vehicle Routing with Target Groups} let $G(\Jeuscr)$ be the digraph with vertex set
$V(\Jeuscr)=\{s\}\cup V\cup \bar T$ and edge set
$E(\Jeuscr)=E_1(\Jeuscr)\cup E_2(\Jeuscr)$, where
\begin{align*}
E_1(\Jeuscr)&=\{(v,w) : v\in \{s\}\cup V,\, w\in V,\, v\not=w\}, \\
E_2(\Jeuscr)&=\{(v,w) : v\in \{s\}\cup V,\, w\in \{s\}\cup \bar T,\, v\not=w\}.
\end{align*}

\begin{definition}[walk solution]\label{def:walk_solution}
Given an instance $\Jeuscr=(V,\bar T,s,c,\mathcal{T}, b)$ of \textsc{Vehicle Routing with Target Groups},
a \emph{walk solution} is a multi-subset $H$ of $E(\Jeuscr)$ such that 
\begin{itemize}
    \item $|\delta^+_H(v)|=|\delta^-_H(v)|$ for all $v\in V$,
    \item $|\delta^-_H(\Tg  )|=-b(\Tg  )$ for each $\Tg  \in\Tscr$, and
    \item  $H$ connects all vertices of $\{s\}\cup V$.
\end{itemize}
Here $\delta^+_H(v)$ and $\delta^-_H(v)$ denote the multi-sets of arcs in $H$ leaving and entering $v$, respectively.
\end{definition}

The \emph{cost} of a walk solution $H$ is $c(H):=\sum_{(v,w)\in H} c(v,w)$, where the sum counts multiplicities. 
Then \textsc{Vehicle Routing with Target Groups} is equivalent to computing a cheapest walk solution:

\begin{proposition}\label{prop:walk_solution_suffices}
For any instance $\Jeuscr$ of \textsc{Vehicle Routing with Target Groups}, there is a walk solution of cost $\OPT(\Jeuscr)$,
and from any walk solution $H$ one can obtain a solution $\Pscr$ with $c(\Pscr)\le c(H)$ in $O(|H|)$ time.
\end{proposition}

\begin{proof}
Given a solution $\Pscr$ to an instance $\Jeuscr$, orient its tours away from $s$, and
let $H$ contain an edge $e\in E(\Jeuscr)$ exactly $k$ times if the number of tours $P\in\Pscr$ that contain $e$ is $k$.
Then $H$ is a walk solution of the same cost.

Conversely, we can decompose a walk solution $H$ in linear time into $\sum_{\Tg  \in\Tscr}b(\Tg  )$ walks 
from $s$ to $\{s\}\cup \bar T$ with inner vertices in $V$ such that each element of $V$ belongs to at least one of them.
Then we can shortcut these walks whenever an element of $V$ is visited more than once. 
Shortcutting does not increase the total cost due to the triangle inequality. 
This yields a set $\Pscr$ of tours with $c(\Pscr) \le c(H)$.
\end{proof}

\begin{definition}[detour]
For an arc $e=(v,w)\in E(\Jeuscr)$ we define 
$$\detour(e) \ := \ \detour(v,w) \ := \ c(v,w)+c(s,v)-c(s,w).$$
For a walk $P$ that visits the vertices $v_0,v_1,\ldots,v_k$ in this order
we write $\detour(P):=\sum_{i=1}^{k} \detour(v_{i-1},v_i)$.
\end{definition}
For a walk $P$ that starts at the depot $s$ and ends in $t$ we have 
$\detour(P)= c(P)-c(s,t)$.
The detour has been called \emph{excess} by \cite{blum2007approximation} and
the \emph{regret metric} by \cite{friggstad}.
The $\detour$ is clearly not symmetric, but it is non-negative and fulfills the triangle inequality.

Having constructed an instance $\Jeuscr$ of \textsc{Vehicle Routing with Target Groups} from
a given difficult instance $\Ieuscr$ of \textsc{Capacitated Vehicle Routing},
we want to show that $\Jeuscr$ has a cheap solution (costing not much more than $\OPT(\Ieuscr)$) with small total detour.
However, this is not exactly the statement we will show.
It will simplify our proofs and lead to better approximation ratios to consider what we call \emph{weak fractional solutions} of $\Jeuscr$ instead of actual solutions.
First, we allow taking only fractions of walks. 
Second, we reduce the amount that must arrive in each target group slightly.
The latter is because we cannot assume that the total demand of a tour near its peak is 1 in an optimum solution to $\Ieuscr$,
but only $1-\tau$ for some constant $\tau \in (0,1)$ to be chosen later.

\begin{definition}[weak fractional solution]\label{def:weak_fractional_solution}
A \emph{weak fractional solution} to an instance $\Jeuscr = (V,\bar T,s,c,\Tscr,b)$ of \textsc{Vehicle Routing with Target Groups} 
is a vector $x\in \mathbb{R}_{\ge 0}^{E(\Jeuscr)}$ such that 
\begin{equation*}
    x = \sum_{P\in \Pscr} \lambda_P \cdot \chi^{E(P)},
\end{equation*}
where $\Pscr$ is a set of walks in $G(\Jeuscr)$ and $\lambda_P \in \mathbb{R}_{\ge 0}$ for all $P\in \Pscr$ such that
\begin{itemize}
    \item every walk $P\in \Pscr$ begins in $s$ and ends in $\{s\}\cup \bar T$, and all inner vertices belong to $V$,
    \item for every $v\in V$, we have $\sum_{P\in \Pscr: v\in V(P)} \lambda_P \ge 1$, and
    \item for every target group $\Tg  \in \Tscr$, the total weight of the walks ending in $\Tg  $ equals $(1-\tau)\cdot b(\Tg  )$, i.e., 
    \[\sum_{P\in \Pscr:\, P \text{ ends in } \Tg} \lambda_P =(1-\tau) \cdot b(\Tg  ).\]
\end{itemize}
Here $\chi^{E(P)}$ denotes the incidence vector of $E(P)$.
We write $c(x):=\sum_{(v,w)\in E(\Jeuscr)}c(v,w)x_{(v,w)}$
and $\detour(x):=\sum_{(v,w)\in E(\Jeuscr)}\detour(v,w)x_{(v,w)}$
for any vector $x\in \mathbb{R}_{\ge 0}^{E(\Jeuscr)}$. 
\end{definition}

We show that the instance $\Jeuscr$ of \textsc{Vehicle Routing with Target Groups} has a weak fractional solution that has small detour
and is not much more expensive than $\OPT(\Ieuscr)$.
To this end, we start with an optimum solution to our \textsc{Capacitated Vehicle Routing} instance $\Ieuscr$.
Then for every customer $v$ that is contained in a tour $Q$ of this solution, we partition $Q$
into two paths starting at the depot $s$ and ending at $v$.
We extend these paths by an edge connecting $v$ either to a close-by target or to the depot $s$.
The latter will happen rarely, as we can show using property~\eqref{item:clustered_informal} and
the fact that we constructed the targets inside clusters and set the value $b$ to roughly twice the demand of its cluster.
The resulting walks $P$ will then contribute with weight $\lambda_P=d(v)$ to the weak fractional solution that we construct.
Our weak fractional solution will have small detour because of property~\eqref{item:informal_no_detour_to_peak} 
of optimum solutions of difficult instances.
A precise description of our construction will be given in Section~\ref{sect:good_solution_exists}.

While \textsc{Vehicle Routing with Target Groups} in general is at least as hard as the traveling salesman problem, we can compute solutions that are not much longer
than the best weak fractional solution with small detour.
More precisely, if there exists a weak fractional solution $x$ with small detour, then we can compute a solution to \textsc{Vehicle Routing with Target Groups} which is not much more expensive than $x$. 
This is formally stated in the following theorem. For small detour, we can choose a small value of $\eta$.
Then the factor on the cost of the given weak fractional solution is close to 1 because we will choose the constant $\tau \in (0,1)$ to be close to $0$.

\begin{theorem}\label{thm:mainVRPwTargetGroups}
There is a polynomial-time algorithm for \textsc{Vehicle Routing with Target Groups} that
computes for every instance $\Jeuscr$ and any given $\eta\in (0,1]$
a feasible solution $\Pscr$ of $\Jeuscr$ such that 
\[
c(\Pscr) \ < \ 
(\sfrac{1}{1-\tau}+\eta)\cdot c(x) 
+ O(\sfrac{1}{\eta}) \cdot \detour(x|_{E_1(\Jeuscr)}),
\]
for every weak fractional solution $x$ of $\Jeuscr$. 
\end{theorem}

We will present two approaches for solving \textsc{Vehicle Routing with Target Groups}, both implying Theorem~\ref{thm:mainVRPwTargetGroups}.
In both approaches we compute a cheap forest and a network flow to obtain tours that visit not necessarily every customer, but every connected component of the forest.
In the network flow problem we ensure that the number of tours ending in each target group meets the requirements.
Doubling the edges of the forest yields the desired walk solution of \textsc{Vehicle Routing with Target Groups}. 

The first approach (Section~\ref{sec:combinatorial}) is a simple and fast combinatorial algorithm based on the following observation.
While computing a cheapest walk solution is NP-hard, we can in polynomial time compute a cheapest walk solution with the additional property 
that every vertex $v$ has at least one predecessor that is closer to the depot than $v$. 
This problem can be reduced to a network flow problem. 
We compute such a walk solution for a subset of customers for which we can guarantee that our additional constraint makes the walk solution not too expensive.
To find such a subset of customers we use a simple greedy algorithm.
Finally, we  connect the remaining customers that are not visited by our walk solution by a minimum-cost forest.

The second approach (Section~\ref{sec:lpbased}) leverages a sophisticated LP relaxation for regret-bounded vehicle routing due to Friggstad and Swamy \cite{friggstad2017compact}. 
In contrast to the combinatorial approach, here the network flow only uses edges moving away from the depot, i.e.\ edges $(v,w)$ with $c(s,w) > c(s,v)$, (and some edges entering $s$).
Both the forest and the network flow are obtained from an optimum LP solution.
However, setting up an LP relaxation that allows to round fractional solutions to an integral forest and an (almost) acyclic network flow that visits all connected components of the forest
 turns out to be tricky.
To achieve this, the LP has variables corresponding to a network flow in some auxiliary digraph.
In this digraph we have $O(n^2)$ many vertices for each customer $v\in V$, which correspond to different possibilities for the forest component containing $v$.
We combine the rounding approach by \cite{friggstad2017compact} with a new
construction of a fractional solution, which will also enable us to obtain a 
better approximation ratio for regret-bounded vehicle routing and the school bus problem.

Sections~\ref{sec:combinatorial} and \ref{sec:lpbased} present the two different proofs of 
Theorem~\ref{thm:mainVRPwTargetGroups} and can be read independently from each other.

\section{Reducing to Vehicle Routing with Target Groups}

In this section we prove that Theorem~\ref{thm:mainVRPwTargetGroups} implies our main result.
First, in Section~\ref{sect:difficult_instances_clustered} we discuss a key property of difficult instances that we exploit in later parts of this chapter.
In Section~\ref{sect:guessing_clusters}, we describe the construction of an instance of \textsc{Vehicle Routing with Target Groups}. 
In Section~\ref{sect:matching} we bound the cost of the matching that we use to
transform a \textsc{Vehicle Routing with Target Groups} solution to a traveling salesman tour.
In Section~\ref{sect:good_solution_exists} we prove that our instance 
of  \textsc{Vehicle Routing with Target Groups}
has a weak fractional solution with cost close to $\OPT(\Ieuscr)$ and small detour, given that the \textsc{Capacitated Vehicle Routing} instance $\Ieuscr$ we consider is difficult.
Hence, Theorem~\ref{thm:mainVRPwTargetGroups} will yield a good bound.
Finally, in Section~\ref{sect:solving_difficult} we combine all this to give an algorithm with a good approximation ratio for difficult instances.

\subsection{Difficult instances are clustered}\label{sect:difficult_instances_clustered}

Let $0 < \tau, \rho \le \frac{1}{6}$ be constants that we will fix later.
\begin{definition}\label{def:peak_cluster}
	Let $(V,s,c,d)$ be an instance of \textsc{Capacitated Vehicle Routing}. Let $Q$ be a cycle with $s \in V(Q)$. Then we define
	$\peak(Q)$ to be a vertex $v\in V(Q)$ with $c(s,v)$ maximal,
	and the \emph{peak cluster} to be 
	\[C(Q) \ := \ \left\{u \in V(Q) : c(u, \peak(Q)) + \kappa \cdot \detour(u, \peak(Q))  < \rho \cdot c(s,\peak(Q))\right\}.\]
	where $\kappa:= \sfrac{1 - 2 \tau - \tau \cdot \rho}{2\tau}$. We call the peak cluster \emph{large} if $d(C(Q)) > 1-\tau$ and \emph{small} otherwise.
\end{definition}

Here and in the following we abbreviate $d(C(Q))=\sum_{v\in C(Q)}d(v)$.
Note that each vertex in the peak cluster is at most $\rho \cdot c(s,\peak(Q))$ away from $\peak(Q)$. 
We also remark that $s\notin C(Q)$ because $\rho<1$.
See Figure~\ref{fig:example_tight} for an example.
For later use we remark that $\kappa>\frac{3}{2}$ because $\tau,\rho\le\frac{1}{6}$.

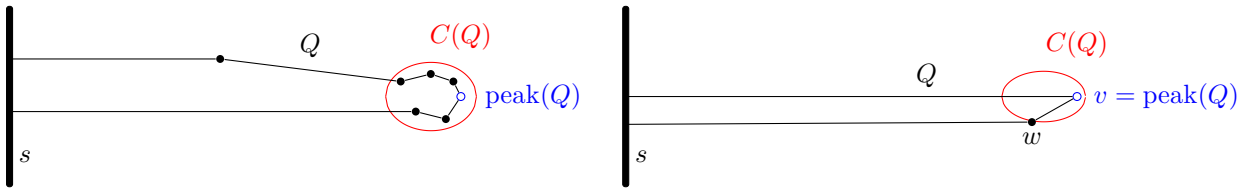
\begin{figure}[b!]
		\begin{tikzpicture}
		\begin{scope}[scale=2]
		
		\draw[draw=red] plot[domain=2.5:90/29, samples = 50] (\x, {sqrt(max(0,1/(23/12+1)^2 * (
			(1/6 + 23/12)^2 *3*3 + (23/12)^2 * \x^2 - 2 * (1/6+23/12)*3*23/12*\x - (1+23/12)^2 * (\x-3)^2)))});
		\draw[draw=red] plot[domain=2.5:90/29, samples = 50] (\x, {-sqrt(max(0,1/(23/12+1)^2 * (
			(1/6 + 23/12)^2 *3*3 + (23/12)^2 * \x^2 - 2 * (1/6+23/12)*3*23/12*\x - (1+23/12)^2 * (\x-3)^2)))});

		\tikzstyle{vertex}=[black,circle,fill,minimum size=3,inner sep=0pt]
		\tikzstyle{vertex2}=[draw=blue,circle,minimum size=3,inner sep=0pt]
		\node[vertex2] (v) at (3,0) {};			
		\node[blue, right] at (3.1,0) {$\peak(Q)$};
		\node[vertex] (1) at (2.9,-0.15) {};
		\node[vertex] (3) at (2.6,0.1) {};
		\node[vertex] (4) at (2.8,0.15) {};
		\node[vertex] (5) at (1.4,0.25) {};
		\node[vertex] (7) at (2.95,0.1) {};
		\node[vertex] (8) at (2.7,-0.1) {};
		\draw[black,fill, rounded corners=1pt] (-0.02,-0.6) rectangle (0.02,0.6);
		\node[black, right] at (0,-0.4) {$s$};
		
		\node[black, above] at (2,0.2) {$Q$};
		\node[red, above] at (3,0.25) {$C(Q)$};

		\draw (0,-0.1)--(8)--(1)--(v)--(7)--(4)--(3)--(5)--(0,0.25);

		\end{scope}
		\end{tikzpicture}
		\hfill
		\begin{tikzpicture}
		\begin{scope}[scale=2]
		
		\draw[draw=red] plot[domain=2.5:162/53, samples = 50] (\x, {sqrt(max(0,1/(235/60+1)^2 * (
			(1/6 + 235/60)^2 *3*3 + (235/60)^2 * \x^2 - 2 * (1/6+235/60)*3*235/60*\x - (1+235/60)^2 * (\x-3)^2)))});
		\draw[draw=red] plot[domain=2.5:162/53, samples = 50] (\x, {-sqrt(max(0,1/(235/60+1)^2 * (
			(1/6 + 235/60)^2 *3*3 + (235/60)^2 * \x^2 - 2 * (1/6+235/60)*3*235/60*\x - (1+235/60)^2 * (\x-3)^2)))});

		\tikzstyle{vertex}=[black,circle,fill,minimum size=3,inner sep=0pt]
		\tikzstyle{vertex2}=[draw=blue,circle,minimum size=3,inner sep=0pt]
		
		\node[vertex2] (v) at (3,0) {};			
		\node[blue, right] at (3.05,0) {$v = \peak(Q)$};
		\node[vertex] (w) at (2.7,-0.17) {};			
		\node[below] at (2.7,-0.175) {$w$};
		\draw[black,fill, rounded corners=1pt] (-0.02,-0.6) rectangle (0.02,0.6);
		\node[black, right] at (0,-0.4) {$s$};
		
		\node[black, above] at (2,0) {$Q$};
		\node[red, above] at (3,0.2) {$C(Q)$};
		
		\draw (0,-0.185)--(w)--(v)--(0,0);
		
		\end{scope}
		\end{tikzpicture}
	\caption{The left-hand side shows a tour $Q$ and its peak cluster $C(Q)$,
	assuming Euclidean distances and $\rho = \frac{1}{6}$ and $\tau = \frac{1}{6}$, which yields $\kappa = \frac{23}{12}$. 
	If each customer has demand $\frac{1}{7}$, the peak cluster of $Q$ is large since $d(C(Q)) > 1 - \tau$. 
On the right-hand side, we show another tour and its peak cluster for $\rho = \frac{1}{6}$ and $\tau = \frac{1}{10}$. This yields $\kappa = \frac{47}{12}$. 
Suppose the customer $v= \peak(Q)$ has demand $1- \tau$, and the customer $w$, which lies slightly outside of $C(Q)$, has demand $\tau$. 
Then the peak cluster of this tour is small. It can be shown that in this case the inequality in Lemma~\ref{lem:difficult_instances_clustered} is  tight.\label{fig:example_tight}}
\end{figure}

Intuitively, in an optimum solution to a difficult instance almost every tour $Q$ has the property that almost all vertices $v\in V(Q)$ are contained in the peak cluster $C(Q)$.
More precisely, the total length of tours with small peak cluster will be small compared to the total length of all tours. 
This can be derived from the following lemma, which gives a lower bound on $c(Q) - 2 \cdot \sum_{v\in V(Q)} d(v) \cdot c(s,v)$ for tours $Q$ with a small peak cluster. 
Because in a difficult instance the sum of these expressions over all tours is small (it is less than $\epsilon \cdot \OPT$), the lemma implies that in an optimum solution the total length of such tours with small peak clusters must be small. 
Note that we will choose $\tau$ and $\rho$ such that  $\tau \cdot \rho$ is much larger than $\epsilon$.

\begin{lemma}\label{lem:difficult_instances_clustered}
	Let $Q$ be a tour with small peak cluster. Then
	\begin{equation*}
	c(Q) - 2 \cdot \sum_{v\in V(Q)\setminus\{s\}} d(v) \cdot c(s,v) \ \ge \ \tau \cdot \rho \cdot  c(Q).
	\end{equation*}
\end{lemma}
\begin{proof}
Note that $s\in V(Q)\setminus C(Q)$. Let $u\in V(Q)\setminus C(Q)$ with $c(s,u)$ maximum. 
We have $c(s,v)\le c(s,u)$ for $v\in V(Q)\setminus C(Q)$ and $c(s,v)\le c(s,\peak(Q))$ for all $v\in V(Q)$.
Since the peak cluster is small,
\begin{equation}
\label{eq:new_proof_1}
\sum_{v\in V(Q)\setminus\{s\}} d(v)\cdot c(s,v) \ \le \ \tau \cdot c(s,u) + (1-\tau)\cdot c(s,\peak(Q)).
\end{equation}
Next, $u\notin C(Q)$ implies $c(u, \peak(Q)) + \kappa \cdot \detour(u, \peak(Q)) \ge \rho \cdot c(s,\peak(Q))$.
Plugging in $\kappa=\sfrac{1 - 2 \tau - \tau \cdot \rho}{2\tau}$ and
$\detour(u,\peak(Q))=c(u,\peak(Q))+c(s,u)-c(s,\peak(Q))$ and multiplying by $2\tau$ yields
\begin{equation}
\label{eq:new_proof_2}
(1-\tau\rho)\cdot c(u, \peak(Q)) + (1 - 2 \tau - \tau \rho) \cdot c(s,u) \ \ge \ (1-2\tau+\tau\rho) \cdot c(s,\peak(Q)).
\end{equation}
Multiplying \eqref{eq:new_proof_1} by 2 and combining with \eqref{eq:new_proof_2} yields
\begin{align*}
2\cdot \sum_{v\in V(Q)\setminus\{s\}} d(v)\cdot c(s,v) &\ \le \ 2\tau \cdot c(s,u) + (2-2\tau)\cdot c(s,\peak(Q)) \\
&\ \le \ (1-\tau\rho) \cdot (c(s,\peak(Q)) + c(u,\peak(Q))+c(s,u)) \\
&\ \le \ (1-\tau\rho) \cdot c(Q),
\end{align*}
where the last inequality holds because $Q$ contains $s$, $u$, and $\peak(Q)$.
\end{proof}

 Figure~\ref{fig:example_tight} shows that $C(Q)$ is chosen smallest possible such that Lemma~\ref{lem:difficult_instances_clustered} holds. 
 Choosing $\kappa = 0$ would already be sufficient to improve the approximation ratio of \textsc{Capacitated Vehicle Routing}. In this case the peak cluster of a tour $Q$ is a ball of radius $\rho \cdot c(s, \peak(Q))$ around $\peak(Q)$. However, choosing $\kappa=0$ would yield a worse approximation ratio.

\subsection{Clustering algorithm}\label{sect:guessing_clusters}

In this section we describe an algorithm to construct an instance of \textsc{Vehicle Routing Problem with Target Groups}.
The resulting instance is the one to which we will later apply Theorem~\ref{thm:mainVRPwTargetGroups}.
Before giving a formal description of the algorithm, let us informally explain some important properties.
Intuitively, our algorithm tries to "guess" the peaks of the tours of an optimum solution; these should become target vertices and in each target $t$ should end two tours, i.e. $b(\{t\})=2$. 

Our algorithm will not always guess the peaks of the optimum tours correctly.
However, for every tour $Q$ with a large peak cluster, our algorithm will always guess a target vertex that is not far away from $\peak(Q)$.
For every target vertex $t\in \bar T$, we then consider an area $B_t$ around $t$ that is chosen large enough to guarantee that each large peak cluster of an optimum solution is fully contained in one of these areas.
Note that our algorithm might also guess targets that are not close to any of the peaks of optimum tours, but we will show that for difficult instances this happens rarely.

We want to set the numbers $b$ in the \textsc{Vehicle Routing with Target Groups} instance large enough so that for every tour with a large peak cluster, two paths are allowed to end in a target close to the peak.
Therefore, the number of paths ending in the targets will depend on the total demand $d$ in the area $B_t$ around $t$.
In order to avoid requiring a too high number of paths, we want to avoid that the demand $d(v)$ of a customer $v$ is counted twice here when $v$ is contained in two of the areas $B_t$.
Therefore, if the areas $B_t$ for different targets $t$ overlap, e.g.\ because the peaks of two tours are close to each other,
we merge the areas $B_t$ and the corresponding targets will form a group.

Let us now describe our algorithm to construct an instance of \textsc{Vehicle Routing with Target Groups} 
 (Algorithm~\ref{algo:constructing_vrpwtg_instance}). See also Figure~\ref{fig:example_reduction_vrpwtg} for an illustration.
Note that the set $\bar T$ of targets will be a subset of $V$,
which formally means that we duplicate the targets. 

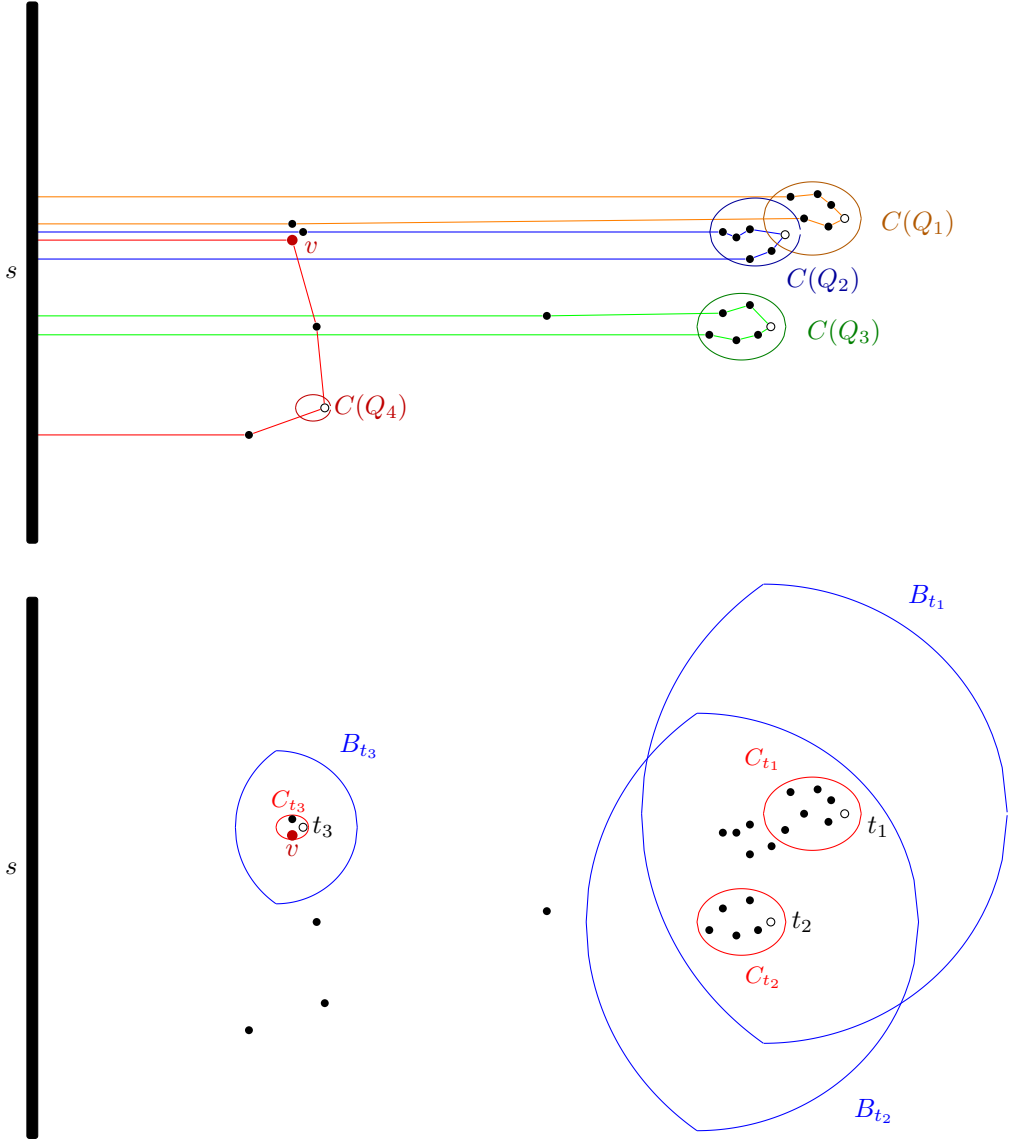
\begin{figure}[p]
	\begin{center}
		\begin{tikzpicture}
\begin{scope}[scale=3.6]
\begin{scope}[shift={(0,2.2)}]

\tikzstyle{vertex}=[black,circle,fill,minimum size=3,inner sep=0pt]
\tikzstyle{vertex2}=[draw=black,circle,minimum size=3,inner sep=0pt]
\tikzstyle{vertexred}=[darkred,circle,fill,minimum size=4,inner sep=0pt]

\node[vertex] (a3) at (2.95,0.25) {};
\node[vertex] (a2) at (2.8,0.28) {};
\node[vertex] (a5) at (2.85,0.2) {};
\node[vertex] (a6) at (2.9,0.29) {};
\node[vertex] (a7) at (2.94,0.17) {};

\node[vertex2] (b4) at (2.78,0.14) {};
\node[vertex] (b3) at (2.73,0.08) {};
\node[vertex] (b2) at (2.65,0.05) {};
\node[vertex] (b5) at (2.65,0.16) {};
\node[vertex] (b6) at (2.55,0.15) {};
\node[vertex] (b7) at (2.6,0.13) {};

\node[vertex] (c4) at (2.6,-0.25) {};
\node[vertex] (c2) at (2.55,-0.15) {};
\node[vertex] (c5) at (2.5,-0.23) {};
\node[vertex] (c1) at (1.9,-0.16) {};
\node[vertex] (c6) at (2.65,-0.12) {};
\node[vertex] (c7) at (2.68,-0.23) {};

\node[vertex] (a1) at (0.96,0.18) {};
\node[vertexred] (d2) at (0.96,0.12) {};
\node[darkred, below] at (1.03,0.15) {$v$};

\node[vertex] (d3) at (1.05,-0.2) {};
\node[vertex2] (d4) at (1.08,-0.5) {};
\node[vertex] (d5) at (0.8,-0.6) {};

\node[vertex2] (a4) at (3,0.2) {};		
\node[vertex2] (c3) at (3/1.1,-0.2) {};		
\node[vertex] (b1) at (1,0.15) {};

\draw[draw=orange] (0,0.28)--(a2)--(a6)--(a3)--(a4)--(a7)--(a5)--(a1)--(0,0.18);
\draw[draw=blue] (0,0.15)--(b1)--(b6)--(b7)--(b5)--(b4)--(b3)--(b2)--(0,0.05);
\draw[draw=green] (0,-0.16)--(c1)--(c2)--(c6)--(c3)--(c7)--(c4)--(c5)--(0,-0.23);
\draw[draw=red] (0,0.12)--(d2)--(d3)--(d4)--(d5)--(0,-0.6);

\draw[black,fill, rounded corners=1pt] (-0.02,-1) rectangle (0.02,1);
\node[black, left] at (-0.02,0) {$s$};

\node[orange!70!black, right] at (3.1,0.18) {$C(Q_1)$};
\node[darkblue, right] at (2.75,-0.03) {$C(Q_2)$};
\node[darkgreen, right] at (3/1.1+0.1,-0.22) {$C(Q_3)$};
\node[darkred, right] at (1.08,-0.5) {$C(Q_4)$};

\draw[draw=orange!70!black] plot[domain=2.7:150/49, samples = 50] (\x, {0.2+sqrt(max(0,1/(39/20+1)^2 * (
	(1/10 + 39/20)^2 *9 + (39/20)^2 * \x^2 - 2 * (1/10+39/20)*3*39/20*\x - (1+39/20)^2 * (\x-3)^2)))});
\draw[draw=orange!70!black] plot[domain=2.7:150/49, samples = 50] (\x, {0.2-sqrt(max(0,1/(39/20+1)^2 * (
	(1/10 + 39/20)^2 *9 + (39/20)^2 * \x^2 - 2 * (1/10+39/20)*3*39/20*\x - (1+39/20)^2 * (\x-3)^2)))});

\draw[draw=darkblue] plot[domain=2.7/(3/2.78):150/49/(3/2.78), samples = 50] (\x, {0.15+sqrt(max(0,1/(39/20+1)^2 * (
	(1/10 + 39/20)^2 *3/(3/2.78)*3/(3/2.78) + (39/20)^2 * \x^2 - 2 * (1/10+39/20)*3/(3/2.78)*39/20*\x - (1+39/20)^2 * (\x-3/(3/2.78))^2)))});
\draw[draw=darkblue] plot[domain=2.7/(3/2.78):150/49/(3/2.78), samples = 50] (\x, {0.15-sqrt(max(0,1/(39/20+1)^2 * (
	(1/10 + 39/20)^2 *3/(3/2.78)*3/(3/2.78) + (39/20)^2 * \x^2 - 2 * (1/10+39/20)*3/(3/2.78)*39/20*\x - (1+39/20)^2 * (\x-3/(3/2.78))^2)))});

\draw[draw=darkgreen] plot[domain=2.7/1.1:150/49/1.1, samples = 50] (\x, {-0.2+sqrt(max(0,1/(39/20+1)^2 * (
	(1/10 + 39/20)^2 *3/1.1*3/1.1 + (39/20)^2 * \x^2 - 2 * (1/10+39/20)*3/1.1*39/20*\x - (1+39/20)^2 * (\x-3/1.1)^2)))});
\draw[draw=darkgreen] plot[domain=2.7/1.1:150/49/1.1, samples = 50] (\x, {-0.2-sqrt(max(0,1/(39/20+1)^2 * (
	(1/10 + 39/20)^2 *3/1.1*3/1.1 + (39/20)^2 * \x^2 - 2 * (1/10+39/20)*3/1.1*39/20*\x - (1+39/20)^2 * (\x-3/1.1)^2)))});	

\draw[draw=darkred] plot[domain=2.7/(3/1.08):150/49/(3/1.08), samples = 50] (\x, {-0.5+sqrt(max(0,1/(39/20+1)^2 * (
	(1/10 + 39/20)^2 *1.08*1.08 + (39/20)^2 * \x^2 - 2 * (1/10+39/20)*1.08*39/20*\x - (1+39/20)^2 * (\x-1.08)^2)))});
\draw[draw=darkred] plot[domain=2.7/(3/1.08):150/49/(3/1.08), samples = 50] (\x, {-0.5-sqrt(max(0,1/(39/20+1)^2 * (
	(1/10 + 39/20)^2 *1.08*1.08 + (39/20)^2 * \x^2 - 2 * (1/10+39/20)*1.08*39/20*\x - (1+39/20)^2 * (\x-1.08)^2)))});	

\end{scope}
\draw[draw=red] plot[domain=2.7:150/49, samples = 50] (\x, {0.2+sqrt(max(0,1/(39/20+1)^2 * (
	(1/10 + 39/20)^2 *9 + (39/20)^2 * \x^2 - 2 * (1/10+39/20)*3*39/20*\x - (1+39/20)^2 * (\x-3)^2)))});
\draw[draw=red] plot[domain=2.7:150/49, samples = 50] (\x, {0.2-sqrt(max(0,1/(39/20+1)^2 * (
	(1/10 + 39/20)^2 *9 + (39/20)^2 * \x^2 - 2 * (1/10+39/20)*3*39/20*\x - (1+39/20)^2 * (\x-3)^2)))});

\draw[draw=blue] plot[domain=9/4:2.7, samples = 50] (\x,{0.2+sqrt(max(0,1 / 9 * \x * \x - (3 - \x) * (3 - \x)))});
\draw[draw=blue] plot[domain=9/4:2.7, samples = 50] (\x,{0.2 - sqrt(max(0,1 / 9 * \x * \x - (3 - \x) * (3 - \x)))});
\draw[draw=blue] plot[domain=18/5:2.7, samples = 50] (\x,{0.2+sqrt(max(0,(18 / 10 - 1 / 3 * \x) * (18 / 10 - 1 / 3 * \x) - (3 - \x) * (3 - \x)))});
\draw[draw=blue] plot[domain=18/5:2.7, samples = 50] (\x,{0.2- sqrt(max(0,(18 / 10 - 1 / 3 * \x) * (18 / 10 - 1 / 3 * \x) - (3 - \x) * (3 - \x)))});

\draw[draw=red] plot[domain=2.7/1.1:150/49/1.1, samples = 50] (\x, {-0.2+sqrt(max(0,1/(39/20+1)^2 * (
	(1/10 + 39/20)^2 *3/1.1*3/1.1 + (39/20)^2 * \x^2 - 2 * (1/10+39/20)*3/1.1*39/20*\x - (1+39/20)^2 * (\x-3/1.1)^2)))});
\draw[draw=red] plot[domain=2.7/1.1:150/49/1.1, samples = 50] (\x, {-0.2-sqrt(max(0,1/(39/20+1)^2 * (
	(1/10 + 39/20)^2 *3/1.1*3/1.1 + (39/20)^2 * \x^2 - 2 * (1/10+39/20)*3/1.1*39/20*\x - (1+39/20)^2 * (\x-3/1.1)^2)))});

\draw[draw=blue] plot[domain=9/4/1.1:2.7/1.1, samples = 50] (\x,{-0.2+sqrt(max(0,1 / 9 * \x * \x - (3/1.1 - \x) * (3/1.1 - \x)))});
\draw[draw=blue] plot[domain=9/4/1.1:2.7/1.1, samples = 50] (\x,{-0.2- sqrt(max(0,1 / 9 * \x * \x - (3/1.1 - \x) * (3/1.1 - \x)))});

\draw[draw=blue] plot[domain=18/5/1.1:2.7/1.1, samples = 50] (\x,{-0.2+sqrt(max(0,(6 * 3/1.1 / 10 - 1 / 3 * \x)^2 - (3/1.1 - \x) * (3/1.1 - \x)))});
\draw[draw=blue] plot[domain=18/5/1.1:2.7/1.1, samples = 50] (\x,{-0.2- sqrt(max(0,(6 * 3/1.1 / 10 - 1 / 3 * \x)^2 - (3/1.1 - \x) * (3/1.1 - \x)))});

\draw[draw=red] plot[domain=2.7/3:150/49/3, samples = 50] (\x, {0.15+sqrt(max(0,1/(39/20+1)^2 * (
	(1/10 + 39/20)^2 *3/3*3/3 + (39/20)^2 * \x^2 - 2 * (1/10+39/20)*3/3*39/20*\x - (1+39/20)^2 * (\x-3/3)^2)))});
\draw[draw=red] plot[domain=2.7/3:150/49/3, samples = 50] (\x, {0.15-sqrt(max(0,1/(39/20+1)^2 * (
	(1/10 + 39/20)^2 *3/3*3/3 + (39/20)^2 * \x^2 - 2 * (1/10+39/20)*3/3*39/20*\x - (1+39/20)^2 * (\x-3/3)^2)))});

\draw[draw=blue] plot[domain=9/4/3:2.7/3, samples = 50] (\x,{0.15+sqrt(max(0,1 / 9 * \x * \x - (3/3 - \x) * (3/3 - \x)))});
\draw[draw=blue] plot[domain=9/4/3:2.7/3, samples = 50] (\x,{0.15- sqrt(max(0,1 / 9 * \x * \x - (3/3 - \x) * (3/3 - \x)))});

\draw[draw=blue] plot[domain=18/5/3:2.7/3, samples = 50] (\x,{0.15+sqrt(max(0,(6 * 3/3 / 10 - 1 / 3 * \x)^2 - (3/3 - \x) * (3/3 - \x)))});
\draw[draw=blue] plot[domain=18/5/3:2.7/3, samples = 50] (\x,{0.15- sqrt(max(0,(6 * 3/3 / 10 - 1 / 3 * \x)^2 - (3/3 - \x) * (3/3 - \x)))});

\tikzstyle{vertex}=[black,circle,fill,minimum size=3,inner sep=0pt]
\tikzstyle{vertex2}=[draw=black,circle,minimum size=3,inner sep=0pt]
\tikzstyle{vertexred}=[darkred,circle,fill,minimum size=4,inner sep=0pt]

\node[vertex] (a3) at (2.95,0.25) {};
\node[vertex] (a2) at (2.8,0.28) {};
\node[vertex] (a5) at (2.85,0.2) {};
\node[vertex] (a6) at (2.9,0.29) {};
\node[vertex] (a7) at (2.94,0.17) {};

\node[vertex] (t) at (2.78,0.14) {};
\node[vertex] (b3) at (2.73,0.08) {};
\node[vertex] (b2) at (2.65,0.05) {};
\node[vertex] (b5) at (2.65,0.16) {};
\node[vertex] (b6) at (2.55,0.13) {};
\node[vertex] (b7) at (2.6,0.13) {};

\node[vertex] (t) at (2.6,-0.25) {};
\node[vertex] (t) at (2.55,-0.15) {};
\node[vertex] (t) at (2.5,-0.23) {};
\node[vertex] (t) at (1.9,-0.16) {};
\node[vertex] (c6) at (2.65,-0.12) {};
\node[vertex] (c7) at (2.68,-0.23) {};

\node[vertex] (t) at (0.96,0.18) {};
\node[vertexred] (t) at (0.96,0.12) {};

\node[vertex] (t) at (1.05,-0.2) {};
\node[vertex] (t) at (1.08,-0.5) {};
\node[vertex] (t) at (0.8,-0.6) {};
\node[darkred, below] at (0.96,0.12) {$v$};

\node[vertex2] (t) at (3,0.2) {};			
\node[black, right] at (3.05,0.15) {$t_1$};
\node[vertex2] (t) at (3/1.1,-0.2) {};			
\node[black, right] at (2.77,-0.2) {$t_2$};
\node[vertex2] (t) at (1,0.15) {};			
\node[black, right] at (1,0.15) {$t_3$};		
\draw[black,fill, rounded corners=1pt] (-0.02,-1) rectangle (0.02,1);
\node[black, left] at (-0.02,0) {$s$};

\node[blue, right] at (3.2,1) {$B_{t_1}$};
\node[blue, right] at (3,-0.9) {$B_{t_2}$};
\node[blue, right] at (1.1,0.45) {$B_{t_3}$};
\node[red, left] at (2.8,0.4) {\small{$C_{t_1}$}};
\node[red, below] at (2.7,-0.33) {\small{$C_{t_2}$}};
\node[red, above] at (0.95,0.17) {\small{$C_{t_3}$}};

\end{scope}
\end{tikzpicture}
\caption{The upper figure shows an instance $\Ieuscr=(V,s,c,d)$ of \textsc{Capacitated Vehicle Routing} with a solution $\Qscr$. The customer $v$ (shown in red) has demand $\frac{4}{7}$ and all other customers have demand $\frac{1}{7}$. The peak clusters are drawn for $\rho = \frac{1}{10}$ and $\tau= \frac{1}{6}$ and Euclidean distances. The peaks of the tours are shown as empty circles
Note that the peak clusters $C(Q_1)$, $C(Q_2)$ and $C(Q_3)$ are large. \newline
The lower figure shows the corresponding output $\Jeuscr = (V ,\bar T, s, c, \Tscr, b)$ of Algorithm~\ref{algo:constructing_vrpwtg_instance}. 
The targets $\bar T=\{t_1,t_2,t_3\}$ (shown as empty circles here) are partitioned into target groups $\Tscr=\{\{t_1,t_2\},\{t_3\}\}$. 
A solution of $\Jeuscr$ consists of $b(\{t_1,t_2\}) = 6$ tours ending in $\{t_1,t_2\}$ and $b(\{t_3\})=2$ tours ending in $t_3$. 
We will see in the next subsection that it does not harm that $t_3$ was selected as a target although it does not belong to any peak cluster. 
Note that $\peak(Q_2)$ was not identified as a target even though $Q_2$ has a large peak cluster. 
However, the peak cluster of $Q_2$ lies completely in $B_{\{t_1,t_2\}}$, which will guarantee that $b(\{t_1,t_2\})$ is chosen large enough so that $\Qscr$ can be easily transformed into a weak fractional solution of $\Jeuscr$ that costs not much more than $\Qscr$.
}\label{fig:example_reduction_vrpwtg}
	\end{center}
\end{figure}

\begin{algorithm}[H]
\caption{Constructing an instance of \textsc{Vehicle Routing with Target Groups} \label{algo:constructing_vrpwtg_instance}}

\vspace*{2mm}
\textbf{Input:} Instance $(V,s,c,d)$ of \textsc{Capacitated Vehicle Routing}. \\
\textbf{Output:} Instance $(V,\bar T, s, c, \Tscr, b)$ of \textsc{Vehicle Routing with Target Groups}.\\

\begin{algorithmic}[1]
\State Number the vertices in $V$ such that $c(s,v_1) \le c(s,v_2) \le \dots \le c(s,v_n)$. \label{step:sort_vertices}
\State Initialize $\bar T:=\emptyset$ and $Y:= \emptyset$.
\For{$v=v_n, v_{n-1},\ldots,v_1$}
\State Define $C_{v}:= \left\{u \in V : c(u, v) + \kappa \cdot \detour(u, v)  < \rho \cdot c(s,v)\right\}$. 
\If{$v \notin Y$ and $d(C_v \setminus Y) > 1 -\tau$}
\State Set $\bar T:= \bar T \cup \{v\}$. Set $Y:= Y \cup C_v$.
\EndIf
\EndFor
\State For $t \in \bar T$ define $B_t := \left\{ v\in V : c(v,t) < \frac{3\rho}{1-\rho}\cdot c(s,v),\  c(v,t) < 6\rho \cdot c(s,t) - \frac{3\rho}{1-\rho}\cdot c(s,v) \right\}$.
\State Compute the edge set $E_{B}$ that consists of all edges $\{t,t'\}$ with $t,t'\in \bar T$ for which
$
B_{t} \cap B_{t'} \neq \emptyset
$.
\State Let $\Tscr$ be the set of vertex sets of connected components of $(\bar T,E_{B})$.
\State For $\Tg   \in \Tscr$ define $B_{\Tg  }:= \bigcup\limits_{t \in \Tg  } B_{t}$.
\State Define $b:\Tscr\rightarrow 2 \cdot \mathbb{Z}_{>0}$ by 
\[
b(\Tg  ):= 2 \cdot  \left \lfloor \frac{d(B_{\Tg  })}{1-\tau} \right \rfloor \text{  } \forall\ \Tg   \in \Tscr.
\]
\State Output $(V,\bar T, s, c, \mathcal{T}, b)$.
\end{algorithmic}
\end{algorithm}

Note that $C(Q) \subseteq C_{\peak(Q)}$ for every tour $Q$ (cf.\ Definition~\ref{def:peak_cluster}). 
As mentioned above, the sets $B_t$ are chosen such that every large peak cluster is contained in one of these sets.
We will prove this below in Lemma~\ref{lem:guessing_large_clusters}.
Before, we show that the sets $B_t$ are chosen such that the following two inequalities hold (which we will need later).

\begin{lemma}\label{lemma:bound_distance_in_B_t}
	Let $t\in \bar T$ and $v\in B_t$. Then $c(v,t)< \frac{3\rho}{1-\rho} \cdot c(s,v)$. 
\end{lemma}

\begin{proof}
	Follows directly by definititon of $B_t$.
\end{proof}

\begin{lemma}\label{lemma:bound_matching_edges}
	Let $t,t'\in \bar T$ such that $B_t\cap B_{t'}\not=\emptyset$.
	Then $c(t,t')< 6\rho \cdot \min\{c(s,t),c(s,t')\}$.
\end{lemma}

\begin{proof}
	Let $w\in B_t \cap B_{t'}$.
	Then 
	\[
	c(t,t') \le c(w,t) + c(w,t') <  \left(6\rho \cdot c(s,t) - \frac{3\rho}{1-\rho}\cdot c(s,w)\right) +\frac{3\rho}{1-\rho}\cdot c(s,w) = 6\rho \cdot c(s,t).
	\]
\end{proof}

For the proof of Lemma~\ref{lem:guessing_large_clusters} we will need the following lemma.

\begin{lemma}\label{lem:possible_simplification}
Let $x,y\in V$ such that $C_x \cap C_y \ne \emptyset$.
Then \[c(x,y) < 2 \rho \cdot \min\{c(s,x),c(s,y)\}.\]
\end{lemma}
\begin{proof}
By symmetry it suffices to show $c(x,y) < 2 \rho \cdot c(s,x)$.
Let $w\in C_x \cap C_y$.

Because $w\in C_y$, we have $c(s,w) \ge c(s,y) - c(w,y) > c(s,y) - \rho \cdot c(s,y)$, implying
\begin{equation}\label{eq:bound_w_y}
c(w,y) < \rho \cdot c(s,y) < \frac{\rho}{1-\rho} \cdot c(s,w).
\end{equation}
Moreover, because $w\in C_x$, we have 
\[
 c(w,x) + \kappa \cdot (c(w,x)- c(s,x) + c(s,w)) = c(w,x) + \kappa \cdot \detour(w,x) < \rho \cdot c(s,x),
\]
implying
\begin{equation}\label{eq:bound_w_x}
c(w,x) < \frac{\kappa + \rho}{1+\kappa} \cdot c(s,x) - \frac{\kappa}{1+\kappa} \cdot c(s,w).
\end{equation}
Next we combine \eqref{eq:bound_w_x} and \eqref{eq:bound_w_y}.
Using $\frac{\rho}{1-\rho} \le \frac{\kappa}{1+\kappa}$ (which holds because $\rho\le\frac{1}{6}$ and $\kappa>\frac{3}{2}$)
as well as $c(s,w) \ge c(s,x) - c(w,x) > (1-\rho)\cdot c(s,x)$, we obtain
\begin{align*}
c(x,y) \le&\ c(w,x) + c(w,y) \\
 <&\ \frac{\kappa + \rho}{1+\kappa} \cdot c(s,x) - \frac{\kappa}{1+\kappa} \cdot c(s,w) + \frac{\rho}{1-\rho} \cdot c(s,w) \\
<&\ \frac{\kappa + \rho}{1+\kappa} \cdot c(s,x) + \left( - \frac{\kappa}{1+\kappa}+\frac{\rho}{1-\rho}\right)(1-\rho)\cdot c(s,x) \\
=&\ 2\rho \cdot c(s,x).
\end{align*}
\end{proof}

\begin{lemma}\label{lem:guessing_large_clusters}
	Let $\Ieuscr=(V,s,c,d)$ be an instance of \textsc{Capacitated Vehicle Routing}, and let $(V,\bar T, s, c, \Tscr, b)$ be the output of Algorithm~\ref{algo:constructing_vrpwtg_instance} with input $\Ieuscr$.
	Then for every tour $Q$ with large peak cluster there exists a target $t\in \bar T$ such that 
	$C(Q) \subseteq B_{t}$.
\end{lemma}

\begin{proof}
First, we show that there exists a vertex $t\in \bar T$ such that
\begin{enumerate}[(i)]\itemsep0pt
\item $C(Q) \cap C_t \ne \emptyset$, \label{item:clusters_intersect_C_sets} and
\item $c(s,\peak(Q)) \leq c(s,t)$. \label{item:peak_below_guessed_center}
\end{enumerate}
Consider the iteration of Algorithm~\ref{algo:constructing_vrpwtg_instance} where $v=\peak(Q)$.
If at this point of the algorithm $Y\cap C(Q) \ne \emptyset$, there is a vertex $t\in \bar T$ with $C(Q) \cap C_t \ne \emptyset$.
Moreover, we then have $c(s,\peak(Q)) \leq c(s,t)$ because $t$ was considered before $\peak(Q)$ in the for-loop of Algorithm~\ref{algo:constructing_vrpwtg_instance}.
Hence, it remains to consider the case where $Y\cap C(Q) = \emptyset$ in  the iteration of Algorithm~\ref{algo:constructing_vrpwtg_instance} where $v=\peak(Q)$.
Then $C(Q) \subseteq C_v \setminus Y$. Since $d(C(Q)) > 1-\tau$ by assumption, the vertex $v=\peak(Q)$ is added to $\bar T$ and
$C(Q) \subseteq C_v \subseteq B_v$.

Now we show that for this vertex $t \in \bar T$ we have $C(Q) \subseteq B_t$. By (i) we have $C(Q) \cap C_t \ne \emptyset$.
Note that $C(Q) \subseteq C_{\peak(Q)}$.
Hence, by Lemma~\ref{lem:possible_simplification}, we have $c(\peak(Q),t) < 2 \rho \cdot c(s,\peak(Q))$.

Let $v\in C(Q)$. Then we have $c(s,v) > (1-\rho) \cdot c(s,\peak(Q))$ and 
\[
 c(v,t) \le c(v,\peak(Q)) + c(\peak(Q),t) < \rho \cdot c(s,\peak(Q)) + 2 \rho \cdot c(s,\peak(Q)),
\]
implying $c(v,t) < \frac{3\rho}{1-\rho} c(s,v)$,
which is the first condition for membership of $v$ in $B_t$. We now show the second condition.

First,
\begin{align*}
\rho \cdot c(s,\peak(Q)) >&\ c(v,\peak(Q)) + \kappa \cdot \detour(v, \peak(Q)) \\
    =&\  (1+\kappa) \cdot c(v,\peak(Q)) + \kappa \cdot (c(s,v) - c(s,\peak(Q))),
\end{align*}
implying
\begin{align*}
c(v,\peak(Q)) <&\ \frac{\rho + \kappa}{1+\kappa} \cdot c(s,\peak(Q)) - \frac{\kappa}{1+\kappa} \cdot c(s,v) \\
=&\ \rho \cdot c(s,\peak(Q)) + \frac{\kappa \cdot (1-\rho)}{1+\kappa} \cdot c(s,\peak(Q)) - \frac{\kappa}{1+\kappa} \cdot c(s,v) \\
=&\ \rho \cdot c(s,\peak(Q)) - \frac{\kappa \cdot (1-\rho)}{1+\kappa} \cdot \left(\frac{c(s,v)}{1-\rho} - c(s,\peak(Q))  \right) \\
\le&\  \rho \cdot c(s,\peak(Q)) - 3 \rho \cdot \left(\frac{c(s,v)}{1-\rho} - c(s,\peak(Q))  \right)
\end{align*}
because $\frac{\kappa \cdot (1-\rho)}{1+\kappa} \ge 3 \rho$ (which holds since $\kappa>\frac{3}{2}$ and $\rho\le\frac{1}{6}$) and
$c(s,v) > (1-\rho)\cdot c(s,\peak(Q))$.
Recall that  $c(\peak(Q),t) < 2 \rho \cdot c(s,\peak(Q))$. This yields
\begin{align*}
c(v,t) \le&\ c( \peak(Q),t) + c(v, \peak(Q)) \\
    <&\ 2 \rho \cdot c(s,\peak(Q)) + \rho \cdot c(s,\peak(Q)) - 3 \rho \cdot \left(\frac{c(s,v)}{1-\rho} - c(s,\peak(Q))  \right) \\
    \le&\ 6\rho \cdot c(s,t) - \frac{3\rho}{1-\rho} \cdot c(s,v).
\end{align*}
because $c(s,\peak(Q)) \le c(s,t)$ by \eqref{item:peak_below_guessed_center}.
\end{proof}

Figure~\ref{fig:example_best_possible} shows that the estimates in Lemma~\ref{lemma:bound_distance_in_B_t} and Lemma~\ref{lemma:bound_matching_edges} are best possible 
for any choice of the sets $B_t$ such that $C_t \cap C(Q) \neq \emptyset$ implies $ C(Q) \subseteq B_t$ for each target $t$ and each tour $Q$, which we use to prove Lemma~\ref{lem:guessing_large_clusters}.

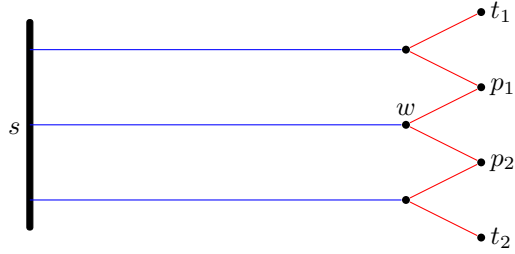
\begin{figure}[ht]
	\begin{center}
	\begin{tikzpicture}
	\begin{scope}[scale=2]

	\tikzstyle{vertex}=[black,circle,fill,minimum size=3,inner sep=0pt]
	\tikzstyle{vertex2}=[draw=blue,circle,minimum size=3,inner sep=0pt]	
	\node[vertex] (1) at (3,0.75) {};	
	\node[vertex] (2) at (3,0.25) {};
	\node[vertex] (3) at (3,-0.25) {};
	\node[vertex] (4) at (3,-0.75) {};
	\node[vertex] (5) at (2.5,0.5) {};
	\node[vertex] (6) at (2.5,0) {};
	\node[vertex] (7) at (2.5,-0.5) {};
	\draw[black,fill, rounded corners=1pt] (-0.02,-0.7) rectangle (0.02,0.7);
	\node[black, left] at (0,-0.02) {$s$};
	
	\node[black, right] at (3,0.75) {$t_1$};
	\node[black, right] at (3,0.25) {$p_1$};
	\node[black, right] at (3,-0.25) {$p_2$};
	\node[black, right] at (3,-0.75) {$t_2$};
	\node[black, above] at (2.5,0) {$w$};
	
	\draw[draw=blue] (0,0.5)--(5);
	\draw[draw=blue] (0,0)--(6);
	\draw[draw=blue] (0,-0.5)--(7);
	\draw[draw=red] (1)--(5)--(2)--(6)--(3)--(7)--(4);

	\end{scope}
	\end{tikzpicture}
	\caption{Example showing that the estimates in Section~\ref{sect:guessing_clusters} are tight.
	Blue edges have length $(1- \rho)$ and red edges have length slightly smaller than $\rho$. 
	The other distances are given by shortest paths in the drawn graph. 
Let $t_1, t_2$ be targets and $p_1,p_2$ be peaks of tours with large peak cluster. Assume that the sets $B_t$ are chosen such that $C_t \cap C(Q) \neq \emptyset$ implies $ C(Q) \subseteq B_t$ for each target $t$ and each tour $Q$.  Then $w \in B_{t_1} \cap B_{t_2}$. Hence, the estimates in Lemma~\ref{lemma:bound_distance_in_B_t} and Lemma~\ref{lemma:bound_matching_edges} are best possible since $c(w,t_1)$ is close to $\sfrac{3 \rho}{1 - \rho} \cdot c(s,w)$ and $c(t_1,t_2)$ is close to $6 \rho \cdot c(s,t_1)= 6\rho \cdot c(s,t_2)$.}\label{fig:example_best_possible}
\end{center}
\end{figure}	

\subsection{Matching paths}\label{sect:matching}

We will turn a solution $\Pscr$ to the \textsc{Vehicle Routing with Target Groups} instance
into a traveling salesman tour by connecting the endpoints of two paths in $\Pscr$ by a matching edge.
Next, we bound the cost of that matching.

We need the following well-known lemma.
\begin{lemma}\label{lem:bound_matching_by_tree}
Let $V$ be a finite set of vertices and let $c :{V \choose 2} \to  \mathbb{R}_{\ge 0}$ be a semi-metric.
Let $(V,S)$ be a spanning tree.
Then for any set $U\subseteq V$ of even cardinality, there is a perfect matching on $U$ with cost at most $c(S)$.
\end{lemma}
\begin{proof} 
We use induction on $|V|$, the case $|V|\le 1$ being trivial.
Let now $|V|\ge 2$, and let $v$ be a leaf and $e=\{v,w\}\in S$ the edge incident to $v$.
Let $(V',S')$ arise from $(V,S)$ by deleting $v$ and $e$.
If $v\notin U$, applying the induction hypothesis to $(V',S',U)$ does the job.
If $v,w\in U$, applying the induction hypothesis to $(V',S',U \setminus\{w\})$
yields a perfect matching of cost at most $c(S')$, and we add $e$ to this matching.
Finally, if $v\in U$ and $w\notin U$, applying the induction hypothesis to $(V',S',U \cup\{w\})$
yields a perfect matching on $U\cup\{w\}$ of cost at most $c(S')$.
If $\{w,x\}$ is the matching edge incident to $w$, we replace it by $\{v,x\}$, which 
increases the cost by at most $c(e)$ due to the triangle inequality.
\end{proof}

We use this to obtain the lemma below, which we will then apply for $U$ being the set $\bar T_{\mathrm{odd}}$ of targets where an odd number of paths ends in our solution $\Pscr$ of \textsc{Vehicle Routing with Target Groups}.
\begin{lemma}\label{lem:bound_matching_cost}
Let $\Ieuscr=(V,s,c,d)$ be an instance of \textsc{Capacitated Vehicle Routing} and let $(V,\bar T,s,c,\Tscr,b)$ be the instance of \textsc{Vehicle Routing with Target Groups} computed by Algorithm~\ref{algo:constructing_vrpwtg_instance} applied to $\Ieuscr$.
Let $U \subseteq \bar T$ such that $|U \cap \Tg  |$ is even for every target group $\Tg  \in \Tscr$.
Then there is a perfect matching on $U$ with cost at most 
\[
  \frac{3 \rho}{(1-\rho) \cdot (1-\tau)}\cdot \OPT(\Ieuscr).
\]
\end{lemma}
\begin{proof}
We consider a fixed target group $\Tg  $. 
By Lemma~\ref{lem:bound_matching_by_tree} we can bound the cost of a minimum-cost perfect matching on $\Tg   \cap U$ by the cost $c(S)$ of any tree $(\Tg  , S)$.
By Lemma~\ref{lemma:bound_matching_edges} we have 
$c(t,t') < 6\rho \cdot \min\{c(s,t), c(s,t')\}$ for every edge $\{t,t'\}\in E_{B}$ (defined in Algorithm~\ref{algo:constructing_vrpwtg_instance}).
Because $\Tg  $ is the vertex set of a connected component of $(\bar T, E_{B})$, 
there is a tree $(\Tg  ,S)$ with $S\subseteq E_{B}$.
In a rooted orientation of this tree every vertex $t$ has at most one entering edge,
which costs at most $6 \rho \cdot c(s,t)$.
Hence $c(S)\le 6 \rho \cdot \sum_{t\in \Tg  } c(s,t)$.

Taking the union of the minimum-cost perfect matchings on $\Tg   \cap U$ for all target groups $\Tg  \in \Tscr$ yields a perfect matching $M$ on $U$ with cost at most $6\rho \cdot  \sum_{t\in \bar T} c(s,t)$.
We use $Y_t$ to denote the set $Y$ at the beginning of the iteration of Algorithm~\ref{algo:constructing_vrpwtg_instance} in which we add $t$ to $\bar T$.
Then the sets $C_t \setminus Y_t$ for $t\in \bar T$ in Algorithm~\ref{algo:constructing_vrpwtg_instance} are disjoint, and hence 
\begin{align*}
(1-\tau) \cdot (1-\rho) \cdot \sum_{t\in \bar T} c(s,t) <&\ (1-\rho) \cdot \sum_{t\in \bar T} c(s,t) \cdot d(C_t \setminus Y_t) \\
<&\   \sum_{t\in \bar T}\sum_{v\in C_t \setminus Y_t} c(s,v)  \cdot d(v) \\
\le&\ \sum_{v\in V} c(s,v)  \cdot d(v) \\
\le&\ \sfrac{1}{2} \OPT(\Ieuscr).
\end{align*}
Thus, we can bound the cost of the matching $M$ by $\frac{3 \rho}{(1-\rho) \cdot (1-\tau)} \cdot \OPT(\Ieuscr)$.
\end{proof}

\begin{theorem}\label{thm:good_tour_difficult_instances}
Let $\Ieuscr=(V,s,c,d)$ be a difficult instance of \textsc{Capacitated Vehicle Routing},
$\Jeuscr$ the instance of \textsc{Vehicle Routing with Target Groups} constructed by Algorithm~\ref{algo:constructing_vrpwtg_instance},
and $\Pscr$ a solution to $\Jeuscr$.
Then we can compute a traveling salesman tour $Q$  such that
\[c(Q) \ \le \ c(\Pscr)+    \frac{3 \rho}{(1-\rho) \cdot (1-\tau)}\cdot \OPT(\Ieuscr)\] in $O(n^3)$ time.
\end{theorem}
\begin{proof}

We will add a matching $M$ to $\bigcupp_{P\in \Pscr} E(P)$ to obtain an Eulerian multi-edge set $H= M\cupp\   \bigcupp_{P\in \Pscr} E(P)$. 
Then $(V,H)$ is connected and Eulerian and therefore it has an Eulerian walk, i.e.,
a closed walk that visits every element of $V\cup \{s\}$ at least once.
After shortcutting we obtain a traveling salesman tour of cost at most $c(H)$.

We choose $M$ to be a minimum-cost perfect matching on the set $\bar T_{\text{odd}}$ of targets $t\in \bar T$ in which an odd number of paths in $\Pscr$ ends.
To bound the cost of this matching, we observe that $b(\Tg  )$ is even for every target group $\Tg  $ and hence an even number of paths in $\Pscr$ ends in $\Tg  $.
Therefore, every target group contains an even number of elements of $\bar T_{\text{odd}}$.
By Lemma~\ref{lem:bound_matching_cost}, the cost of $M$ is at most
\[
   \frac{3 \rho}{(1-\rho) \cdot (1-\tau)}\cdot \OPT(\Ieuscr).
\]
\end{proof}

\subsection{Existence of a good fractional solution for vehicle routing with target groups}
\label{sect:good_solution_exists}

We now prove that the instance $\Jeuscr$ constructed by Algorithm~\ref{algo:constructing_vrpwtg_instance} with input $\Ieuscr$
has a weak fractional solution $x$ that is not much more expensive than $\OPT(\Ieuscr)$ and has small detour.
Our main objective will be to minimize the detour, not counting the detour of the final edges of the tours 
(entering $\{s\}\cup \bar T$), because we will apply Theorem~\ref{thm:mainVRPwTargetGroups} to this $x$. 

\begin{lemma}\label{lem:good_sol_exists}
Let $\Ieuscr=(V,s,c,d)$ be a difficult instance of \textsc{Capacitated Vehicle Routing}, 
and let $\Jeuscr=(V,\bar T,s,c,\Tscr,b)$ be the instance of  \textsc{Vehicle Routing with Target Groups} constructed by Algorithm~\ref{algo:constructing_vrpwtg_instance} applied to $\Ieuscr$.
Then there is a weak fractional solution $x$ to $\Jeuscr$ such that
\begin{align}
\detour(x |_{E_1(\Jeuscr)})&\ \le\ \epsilon \cdot \OPT(\Ieuscr), \label{eq:bound_detour}
\intertext{and}
c(x |_{E_1(\Jeuscr)})&\ \le \ \OPT(\Ieuscr), \label{eq:bound_length_before_t}
\intertext{and}
c(x) &\ <\ \left(1 + \zeta\right) \cdot \OPT(\Ieuscr), \label{eq:bound_length_all}
\intertext{where}
\zeta &\ = \  \frac{3\rho+\tau-4\tau\rho}{1-\rho} 
   + \frac{\epsilon}{\tau \cdot \rho}\cdot \left(1-\tau\cdot \rho - \frac{3\rho+\tau-4\tau\rho}{1-\rho}\right).
\label{eq:def_zeta}   
\end{align}
\end{lemma}
\begin{proof}

We fix an optimum solution $\Qscr$ to $\Ieuscr$ and construct a weak fractional solution $x$ to $\Jeuscr$ as follows.
During the course of the construction, we may modify the instance multiple times 
by ``splitting'' a customer $v$, replacing it by two copies $v'$ and $v''$ with $d(v')+d(v'')=d(v)$.
Of course, every weak fractional solution to the resulting instance $\Jeuscr'$ induces a weak fractional solution
to $\Jeuscr$ with the same cost and detour.

\begin{enumerate}[(a)]
\item First, we construct a subset $C^*$ of the set $C:=\bigcup_{\Tg  \in\Tscr} B_{\Tg  }$ of clustered points. This will have the property that 
for every target group $\Tg  \in \Tscr$, we have 
\[
(1-\tau) \cdot b(\Tg  )=2\cdot d(B_{\Tg  } \cap C^*).
\]
By Lemma~\ref{lem:guessing_large_clusters}, for each tour $Q$ with large peak cluster there is a target group $\Tg  $ such that $d(C(Q)\cap B_{\Tg  }) \ge 1-\tau$. We will choose $C^*\subseteq C$ such that we maintain this property when restricting to $C^*$, i.e.\ we will have $d(C(Q)\cap B_{\Tg  }\cap C^*) \ge 1-\tau$.

We now define $C^*$. To this end, we define $C^* \cap B_{\Tg  }$ for each target group $\Tg  \in \Tscr$.
This defines $C^*\subseteq C$ because $\{B_{\Tg  } : \Tg  \in\Tscr\}$ is a partition of $C$.
For each target group $\Tg  \in\Tscr$, we consider every tour $Q\in \Qscr$ with $d(C(Q)\cap B_{\Tg  }) \ge 1-\tau$.
For each such tour we choose a subset of $C(Q)\cap B_{\Tg  }$ of total demand exactly $1-\tau$ and include it in  $C^* \cap B_{\Tg  }$. 
(This can be achieved after splitting one customer.) 
Then we include arbitrary additional vertices from $B_{\Tg  }$ in $B_{\Tg  }\cap C^*$ such that $(1-\tau) \cdot b(\Tg  )=2\cdot d(B_{\Tg  } \cap C^*)$. 
(Again, this can be achieved after possibly splitting one customer.)

\item \label{item:satisfy_target_demands}
Now we construct $x$, starting with the all-zero vector.
For each customer $v \in V$, we partition the tour $Q\in \Qscr$  containing $v$ into two $s$-$v$-paths, which we orient away from $s$.
Then we add $d(v)$ times the incidence vectors of these oriented paths to $x$. 

\item \label{item:final_arcs}
Moreover, for each $v\in V$, we add  $2\cdot d(v)$ times the incidence vector of an edge $(v,t)$, where $t\in \{s\} \cup \bar T$.
If $v\in C^*$ we choose $t$ to be a target with $v\in B_t$.
Otherwise, we choose $t=s$.
Thus, overall for every vertex $v\in V$, we add $d(v)$ times the incidence vector of two $s$-$t$-paths. 
After this step we have $x(\delta^-(\Tg  ))= 2\cdot d(B_{\Tg  } \cap C^*) = (1-\tau) \cdot b(\Tg  )$ for every target group $\Tg  $.

\item \label{item:visit_all}
Finally, for each tour $Q\in \Qscr$, we add $1-d(V(Q)\setminus\{s\})$ times the incidence vector of any orientation of $Q$ (to a directed cycle) to $x$.
\end{enumerate}

The resulting vector $x$ is a weak fractional solution to the instance $\Jeuscr$. 
Indeed, we added fractions of walks from $s$ to $\{s\}\cup \bar T$ in \eqref{item:satisfy_target_demands} and \eqref{item:final_arcs}
and fractions of walks from $s$ to $s$ in \eqref{item:visit_all}.
Moreover, every vertex $v$ of a tour $Q$ is visited by walks of total value $d(V(Q)\setminus\{s\})+d(v)$ 
in \eqref{item:satisfy_target_demands} and \eqref{item:final_arcs}
and by a walk of value $1-d(V(Q)\setminus\{s\})$ in \eqref{item:visit_all}, so at least 1 overall. 
Finally, the total value of the walks ending in each target group $\Tg  \in\Tscr$ is exactly $(1-\tau) \cdot b(\Tg  )$, as required.

Let $x^{\eqref{item:final_arcs}}$ denote the contribution of the edges $(v,t)$ added in step \eqref{item:final_arcs} to $x$.
Note that by construction $c(x-x^{\eqref{item:final_arcs}})=c(\Qscr)$, which implies \eqref{eq:bound_length_before_t}.

Next, for all $v\in V$, the total value of walks ending with an edge $(v,t)$ added in \eqref{item:final_arcs} is exactly $2d(v)$.  
Thus 
\[\detour(x-x^{\eqref{item:final_arcs}}) = c(x-x^{\eqref{item:final_arcs}}) - \sum_{v\in V} 2d(v)c(s,v) = c(\Qscr)  - \sum_{v\in V} 2d(v)c(s,v),\] 
implying \eqref{eq:bound_detour} because $\Ieuscr$ is difficult.

To show \eqref{eq:bound_length_all}, we finally bound the total cost of $x^{\eqref{item:final_arcs}}$.
When we add a contribution of an edge $(v,t)$ in \eqref{item:final_arcs}, we have $v\in C^*\cap B_t$ and
hence $c(v,t)< \frac{3\rho}{1-\rho} \cdot c(s,v)$
by Lemma~\ref{lemma:bound_distance_in_B_t},
or $v\in V\setminus C^*$ and $t=s$.

For a tour $Q$ with large peak cluster, the contribution of edges $(v,t)$ that we add in \eqref{item:final_arcs} for $v\in V(Q)$ is at most
\begin{align*}
\sum_{v\in V(Q)\cap C^*} 2 \cdot d(v)\cdot \sfrac{3\rho}{1-\rho} \cdot c(s,v) 
       + \sum_{v\in V(Q)\setminus (\{s\}\cup C^*)} 2 \cdot d(v) \cdot c(s,v) 
&\ \le \ \left( \sfrac{3\rho}{1-\rho} (1-\tau) + \tau \right) \cdot c(Q) \\     
&\ = \ \sfrac{3\rho+\tau-4\rho\tau}{1-\rho} \cdot c(Q)      
\end{align*}       
because $\sfrac{3\rho}{1-\rho}< 1$ and $d(V(Q)\cap C^*)\ge 1-\tau$.

For a tour $Q$ with small peak cluster, the contribution of edges $(v,t)$ that we add in \eqref{item:final_arcs} for $v\in V(Q)$ is at most
\begin{align*}
\sum_{v\in V(Q)\cap C^*} \sfrac{3\rho}{1-\rho} \cdot 2 \cdot d(v)\cdot c(s,v) 
       + \sum_{v\in V(Q)\setminus (\{s\}\cup C^*)} 2 \cdot d(v) \cdot c(s,v) 
&\ \le \  \sum_{v\in V(Q)\setminus \{s\}} 2 \cdot d(v) \cdot c(s,v)  \\
&\ \le \ (1-\tau\cdot\rho) \cdot c(Q)     
\end{align*}       
by Lemma~\ref{lem:difficult_instances_clustered}.
Summing over all tours $Q\in\Qscr$, we get
\begin{align*}
c(x^{\eqref{item:final_arcs}}) 
\ \le \ \sfrac{3\rho+\tau-4\rho\tau}{1-\rho} \cdot c(\Qscr) 
+  \left(1-\tau\cdot \rho - \sfrac{3\rho+\tau-4\tau\rho}{1-\rho} \right) \cdot \sum_{Q\in\Qscr: d(C(Q))\le 1-\tau}  c(Q). 
\end{align*} 
Again using Lemma~\ref{lem:difficult_instances_clustered}, we can bound this by
\begin{align*}
c(x^{\eqref{item:final_arcs}}) 
&\le \ \sfrac{3\rho+\tau-4\rho\tau}{1-\rho} \cdot c(\Qscr) 
+   \left(1-\tau\cdot \rho - \sfrac{3\rho+\tau-4\tau\rho}{1-\rho} \right) \cdot \!\!
\sum_{Q\in\Qscr: d(C(Q))\le 1-\tau} \sfrac{1}{\tau\cdot\rho} \left( c(Q) - \!\sum_{v\in V(Q)\setminus\{s\}} 2 \cdot d(v) \cdot c(s,v)\right) \\
&\le \ \sfrac{3\rho+\tau-4\rho\tau}{1-\rho} \cdot c(\Qscr) 
+  \frac{ \left(1-\tau\cdot \rho - \sfrac{3\rho+\tau-4\tau\rho}{1-\rho} \right) }{\tau\cdot\rho} \cdot 
\sum_{Q\in\Qscr} \left( c(Q) - \!\sum_{v\in V(Q)\setminus\{s\}} 2 \cdot d(v) \cdot c(s,v)\right) \\
&\le \  \left( \frac{3\rho+\tau-4\tau\rho}{1-\rho} 
   + \epsilon\cdot \frac{1-\tau\cdot \rho - \sfrac{3\rho+\tau-4\tau\rho}{1-\rho}}{\tau\cdot\rho} \right) \cdot c(\Qscr)
\end{align*}  
because $\Ieuscr$ is difficult.
Together with \eqref{eq:bound_length_before_t} this implies \eqref{eq:bound_length_all}.
\end{proof}

\subsection{Completing the proof}\label{sect:solving_difficult}

In this section we derive the main result of this chapter from Theorem~\ref{thm:mainVRPwTargetGroups}.

\begin{theorem}\label{thm:CVR_3plus}
There is a function $f:\mathbb{R}_{>0}\to\mathbb{R}_{>0}$ with
$\lim_{\epsilon\to 0}f(\epsilon)=0$ and a
polynomial-time algorithm for \textsc{Capacitated Vehicle Routing} that returns a solution of cost at most
\begin{equation*}
\left( 3+ f(\epsilon) \right)\cdot \OPT(\Ieuscr)
\end{equation*}
for any given difficult instance $\Ieuscr$.
\end{theorem}
\begin{proof}
Let $\epsilon\le 6^{-3}$.
We set $\tau = \rho = \sqrt[3]{\epsilon}$ and $\eta = \sqrt{\epsilon}$.
We first apply Algorithm~\ref{algo:constructing_vrpwtg_instance} to obtain an instance $\Jeuscr=(V,\bar T,s,c,\Tscr,b)$ of \textsc{Vehicle Routing with Target Groups}. 
Then we apply Theorem~\ref{thm:mainVRPwTargetGroups} to compute a solution $\Pscr$ for this instance $\Jeuscr$.
Lemma~\ref{lem:good_sol_exists} and Theorem~\ref{thm:mainVRPwTargetGroups} imply
\begin{align*}
 c(\Pscr) \le&\  (\sfrac{1}{1-\tau}+\eta)\cdot c(x) 
 + O(\sfrac{1}{\eta}) \cdot \detour(x|_{E_1(\Jeuscr)}) \\
 \le&\ (\sfrac{1}{1-\tau}+\eta)\cdot (1 + \zeta) \cdot \OPT(\Ieuscr) +
 O(\sfrac{1}{\eta}) \cdot \epsilon \cdot \OPT(\Ieuscr),
\end{align*}
where $\zeta = \frac{3\rho+\tau-4\tau\rho}{1-\rho} 
   + \frac{\epsilon}{\tau\cdot\rho} \cdot (1-\tau\cdot \rho - \sfrac{3\rho+\tau-4\tau\rho}{1-\rho})$.
By the choice of $\tau$ and $\rho$, we have $ \lim_{\epsilon\to 0} \zeta = 0$.
Using also the choice of $\eta$, we get $\lim_{\epsilon\to 0}  c(\Pscr) = \OPT(\Ieuscr)$.
Next, we apply Theorem~\ref{thm:good_tour_difficult_instances} and get a traveling salesman tour $Q$ with cost at most
\begin{align*}
c(Q) \le&\ c(\Pscr) + \sfrac{3 \rho}{(1-\rho) \cdot (1-\tau)} \cdot \OPT(\Ieuscr).
\end{align*}
Because $\lim_{\epsilon\to 0}  \sfrac{3 \rho}{(1-\rho) \cdot (1-\tau)} = 0$, we get that 
$c(Q) \le (1+f(\epsilon)) \cdot \OPT(\Ieuscr)$ for some function $f:\mathbb{R}_{>0}\to\mathbb{R}_{>0}$ with
$\lim_{\epsilon\to 0}f(\epsilon)=0$.
Finally we apply Theorem~\ref{thm:tour_splitting} to this traveling salesman tour and obtain a solution $\Qscr$ to $\Ieuscr$ of cost at most
\begin{align*}
c(\Qscr) \le&\ c(Q) + \sum\limits_{v\in V}4d(v)c(s,v) \
\le\   (1+f(\epsilon)) \cdot \OPT(\Ieuscr) + 2 \cdot \OPT(\Ieuscr),
\end{align*}
where we used Proposition~\ref{prop:trivial_lower_bound} in the second inequality. 
\end{proof}

This now leads to a better approximation ratio than $\alpha + 2$. 
Call the algorithm by Altinkemer and Gavish (Section~\ref{sect:classical_algo})
and our new algorithm and return the cheaper of the two solutions. 
If the given instance is not difficult, the algorithm by Altinkemer and Gavish returns a solution of cost at most 
$\alpha + 2 \cdot (1-\epsilon)$ times the optimum. 
If the given instance is difficult, our new algorithm returns a solution of cost at most $3 + f(\epsilon)$ times the optimum. 
Choose $\epsilon>0$ such that $3 + f(\epsilon) \le \alpha + 2 \cdot (1-\epsilon)$. 
Note that the best choice of $\epsilon$ depends on $\alpha$. 
We will compute the constants in Section~\ref{sect:final_result}.

More importantly, we still have to prove Theorem~\ref{thm:mainVRPwTargetGroups}.

\section{Combinatorial algorithm for vehicle routing with target groups \label{sec:combinatorial}}

In this section we prove Theorem~\ref{thm:mainVRPwTargetGroups}, which we restate here in a more specific form:

\begin{theorem}\label{thm:mainVRPwTargetGroupsCombinatorial}
There is an algorithm for \textsc{Vehicle Routing with Target Groups} that 
runs in $O(n^3)$ time, where $n:=|V|+|\bar T|$, and
computes for every instance $\Jeuscr=(V,\bar T,s,c,\mathcal{T}, b)$ and any given $\gamma>2$
a feasible solution $\Pscr$ of $\Jeuscr$ such that 
\[
c(\Pscr) \ < \ 
\sfrac{1}{1-\tau} \cdot c(x) 
+ \sfrac{8}{\gamma-2} \cdot c(x|_{E_1(\Jeuscr)}) 
+ \bigl( 4+\sfrac{8}{\gamma-2}+2\gamma - \sfrac{\tau}{1-\tau} \bigr) \cdot \detour(x|_{E_1(\Jeuscr)}),
\]
for every weak fractional solution $x$ of $\Jeuscr$. 
\end{theorem}

The first observation is that, informally, if a weak fractional solution $x$ has small total detour, 
then for almost every arc $(u,v)$ in almost every walk that is part of $x$:
\begin{enumerate}[(i)]
\item $u$ is closer to $s$ than $v$, or
\item the distance from $u$ to $v$ is small.
\end{enumerate}

Of course, we do not know $x$ and hence we cannot classify the vertices accordingly.
Nevertheless we do something similar. 
In a walk from $s$ to a target, a customer may have more than one entering arc, 
and we will require that for every customer $v$ there is at least one entering arc $(u,v)$ that fulfills (i).
Let us call this a \emph{forward walk solution} (see Subsection~\ref{subsection:forward_walks} for precise definitions). 
A cheap forward walk solution does not always exist, but we will give a sufficient condition (which we call \emph{nice}): 
roughly, for every $v\in V$, 
there is a vertex closer to $s$ that is not much farther from $v$ than any other vertex.

In Subsection~\ref{subsection:compute_nice_subinstance} we will compute a subset of $V$ 
that induces a nice subinstance and hence is spanned by a cheap forward walk solution.
Finding a cheapest forward walk solution is a simple network flow problem (cf.\ Subsection~\ref{subsection:forward_walks}).
Afterwards, we can insert the other vertices into the resulting walks at a small cost; 
this is shown in Subsection~\ref{section:short_forest}. 
In the end, we shortcut the walks to paths, using Proposition~\ref{prop:walk_solution_suffices}.

\subsection{Cheapest forward walk solutions \label{subsection:forward_walks}}

Given an instance $\Jeuscr=(V,\bar T,s,c,\mathcal{T}, b)$ of \textsc{Vehicle Routing with Target Groups},
let $\prec$ be a total order on $\{s\}\cup V$ such that $s\prec v$ for all $v\in V$ and
$v\prec w$ whenever $c(s,v)<c(s,w)$. We call such an order a \emph{depot distance order}.
Now consider the digraph $G(\Jeuscr)$ associated with this instance (cf.\ Section~\ref{sec:overview}).
We call an arc $(v,w)\in E_1(\Jeuscr)$ a \emph{forward arc} if $v\prec w$, otherwise a \emph{backward arc}.
A \emph{forward walk solution} is a walk solution $H$ that 
contains a forward arc entering $v$ for every $v\in V$. 
Note that this condition implies that every $v\in V$ can be reached from $s$ 
and thus $H$ connects all elements of $V$.

We have seen that computing a cheapest walk solution is 
equivalent to solving \textsc{Vehicle Routing with Target Groups} and hence APX-hard.
In contrast, finding a cheapest \emph{forward} walk solution is easy:

\begin{lemma}\label{lemma:find_cheap_forward_walk}
Given an instance of \textsc{Vehicle Routing with Target Groups} and a depot distance order, 
one can compute a cheapest forward walk solution $H$ in $O(n^3)$ time, 
where $n=|V|+|\bar T|$. Moreover, $|H|\le n^2$.
\end{lemma}

\begin{proof}
This can be formulated as a minimum-cost flow problem.
From the digraph $G(\Jeuscr)$ associated with the instance $\Jeuscr$ we construct another digraph $G'$.
To this end, we replace the depot vertex $s$ by $s^+$,
replace every customer vertex $v\in V$ by two vertices $v^-$ and $v^+$, joined by a new arc $(v^-,v^+)$ of cost zero, 
replace every forward arc $(v,w)$ by $(v^+,w^-)$, every backward arc $(v,w)$ by $(v^+,w^+)$,
every arc $(v,t)$ for $t\in \bar T$ by $(v^+,t)$, 
and every arc $(v,s)$ by $(v^+,s^+)$. 
For each $\Tg  \in\Tscr$ we add a vertex $\Tg  $ and arcs $(t,\Tg  )$ with cost zero for all $t\in \Tg  $.

We are looking for a flow that ships $\sum_{\Tg  \in\Tscr}b(\Tg  )$ units out of $s^+$, 
$b(\Tg  )$ units into $\Tg  $ for each $\Tg  \in\Tscr$, and at least one unit of flow along each new arc $(v^-,v^+)$.
We call such a flow \emph{feasible}. 

Every integral feasible flow $f$ corresponds to a forward walk solution of the same cost
(by taking $f(e')$ copies of each $e\in E$, where $e'$ is the corresponding arc in $G'$),
and vice versa.
It is well-known that there exists a minimum-cost feasible flow that is integral. 

Moreover, such a flow can
be found in $O(n^3)$ time by the successive shortest paths algorithm as follows. 
First, for each target group $\Tg  \in\Tscr$, ship $b(\Tg  )$ units from $s^+$ via $t$ to $\Tg  $, where $t\in \Tg  $
is a target in that group for which $c(s,t)$ is minimum.
Then, successively for each $v\in V$, find a shortest augmenting path from $v^+$ to $v^-$ in the residual network.
Each of these $|V|$ iterations can be performed by Dijkstra's algorithm and thus 
takes $O(E(G'))=O(|V|\cdot (|V|+|\bar T|))$ time if we keep updating a feasible potential in the residual network
(\cite{EdmondsKarp,Tomizawa}).
At the end, add one unit of flow from $v^-$ to $v^+$ for all $v\in V$.

Since the flow results from augmenting along at most $n$ paths, each of length at most $n$, 
we get $|H|\le n^2$ for the resulting forward walk solution $H$.
\end{proof}

Unfortunately, a cheap forward walk solution does not always exist.
For $w\in V$ and $W\subseteq V$ we denote by 
\begin{equation*}
\parent_W(w) \ := \ \arg\min\{c(p,w) : p\in \{s\}\cup W,\, p\prec w\}
\end{equation*}
a vertex closest to $w$ among all vertices in $\{s\}\cup W$ that precede $w$ in the depot distance order.
Then the cheapest forward walk solution has cost at least $\sum_{v\in V}c(\parent_V(v),v)$.
Figure~\ref{figure:no_cheap_forward_walk_solution} shows that this can be a factor $\Omega(\log n)$
larger than a cheapest solution to the instance of \textsc{Vehicle Routing with Target Groups} (even with a single target $t$).
However, this is not the case if the instance is \emph{nice}, in the following sense:

\begin{figure}[ht]
\begin{center}
\begin{tikzpicture}[xscale=0.9,yscale=0.9, rotate=270]

\tikzstyle{vertex}=[blue,circle,fill,minimum size=3,inner sep=0pt]
\node[vertex] (v2) at (4,-0.1) {};
\node[vertex] (v3) at (2,0.0) {};
\node[vertex] (v4) at (6,0.0) {};
\node[vertex] (v5) at (1,0.1) {};
\node[vertex] (v6) at (3,0.1) {};
\node[vertex] (v7) at (5,0.1) {};
\node[vertex] (v8) at (7,0.1) {};
\node[vertex] (v9) at (0.5,0.2) {};
\node[vertex] (v10) at (1.5,0.2) {};
\node[vertex] (v11) at (2.5,0.2) {};
\node[vertex] (v12) at (3.5,0.2) {};
\node[vertex] (v13) at (4.5,0.2) {};
\node[vertex] (v14) at (5.5,0.2) {};
\node[vertex] (v15) at (6.5,0.2) {};
\node[vertex] (v16) at (7.5,0.2) {};
\node[vertex] (t) at (8,0.2) {};

\draw[blue,fill, rounded corners=1pt] (0,-4.04) rectangle (8,-3.96);

\draw[draw=darkgreen, thick] (0.5,-4)--(v9)--(v5)--(v10)--(v3)--(v11)--(v6)--(v12)--(v2)--(v13)--(v7)--(v14)--(v4)--(v15)--(v8)--(v16)--(t);
 
\node[right] at (3,-4) {\textcolor{blue}{$s$}};
\node[left] at (t) {\textcolor{blue}{$t$}};

\end{tikzpicture}
\hspace*{5cm}
\begin{tikzpicture}[xscale=0.9,yscale=0.9, rotate=270]

\tikzstyle{vertex}=[blue,circle,fill,minimum size=3,inner sep=0pt]
\node[vertex] (v2) at (4,-0.1) {};
\node[vertex] (v3) at (2,0.0) {};
\node[vertex] (v4) at (6,0.0) {};
\node[vertex] (v5) at (1,0.1) {};
\node[vertex] (v6) at (3,0.1) {};
\node[vertex] (v7) at (5,0.1) {};
\node[vertex] (v8) at (7,0.1) {};
\node[vertex] (v9) at (0.5,0.2) {};
\node[vertex] (v10) at (1.5,0.2) {};
\node[vertex] (v11) at (2.5,0.2) {};
\node[vertex] (v12) at (3.5,0.2) {};
\node[vertex] (v13) at (4.5,0.2) {};
\node[vertex] (v14) at (5.5,0.2) {};
\node[vertex] (v15) at (6.5,0.2) {};
\node[vertex] (v16) at (7.5,0.2) {};
\node[vertex] (t) at (8,0.2) {};

\draw[blue,fill, rounded corners=1pt] (0,-4.04) rectangle (8,-3.96);

\tikzstyle{edge}=[red,draw, thick]
\draw[edge] (4,-4)--(v2);
\draw[edge] (v2)--(v3);
\draw[edge] (v2)--(v4);
\draw[edge] (v3)--(v5);
\draw[edge] (v3)--(v6);
\draw[edge] (v4)--(v7);
\draw[edge] (v4)--(v8);
\draw[edge] (v5)--(v9);
\draw[edge] (v5)--(v10);
\draw[edge] (v6)--(v11);
\draw[edge] (v6)--(v12);
\draw[edge] (v7)--(v13);
\draw[edge] (v7)--(v14);
\draw[edge] (v8)--(v15);
\draw[edge] (v8)--(v16);
 
\node[right] at (3,-4) {\textcolor{blue}{$s$}};
\node[left] at (t) {\textcolor{blue}{$t$}};

\end{tikzpicture}
\caption{An instance of \textsc{Vehicle Routing with Target Groups} with a single target $t$ and $b(\{t\})=1$.
Distances to $s$ correspond to $x$-coordinates, other distances correspond to differences of $y$-coordinates.
On the left, an optimum solution (an $s$-$t$-tour) is shown. 
On the right, we see the edges $\{\parent_V(v),v\}$ for $v \in V$. 
The total length of these edges, and hence the cost of any forward walk solution, 
is more than a factor $\frac{1}{3}\log_2 n$ larger.\label{figure:no_cheap_forward_walk_solution}}
\end{center}
\end{figure}
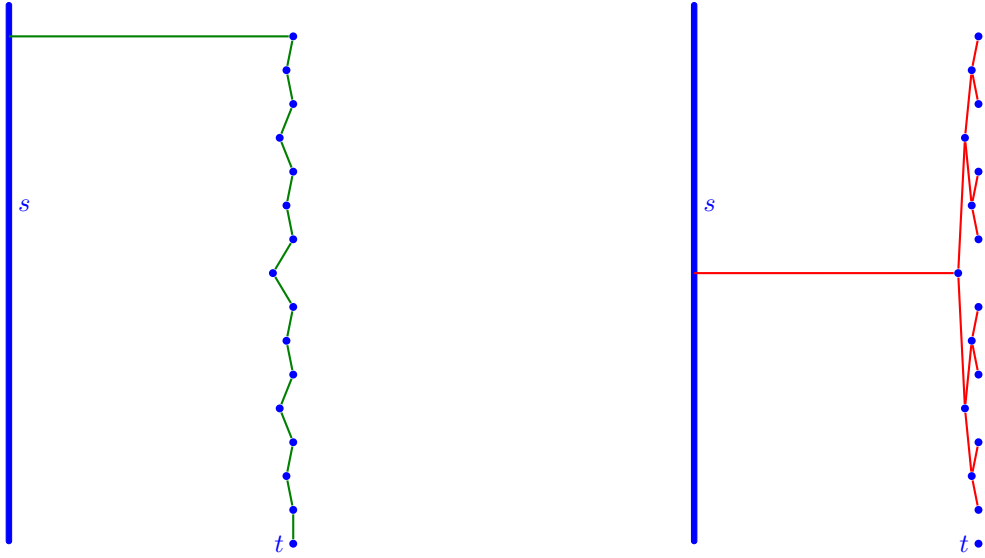

\begin{definition} 
Let $\gamma >2$ be a constant. 
Let $\Jeuscr=(V,\bar T,s,c,\mathcal{T},b)$ be an instance of \textsc{Vehicle Routing with Target Groups} 
with a depot distance order $\prec$. 
We call the pair $(\Jeuscr,\prec)$ \emph{nice} if for each $v \in V$:
\begin{equation}\label{eq:parent_close}
c(\parent_V(v),v) \ \leq \ \gamma \cdot \detour(y,v) \ \text{ for all } y \in V \text{ with } v\prec y.
\end{equation}
\end{definition}

\begin{figure}[ht]
\begin{center}
\begin{tikzpicture}
\begin{scope}[scale=0.9,rotate=270]
\draw[->, thick] (-3.3,-3)--(-3.3,1);
\node[below] at (-3.3,1) {$\prec$};

\draw[draw=yellow,fill=yellow] (-3,0) arc (180:360:3)--cycle;

\draw[draw=red,fill=red] plot[domain=-1.5:1.5,samples=50] (\x,{(2.25-\x*\x)/3});
\node[red,right] at (0.6,0.6) {\small{$A_{v}$}};

\tikzstyle{vertex}=[blue,circle,fill,minimum size=3,inner sep=0pt]
\node[vertex] (w) at (0,0) {};			
\node[blue, left] at (0,0) {$v$};
\node[vertex] (parent) at (2.4,-1.8) {};
\node[blue, left] at (2.4,-1.8) {$\parent(v)$};

\end{scope}
\end{tikzpicture}
\end{center}
\caption{In a nice instance, the parent of any vertex $v$ must not be much farther away than any other vertex.
For $\gamma=2$, Euclidean distances, and the depot located far west, this means that if the large (yellow) hemicycle is empty, 
the (red) area $A_v$ must also be empty.\label{fig:explain_nice}}
\end{figure}
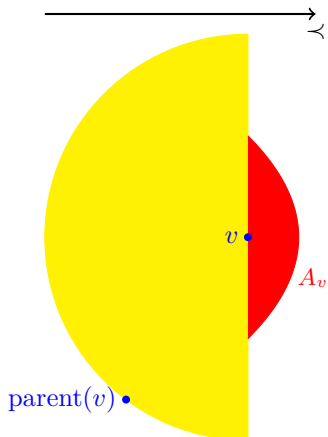				

See Figure~\ref{fig:explain_nice} for an illustration.
We now show that if an instance is nice, then it has a cheap forward walk solution:

\begin{lemma}\label{lemma:cheap_if_nice}
Let $\Jeuscr=(V,\bar T,s,c,\mathcal{T},b)$ be an instance of \textsc{Vehicle Routing with Target Groups} 
with a depot distance order $\prec$ such that $(\Jeuscr,\prec)$ is nice.
Then, for every weak fractional solution $x$ there is a forward walk solution of cost at most
\[
\sfrac{1}{1-\tau} \cdot c(x) + (2\gamma -  \sfrac{\tau}{1-\tau}) \cdot \detour(x|_{E_1(\Jeuscr)}).
\]
\end{lemma}

\begin{proof}
Let $x$ be a weak fractional solution. Then $x$ can be interpreted as a flow
in $G(\Jeuscr)$ that sends $(1-\tau)\sum_{\Tg  \in\Tscr}b(\Tg  )$ units out of $s$ 
and $(1-\tau)b(\Tg  )$ units into $\Tg  $ for all $\Tg  \in\Tscr$.
This maps to a flow $x'$ in $G'$ with the same properties, where $G'$ is the digraph constructed in the
proof of Theorem~\ref{lemma:find_cheap_forward_walk}.

In order to make $x'$ feasible, we need to send more flow into the target groups and 
ensure at least one unit of flow along $(v^-,v^+)$ for all $v\in V$.

First, to ensure enough flow arriving at the target groups, 
we send $\frac{\tau}{1-\tau} x(\delta^-(t))$ units of flow directly from $s$ to $t$ for all $t\in \bar T$.
This makes the missing $\tau b(\Tg  )$ units of flow arrive in $\Tg  $ for each target group $\Tg  \in\Tscr$.
This step increases the cost by $\frac{\tau}{1-\tau}(c(x)-\detour(x))$.

Second, we ensure at least one unit of flow along $(v^-,v^+)$ for all $v\in V$.
To this end, for each $v\in V$, we add $\max\{0,1-x'_{(v^-,v^+)}\}$ units of flow along the cycle
$(\parent(v)^+,v^-)$, $(v^-,v^+)$, $(v^+,\parent(v)^+)$.
Let us call the resulting flow $f$.
Note that $x$ ships at least one unit through $v$ and hence at least $\max\{0,1-x'_{(v^-,v^+)}\}$ 
into $v$ along backward arcs. Therefore the cost of $f$ is
\begin{align*}
& c(x) + \sfrac{\tau}{1-\tau} (c(x) -\detour(x) )
   + \sum_{v\in V} \max\{0,1-x'_{(v^-,v^+)}\} \cdot 2 \cdot c(\parent(v),v) \\
\ \le \ & c(x) + \sfrac{\tau}{1-\tau} (c(x) -\detour(x) )
   + \sum_{v\in V} \sum_{y\in V:v\prec y} x_{(y,v)} \cdot 2 \cdot c(\parent(v),v) \\
\ \le \ & c(x) + \sfrac{\tau}{1-\tau} (c(x) -\detour(x) )
   + \sum_{v\in V} \sum_{y\in V:v\prec y} x_{(y,v)} \cdot 2\gamma\cdot \detour(y,v) \\
\ \le \ & \sfrac{1}{1-\tau} \cdot c(x) + (2\gamma -  \sfrac{\tau}{1-\tau}) \cdot \detour(x|_{E_1(\Jeuscr)}),
\end{align*}
where we used that $(\Jeuscr,\prec)$ is nice in the second inequality.
\end{proof}

We remark that it is essential that we allow walks rather than paths in a forward walk solution
because otherwise the cost could increase by a factor $\Omega(n)$, even for nice instances.

Of course, not every instance is nice. We will now deal with general instances.

\subsection{Computing a nice subinstance and overall algorithm\label{subsection:compute_nice_subinstance}}

If a given instance is not nice, we compute a maximal nice subinstance.

\begin{lemma}\label{lemma:compute_nice_subinstance}
Given an instance of \textsc{Vehicle Routing with Target Groups} with a depot distance order, 
we can compute in $O(n^2)$ time 
a subset $Y\subseteq V$ such that the instance that results from deleting $Y$ is nice
and for every $y\in Y$ there exists a $w(y)\in V\setminus Y$ such that $w(y)\prec y$ and
\begin{equation*}
c \bigl( \parent_{V\setminus Y}(w(y)),w(y) \bigr) \ > \ \gamma \cdot \detour(y,w(y)).
\end{equation*}
\end{lemma}

\begin{proof}
Initially set $Y=\emptyset$. 
Scan the elements of $V$ in the depot distance order $\prec$.
If we scan an element that has been added to $Y$ before, we do nothing.
Otherwise, if we scan an element $v$ that has not been added to $Y$, then we add all elements of
\begin{equation}\label{eq:def_A_w}
A_v \ := \ \left\{ y\in V: v\prec y,\, c(\parent_{V\setminus Y}(v),v) > \gamma\cdot\detour(y,v) \right\}
\end{equation}
to $Y$ and set $w(y):=v$ for all $y\in A_v$.
See Figure~\ref{fig:explain_nice}.
\end{proof}

We will solve the nice subinstance that results from deleting $Y$ as described in the previous subsection. 
The deleted vertices will then be inserted into the solution of this subinstance.
To this end, we will compute a cheapest forest that connects all elements of $Y$ to some elements of $V\setminus Y$;
in the end we traverse the edges of this forest in both directions.
The main technical lemma, to be shown in the next subsection, says that there is a cheap forest
that connects all elements of $Y$ to some elements of $V\setminus Y$:

\begin{lemma}\label{lemma:short_forest_from_fractional} 
Let $\Jeuscr$ be an instance of \textsc{Vehicle Routing with Target Groups} and $x$ a weak fractional solution.
Let $\gamma>2$. 
Let $Y$ be a set of vertices as in Lemma~\ref{lemma:compute_nice_subinstance}. 
Then there exists a forest $(V,F_x)$ that connects each vertex in $Y$ to $V\setminus Y$ and such that
\begin{equation}\label{eq:boundlengthofforest_fractional}
	c(F_x) \ < \ \sfrac{4}{\gamma-2} \cdot c(x|_{E_1(\Jeuscr)}) + (2+\sfrac{4}{\gamma-2}) \cdot \detour(x|_{E_1(\Jeuscr)}).
\end{equation}
\end{lemma}

Using this lemma, which we prove in the following subsection, 
we can prove Theorem~\ref{thm:mainVRPwTargetGroupsCombinatorial}:

\begin{proof}[Proof of Theorem~\ref{thm:mainVRPwTargetGroupsCombinatorial}]
The algorithm consists of the following steps.
\begin{enumerate} 
\item Apply Lemma~\ref{lemma:compute_nice_subinstance} to compute a subset $Y\subseteq V$.
\item Apply Lemma~\ref{lemma:find_cheap_forward_walk} to find a cheapest forward walk solution $H$
in the nice subinstance obtained by deleting $Y$.
\item Find a minimum-cost set $F$ of edges that connect all elements of $Y$ to some elements of $V\setminus Y$.
\item Orient the edges of $F$ in both directions and insert them into $H$ to obtain a walk solution to the original instance.
\item Apply Proposition~\ref{prop:walk_solution_suffices}.
\end{enumerate}
Step 3 reduces to computing a minimum-cost spanning tree in the complete graph $Y\cup\{w\}$ where the cost
of an edge $\{y,w\}$ is $\min\{c(y,v): v\in V\setminus Y\}$ and can thus be done in $O(n^2)$ time. 
Step 5 runs in $O(n^2)$ time because $|H|\le n^2$.
Hence the overall running time is dominated by the application of Lemma~\ref{lemma:find_cheap_forward_walk}.
It is obvious that step 4 yields a walk solution, hence the algorithm is correct.

To bound the total cost of the solution, let $x$ be any weak fractional solution to a given instance $\Jeuscr$.
From $x$ we obtain a weak fractional solution to our nice subinstance by shortcutting without increasing cost or detour.
By Lemma~\ref{lemma:cheap_if_nice}, 
$c(H)\le \frac{1}{1-\tau} \cdot c(x) +(2\gamma-\frac{\tau}{1-\tau}) \cdot \detour(x|_{E_1(\Jeuscr)})$.

By Lemma~\ref{lemma:short_forest_from_fractional}, there is a forest $(V,F_x)$ that 
connects each vertex in $Y$ to $V\setminus Y$ and such that \eqref{eq:boundlengthofforest_fractional} holds.
We conclude that our final solution has total cost at most
\begin{equation*}
c(H)+2c(F) \ \le \ 
\sfrac{1}{1-\tau} \cdot c(x) + \sfrac{8}{\gamma-2} \cdot c(x|_{E_1(\Jeuscr)}) 
+(4+\sfrac{8}{\gamma-2}+2\gamma-\sfrac{\tau}{1-\tau}) \cdot \detour(x|_{E_1(\Jeuscr)}),
\end{equation*}
as required.
\end{proof}

\subsection{Connecting the omitted vertices (Proof of Lemma~\ref{lemma:short_forest_from_fractional})}\label{section:short_forest}

To prove Lemma~\ref{lemma:short_forest_from_fractional}, we fix a set
$Y$ of vertices as in Lemma~\ref{lemma:compute_nice_subinstance}
and decompose $x$ into walks, which we shortcut to paths that begin in $s$ and whose other vertices belong to $Y$.
Let $P$ be such a path.
We will construct a forest that connects each vertex in $Y\cap V(P)$ to some vertex in $W:=V\setminus Y$. 
This forest will consist of pieces of $P$ and a connection of each piece to $W$.
More specifically, we will prove the following and apply it to each of these paths:

\begin{lemma}\label{lemma:short_forest} 
Let $\gamma>2$. 
Let $Y$ be a set of vertices as in Lemma~\ref{lemma:compute_nice_subinstance}. 
Let $P$ be a path that begins in $s$ and with $V(P) \setminus \{s\} \subseteq Y$. 
Then there exists a forest $(V,F_P)$ that connects each vertex in $Y \cap V(P)$ to $V\setminus Y$ and such that
\begin{equation}\label{eq:boundlengthofforest}
	c(F_P) \ < \ \sfrac{4}{\gamma-2} \cdot c(P) + (2+\sfrac{4}{\gamma-2}) \cdot \detour(P).
\end{equation}
\end{lemma}

Before we prove this lemma, we show that it implies Lemma~\ref{lemma:short_forest_from_fractional}.

\begin{proof}[Proof of Lemma~\ref{lemma:short_forest_from_fractional}]
We decompose $x=\sum_{P\in\Pscr}\lambda_P\chi^{E(P)}$ according to Definition~\ref{def:weak_fractional_solution}.
Each $P\in\Pscr$ is a walk in $G(\Jeuscr)$ that begins in $s$ and ends in $\{s\}\cup \bar T$.
Now shortcut $P$ by skipping all vertices in $V\setminus Y$,
skipping a vertex in $Y$ if it has been visited by $P$ before, and removing the end vertex of $P$.
Apply Lemma~\ref{lemma:short_forest} to the resulting path, and let $F_P$ be the resulting forest.
Consider $l=\sum_{P\in\Pscr} \lambda_P \chi^{F_P}$. 
The cost of $l$ is $\sum_{P\in\Pscr} \lambda_P c(F_P) \le \sum_{P\in\Pscr} \lambda_P
\left( \sfrac{4}{\gamma-2} \cdot c(P) + (2+\sfrac{4}{\gamma-2}) \cdot \detour(P) \right)
\le \sfrac{4}{\gamma-2} \cdot c(x|_{E_1(\Jeuscr)}) + (2+\sfrac{4}{\gamma-2}) \cdot \detour(x|_{E_1(\Jeuscr)})$.
In the last inequality we used the triangle inequality: shortcutting does not increase cost or detour.

We claim that after contracting $V\setminus Y$ to a vertex $w$, the vector $l$ belongs to the connector polyhedron
(the convex hull of incidence vectors of all connected multigraphs on vertex set $Y\cup\{w\}$). This directly implies the result.

To show that $l$ belongs to the connector polyhedron, 
we need to show that for every partition $\Yscr=\{Y_0,\ldots,Y_k\}$ of $V$ for which $V\setminus Y\subseteq Y_0$
we have $l(\delta(\Yscr))\ge k$, where $\delta(\Yscr)$ denotes the set of edges with endpoints in different sets of the partition.
For every $P\in\Pscr$, let $k_P=\{j\in\{1,\ldots,k\}: F_P\cap \delta(Y_j)\not=\emptyset\}$ 
be the number of sets that $F_P$ connects to $Y_0$.
Since  $\sum_{P\in\Pscr: v\in V(P)} \lambda_P\ge 1$ for all $v\in Y$, we have $\sum_{P\in\Pscr}\lambda_P k_P \ge k$.
Moreover, $k_P\le |F_P\cap\delta(\Yscr)|$. 
Therefore $l(\delta(\Yscr))=\sum_{P\in\Pscr} \lambda_P |F_P\cap\delta(\Yscr)| \ge \sum_{P\in\Pscr} \lambda_P k_P \ge k$,
as required.
\end{proof}

\begin{figure}[ht]
\begin{center}
\begin{tikzpicture}
\begin{scope}[scale=1.1,rotate=270]
\draw[->, thick] (-3.3,-5)--(-3.3,6);
\node[below] at (-3.3,6) {$\prec$};
\tikzstyle{wvertex}=[circle,fill,minimum size=5,inner sep=0pt]
\tikzstyle{yvertex}=[circle,draw,minimum size=4,inner sep=0pt]

\draw[violet,very thick] (5.1,-0.1)--(5.1,0.4);
\draw[red,very thick] (5.0,0)--(5.0,1.5);
\draw[brown,very thick] (5.1,0.5)--(5.1,1.15);
\draw[darkgreen,very thick] (5.0,2.3)--(5.0,3.6);
\draw[blue,very thick] (5.1,3.2)--(5.1,3.95);

\fill[red!10] (-3.25,0) rectangle (4.95,1.5);
\fill[darkgreen!10] (-3.25,2.3) rectangle (4.95,3.6);
\node[above,red] at (4.95,0.75) {$S_{w_2}$};
\node[above,darkgreen] at (4.95,2.95) {$S_{w_4}$};

\draw[draw=blue!50, thick, densely dotted] (-0.6,3.2) arc (180:360:1.5)--cycle;
\draw[draw=blue,fill=blue!30] plot[domain=-0.75:0.75,samples=50] ({\x+0.9},{3.2+(9/16-\x*\x)/1.5});
\node[blue,right] at (0.8,3.5) {{$A_{w_5}$}};
\node[wvertex,blue] (w5) at (0.9,3.2) {};			
\node[blue, left] at (0.7,3.3) {$w_5$};

\draw[draw=darkgreen!50, thick, densely dotted] (-0.5,2.3) arc (180:360:2.6)--cycle;
\draw[draw=darkgreen,fill=darkgreen!30] plot[domain=-1.3:1.3,samples=50] ({\x+2.1},{2.3+(1.69-\x*\x)/2.6});
\node[darkgreen,right] at (3.2,2.45) {{$A_{w_4}$}};
\node[wvertex,darkgreen] (w4) at (2.1,2.3) {};			
\node[darkgreen, left] at (2.2,2.3) {$w_4$};

\draw[draw=brown!50, thick, densely dotted] (-2.5,0.5) arc (180:360:1.3)--cycle;
\draw[draw=brown,fill=brown!30] plot[domain=-0.65:0.65,samples=50] (\x-1.2,{0.5+(1.69/4-\x*\x)/1.3});
\node[brown,right] at (-1.7,0.6) {{$A_{w_3}$}};
\node[wvertex,brown] (w3) at (-1.2,0.5) {};			
\node[brown, left] at (-1.45,0.6) {$w_3$};

\draw[draw=red!50, thick, densely dotted] (-3,0) arc (180:360:3)--cycle;
\draw[draw=red,fill=red!30] plot[domain=-1.5:1.5,samples=50] (\x,{(2.25-\x*\x)/3});
\node[red,right] at (0.1,0.7) {{$A_{w_2}$}};
\node[wvertex,red] (w2) at (0,0) {};			
\node[red, left] at (0,0) {$w_2$};

\draw[draw=violet!50, thick, densely dotted] (2.1,-0.1) arc (180:360:1)--cycle;
\draw[draw=violet,fill=violet!30] plot[domain=-0.5:0.5,samples=50] (3.1+\x,{-0.1+(0.25-\x*\x)/1});
\node[violet,right] at (3.2,0.1) {{$A_{w_1}$}};
\node[wvertex,violet] (w1) at (3.1,-0.1) {};			
\node[violet, left] at (3.1,-0.1) {$w_1$};

\node[wvertex,gray] (parent2) at (1.8,-2.4) {};
\node[gray,left] at (1.8,-2.4) {$\parent(w_2)$};
\node[wvertex,gray] (parent1) at (3.9,-0.7) {};
\node[gray,left] at (3.9,-0.7) {$\parent(w_1)$};

\node[yvertex] (y1) at (2.9,0.05) {};
\node[above right] at (y1) {$x(y)$};
\node[yvertex] (y2) at (0.1,0.4) {};
\node[below] at (y2) {$y$};
\node[yvertex] (y3) at (-0.7,0.5) {};
\node[yvertex] (y4) at (-1.2,0.7) {};
\node[right] at (y4) {$z(y)$};
\node[yvertex] (y5) at (2.5,2.45) {};
\node[below] at (y5) {$x(y')$};
\node[yvertex] (y6) at (1.9,2.8) {};
\node[below right] at (y6) {$\!\!\!\!y'\!=\!z(y')$};
\node[yvertex] (y7) at (1.2,3.3) {};

\draw (2.9,-2) -- (y1) -- (y5) -- (y2) -- (y3) -- (y4) -- (y7) -- (y6);
\draw[->] (y6) -- (1.9,6);
\node[below] at (2.9,-2) {$P$};
\node[below] at (1.9,6) {$P$};

\end{scope}
\end{tikzpicture}
\end{center}
\caption{An example how a set $\bar Y \subseteq Y$ of nondominated customers could be chosen so that
$\bigcup_{y\in\bar Y} V(P_{[x(y),z(y)]})$ contains all elements of $Y\cap V(P)$. 
Here the filled circles represent customers $\parent(w_1),\parent(w_2),w_1,w_2,w_3,w_4,w_5$ in $W=V \setminus Y$, 
and the empty circles represent customers in $Y\cap V(P)$.
The hemicycles and the sets $A_w$ and $S_w$ ($w\in W$) are drawn for $\gamma=2$, Euclidean distances, and the depot located far west
(the stripes $S_{w_1},S_{w_3},S_{w_5}$ are only indicated at the bottom).
The customer $z(y)$ in $A_{w_3}$ is dominated because $w(y)=w_2\prec w_3 = w(z(y))$ and $S_{w_3}\subseteq S_{w_2}$.
A possible choice is $\bar Y=\{y,y'\}$.\label{fig:select_bar_Y}}
\end{figure}
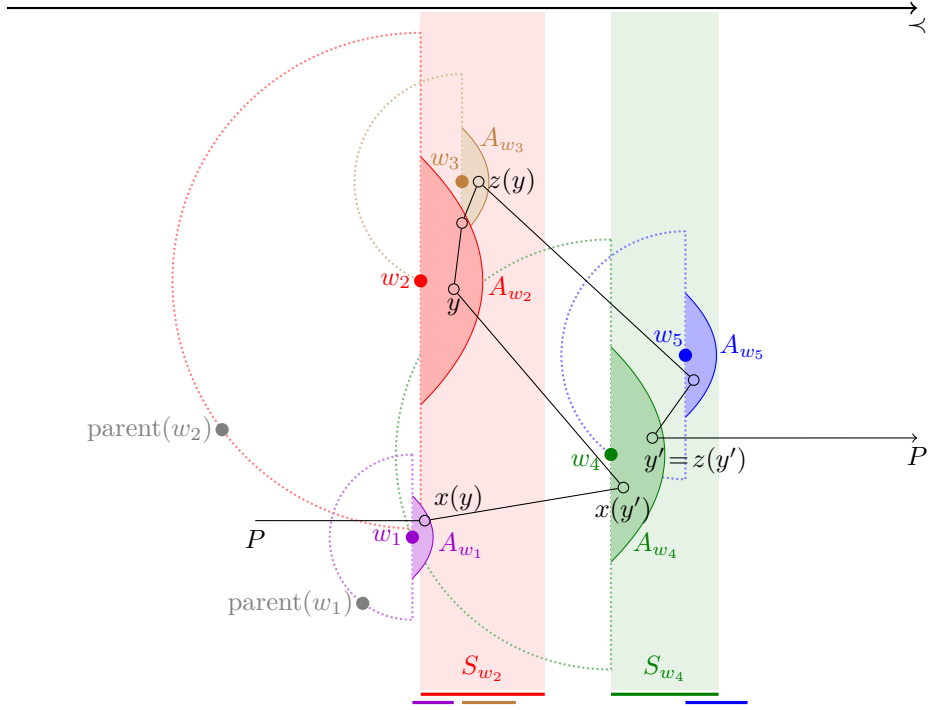

Now it remains to prove Lemma~\ref{lemma:short_forest}.

First, we will associate with each $w\in W$ a ``radius''
\begin{equation}
r_w \ := \ \sfrac{1}{\gamma} \cdot c \bigl(\parent_W(w),w \bigr).
\end{equation}
We have 
$A_w =\{ y\in V: w\prec y,\,  \detour(y,w) < r_w\}$ (cf.\ \eqref{eq:def_A_w}), 
in particular every element of $A_w$ has distance to $w$ smaller than $r_w$ (cf.\ Figure~\ref{fig:explain_nice}).

Now we define a larger region $S_w$ by
\begin{equation}
S_w \ := \ \{y \in Y : w\prec y,\, c(s,y) \le c(s,w) + r_w \}.
\end{equation}
Note that $A_w\subseteq S_w$.

Let now $P$ be a path that begins in $s$ and with $V(P)\setminus\{s\}\subseteq Y$.
For $y\in Y\cap V(P)$ let $w(y)$ be a vertex as guaranteed by Lemma~\ref{lemma:compute_nice_subinstance}.
Let $x(y)$ be the first vertex of $P$ with $w(y)\prec x(y)$,
and let $z(y)$ be the last vertex of $P$ with $c(s,z(y))\le c(s,w(y))+ r_{w(y)}$.
Denote the $x(y)$-$z(y)$-subpath of $P$ by $P_{[x(y),z(y)]}$. 
This subpath contains all of $S_{w(y)}\cap V(P)$, but it may contain more elements of $Y$.
We will select a subset $\bar Y$ of $Y\cap V(P)$. 
For each selected customer $y\in\bar Y$ we add $P_{[x(y),z(y)]}$ plus the edge $\{y,w(y)\}$ to the forest. 
The set $\bar Y$ will be chosen such that this forest connects all customers in $V(P) \cap Y$ to $W$.

We call $y\in V(P)\cap Y$ \emph{dominated} by $y'\in V(P)\cap Y$ if $w(y') \prec w(y)$ and $S_{w(y)}\subseteq S_{w(y')}$.
We will not include dominated customers in $\bar Y$.
If $y$ is dominated by $y'$, then $P_{[x(y),z(y)]}$ is a subpath of $P_{[x(y'),z(y')]}$,
so we can indeed cover all of $V(P)\cap Y$ without considering dominated vertices.
See Figure~\ref{fig:select_bar_Y}.

Let $y$ be nondominated, and let $p(y)$ be the predecessor of $x(y)$ on $P$. Then $p(y)\prec w(y)$.
The key lemma that we use to bound the cost of our forest is the following.
The idea of its proof is that either $x(y)$ is close to $w(y)$, then $p(y)$ is far away from $x(y)$, or 
$x(y)$ is far away from $w(y)$, then $P$ makes a large detour between $x(y)$ and $y$.
(In Figure~\ref{fig:select_bar_Y}, $y'$ is of the first kind and $y$ is of the second kind.) 

\begin{lemma}\label{lem:key_forest_lemma}
Let $y\in Y$ be nondominated and $w=w(y)$. 
Let $p\in\{s\}\cup Y$ and $x,z\in Y$ with $p\prec w\prec x$ and $c(s,z) \le c(s,w)+ r_w$.
Then 
\begin{equation}\label{eq:forest_bound_main_lemma}
c(x,y)+c(y,z)+c(y,w) \ \le \ (1+\sfrac{2}{\gamma-2}) \cdot \left( \detour(x,y) + \detour(y,z) \right) + \sfrac{4}{\gamma-2} \cdot c(p,x).
\end{equation}
\end{lemma}

\begin{proof}
First we claim
\begin{equation}\label{eq:boundrw}
(\gamma-2) \cdot r_w \ \le \ 
2 \cdot c(p,x) + \detour(x,y).
\end{equation}
If $p=s$, we have $c(p,x)\ge c(p,w)\ge c(\parent(w),w)=\gamma \cdot r_w$ and hence \eqref{eq:boundrw}.
So assume $p\in Y$. 
Then 
\begin{equation}\label{eq:boundnondominated}
c(s,w) +  r_w \ \ge \ c(s,w(p)) + r_{w(p)} 
\end{equation}
because otherwise $S_{w}\subseteq S_{w(p)}$ and $y$ would be dominated.
Next, using the triangle inequality, $p\in A_{w(p)}$, $y\in A_w$, $w\prec x$, and \eqref{eq:boundnondominated}, we get
\begin{align*}
\gamma\cdot r_w & \ \le \ c(w(p),w) \\
& \ \le \ c(w(p),p) + c(p,x) + c(x,y) + c(y,w) \\ 
& \ = \ (\detour(p,w(p)) + c(s,w(p)) - c(s,p)) + c(p,x) + \detour(x,y) + \detour(y,w) + c(s,w) - c(s,x) \\
& \ \le \ r_{w(p)} + c(s,w(p)) - c(s,p) + c(p,x) + \detour(x,y) + r_w + c(s,w) - c(s,x) \\
& \ \le \ r_{w(p)} + c(s,w(p)) - c(s,p) + c(p,x) + \detour(x,y) + r_w + c(s,x) - c(s,w) \\
& \ \le \ r_{w(p)} + c(s,w(p)) + 2c(p,x) + \detour(x,y) + r_w - c(s,w) \\
& \ \le \ r_{w(p)} + 2c(p,x) + \detour(x,y) + r_w + r_w - r_{w(p)}, 
\end{align*}
which yields \eqref{eq:boundrw}.

Finally, using  $c(s,z)-c(s,x)\le c(s,z)-c(s,w)\le r_w$ and $c(y,w)\le \detour(y,w)\le r_w$,
\begin{align*}
c(x,y) + c(y,z) + c(y,w) 
&\ = \ \detour(x,y) + \detour(y,z) + c(s,z)-c(s,x) + c(y,w) \\
&\ \leq \ \detour(x,y) + \detour(y,z) + 2\cdot r_w \\
&\ \le \ \detour(x,y) + \detour(y,z) + \sfrac{4}{\gamma-2} c(p,x) + \sfrac{2}{\gamma-2}  \detour(x,y),
\end{align*}
where we used the bound for $r_w$ that we derived in \eqref{eq:boundrw} in the second inequality.
\end{proof}

To finish the proof of Lemma~\ref{lemma:short_forest},
let $\bar Y\subseteq V(P)\cap Y$ be a minimal subset of nondominated vertices such that
\begin{equation*}
\bigcup_{y\in\bar Y} V(P_{[x(y),z(y)]}) \ = \ Y\cap V(P).
\end{equation*}
Then no vertex belongs to more than two of the subpaths $P_{[x(y),z(y)]}$ ($y\in\bar Y$).
In particular, if we sum up the detours of these subpaths, we get at most twice the total detour of $E(P)$.

Taking the union of these subpaths plus the edges $\{y,w(y)\}$ for $y\in\bar Y$ yields
a forest that connects all elements of $Y\cap V(P)$ to $V\setminus Y$ and that has length at most
\begin{align*}
&\sum_{y\in\bar Y} \left( c(P_{[x(y),z(y)]}) + c(y,w(y)) \right) \\
\ = \ & \sum_{y\in\bar Y} \bigl( c(x(y),y)+ c(y,z(y)) + \detour(P_{[x(y),y]}) + \detour(P_{[y,z(y)]})  - \detour(x(y),y) - \detour(y,z(y)) \\
& \hspace*{15mm} + c(y,w(y)) \bigr) \\
\ \le \ & 2\detour(P) + \sum_{y\in\bar Y} \left(\big. c(x(y),y)+ c(y,z(y)) + c(y,w(y))  - \detour(x(y),y) - \detour(y,z(y)) \right).  \\
\intertext{
By Lemma~\ref{lem:key_forest_lemma}, this is at most
}
& 2\detour(P) + \sum_{y\in\bar Y} \left( \sfrac{2}{\gamma-2} \cdot \left( \detour(x(y),y) + \detour(y,z(y)) \right) 
+ \sfrac{4}{\gamma-2} \cdot c(\pred_P(x(y)),x(y)) \right) \\
\ \le \ &(2+\sfrac{4}{\gamma-2})\detour(P) + \sfrac{4}{\gamma-2} \cdot c(P),
\end{align*}
where $\pred_P(x(y))$ denotes the predecessor of $x(y)$ on $P$.
In the last inequality we used that by the minimality of $\bar Y$, the edges $(\pred_P(x(y)),x(y))$ are distinct for different $y\in \bar Y$.
This concludes the proof of Lemma~\ref{lemma:short_forest}.

\section{LP-based approach to vehicle routing with target groups\label{sec:lpbased}}

In this section we give a second proof of Theorem~\ref{thm:mainVRPwTargetGroups},
using an approach of Friggstad and Swamy \cite{friggstad2017compact}.
This will lead to a worse running time but a better approximation ratio.
To be precise, we prove:

\begin{theorem}\label{thm:mainVRPwTargetGroups2}
There is a polynomial-time algorithm for \textsc{Vehicle Routing with Target Groups} that
computes for every instance $\Jeuscr=(V,\bar T,s,c,\mathcal{T}, b)$ and any given $\theta  \in (0,1-\tau]$
a feasible solution $\Pscr$ to $\Jeuscr$ such that 
\[
c(\Pscr) \ < \  \frac{1}{\theta}  \cdot c(x) +  \frac{1-\tau - \theta}{\theta \cdot (1-\tau)} \cdot c(x|_{E_1(\Jeuscr)}) 
      + \frac{3}{1-\theta} \cdot \detour(x|_{E_1(\Jeuscr)}),
\]
for every weak fractional solution $x$ of $\Jeuscr$. 
\end{theorem}

Large parts of this section follow closely \cite{friggstad2017compact}, although our setting is slightly different 
and we improve on this work at one point (in Section~\ref{sec:existence_cheap_lp_solution}).
In particular, the linear program (Section~\ref{sect:FS_LP}) is largely identical to the one in \cite{friggstad2017compact},
and the rounding procedure (Section~\ref{sec:rounding_lp_solution}) is taken from that paper. 
Nevertheless we cannot simply use a result of \cite{friggstad2017compact} in a black-box manner, 
so for sake of readability we give a self-contained proof here.

Friggstad and Swamy~\cite{friggstad2017compact} devised an LP-based 15-approximation algorithm for
regret-bounded vehicle routing. Here we have a metric space with a depot and a set of customers as well as a regret bound $R$ 
and ask for a set $\Pscr$ of paths such that each customer belongs to some path and each path begins at the depot,
ends at some customer, and has total detour at most $R$. The goal is to minimize the number of paths in $\Pscr$.
So on the one hand, the paths can end anywhere, on the other hand, 
the detour of each single path is bounded and we minimize the number of paths.
In Section~\ref{sect:improve_FS} we show how the approximation ratio can be improved from 15 to 10.

\subsection{The linear program\label{sect:FS_LP}}

\begin{definition}
An $s$-$t$-walk $P$ in $G(\Jeuscr)$ with $t\in \bar T \cup \{s\}$ is called \emph{monotone}
if $c(s,u) < c(s,v)$ for every edge $(u,v)\in E(P)$ with  $v\in V$.
A solution $\Pscr$ is monotone if all $P\in\Pscr$ are monotone.
\end{definition}

Suppose $\Jeuscr$ has a feasible solution $\Pscr$ with small detour.
We will observe that then there exists a subset $U$ of customers such that
\begin{itemize}\itemsep0pt
\item the instance $(U,\bar T,s,c,\Tscr,b)$ has a monotone solution $\Pscr_U$ that is no longer than $\Pscr$, and
\item there is a short forest $F$ connecting every vertex $v\in V\setminus U$ to a vertex in $U$.
\end{itemize}
If we knew the solution $\Pscr_U$ for $(U,\bar T,s,c,\Tscr,b)$ and the forest $F$, 
we could obtain a cheap solution for $\Jeuscr$ as follows. We take two copies of the edges in $F$ and orient them in different directions.
Then we add them to $\Pscr_U$ to obtain a walk solution.

We will consider an LP relaxation for the problem of finding a subset $U\subseteq V$, a monotone solution $\Pscr_U$, 
and a forest $(V,F)$ connecting every vertex to an element of $U$.

Let us first explain how we can obtain $\Pscr_U$ and $F$ from a solution $\Pscr$ of $\Jeuscr$ with small detour.
Of course, $\Pscr$ is unknown but our construction leads to a particular choice of $\Pscr_U$ and $F$ 
for which we can impose certain constraints in the LP.

For every walk $P\in\Pscr$, say from $s$ to $t$, we color the points in
$[0,\max\{c(s,u):u\in V(P)\}]$ blue or green.
A point $x$ is \emph{green} if there is an edge $(v,w)$ of $P$ such that 
$c(s,u)< x$ for all $u\in V(P_{[s,v]})\setminus\{s\}$ and $c(s,u)>x$ for all $u\in V(P_{[w,t]}) \setminus \{t\}$.
All other points are \emph{blue}, including $c(s,u)$ for all $u\in V(P)\setminus\{s,t\}$.
The blue intervals of $P$ form a subset of  the set $\Iscr$ of all closed intervals with endpoints among 
$\{c(s,u): u\in V(P)\setminus\{s,t\}\}$, including the trivial intervals consisting of a single point.

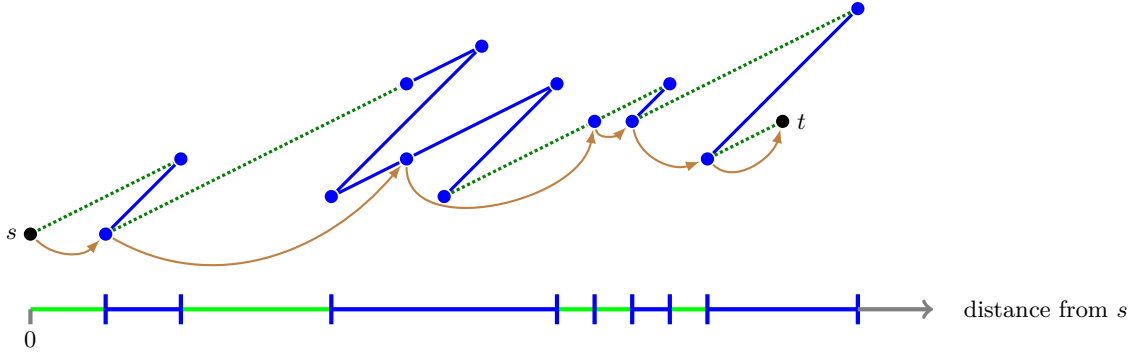
\begin{figure}[ht]
\begin{center}
\begin{tikzpicture}[scale=1]
\tikzstyle{vertex}=[blue,circle,fill,minimum size=5,inner sep=0pt]
\tikzstyle{egreen}=[darkgreen,line width=1.2, densely dotted]
\tikzstyle{eblue}=[blue,line width=1.2]

\begin{scope}[line width=1.6]
 \draw[green] (0,0) -- (1,0);
 \draw[blue] (1,0) -- (2,0);
 \draw[green] (2,0) -- (4,0);
 \draw[blue] (4,0) -- (7,0);
 \draw[green] (7,0) -- (8,0);
 \draw[blue] (8,0) -- (8.5,0);
 \draw[green] (8.5,0) -- (9,0);
 \draw[blue] (9,0) -- (11,0);
 \draw[blue] (1,-0.2) -- (1,0.2); 
 \draw[blue] (2,-0.2) -- (2,0.2); 
 \draw[blue] (4,-0.2) -- (4,0.2); 
 \draw[blue] (7,-0.2) -- (7,0.2); 
 \draw[blue] (7.5,-0.2) -- (7.5,0.2); 
 \draw[blue] (8,-0.2) -- (8,0.2); 
 \draw[blue] (8.5,-0.2) -- (8.5,0.2); 
 \draw[blue] (9,-0.2) -- (9,0.2); 
 \draw[blue] (11,-0.2) -- (11,0.2); 
 \draw[gray,->] (11,0) -- (12,0);
 \draw[gray] (0,-0.2) -- (0,0.0); 
 \node at (0,-0.4) {\small 0};
 \node at (13.5,0) {\small distance from $s$};
 \node at (-0.25,1) {\small $s$};
 \node[black] at (10.25,2.5) {\small $t$};
 \end{scope}
 
\begin{scope}[every node/.style={vertex}, line width=1.2]
 \node[black] (s) at (0,1) {};
 \node (v1) at (2,2) {};
 \node (v2) at (1,1) {};
 \node (v3) at (5,3) {};
 \node (v4) at (6,3.5) {};
 \node (v5) at (4,1.5) {};
 \node (v6) at (5,2) {};
 \node (v7) at (7,3) {};
 \node (v8) at (5.5,1.5) {};
 \node (v9) at (7.5,2.5) {};
 \node (v10) at (8.5,3) {};
 \node (v11) at (8,2.5) {};
 \node (v12) at (11,4) {};
 \node (v13) at (9,2) {};
 \node[black] (t) at (10,2.5) {};
\end{scope}
\draw[egreen] (s)--(v1);
\draw[eblue] (v1)--(v2);
\draw[egreen] (v2)--(v3);
\draw[eblue] (v3)--(v4)--(v5)--(v6)--(v7)--(v8);
\draw[egreen] (v8)--(v9)--(v10);
\draw[eblue] (v10)--(v11);
\draw[egreen] (v11)--(v12);
\draw[eblue] (v12)--(v13);
\draw[egreen] (v13)--(t);

\begin{scope}[brown, ->, >=latex, thick]
\draw (s) to[out=-45, in=-135] (v2);
\draw (v2) to[out=-30, in=-130] (v6);
\draw (v6) to[out=-90, in=-100] (v9);
\draw (v9) to[out=-80, in=-135] (v11);
\draw (v11) to[out=-80, in=-160] (v13);
\draw (v13) to[out=-45, in=-105] (t);
\end{scope}
\end{tikzpicture}
\end{center}
\caption{\label{fig:blue_intervals_and_forward_path} 
An example of a path from $s$ to $t$, consisting of the dotted green and solid blue edges.
The interval $[0,\max\{c(s,u):u\in V(P)\}]$, is colored green and blue (bottom).
The blue edges are those of the forest $F$.
The brown directed arcs show the monotone path visiting the sentinels of $P$.
}
\end{figure}

Then for every vertex $v\in V(P)\setminus \{s,t\}$, the distance $c(s,v)$ is contained in one of the blue intervals.
We consider the forest $F_P$ that contains those edges $\{v,w\}$ of $P$ for which $c(s,v)$ and $c(s,w)$ are contained in the same blue interval and $v,w\in V(P)\setminus\{s,t\}$. See Figure~\ref{fig:blue_intervals_and_forward_path}.
We will bound the cost of this forest using the following observation, which is essentially
Corollary 3.2 of \cite{blum2007approximation}.

\begin{lemma}[\cite{blum2007approximation}]\label{lem:bound_forest_length_lp}
We have $c(F_P) \le \frac{3}{2}\detour(E(P)\cap E_1(\Jeuscr))$.
\end{lemma}
\begin{proof}
Obtain $P'$ from $P$ by deleting the last edge and subdividing edges so that $P'$ 
alternates between subpaths in green and subpaths in blue intervals.
Then the total length of the green subpaths is at least the total length of the green intervals,
and the total length of the blue subpaths is at least three times the total length of the blue intervals up to $c(s,v)$,
where $v$ is the last vertex of $P'$.
Hence 
$c(s,v) \le c(P'_{\text{green}}) + \frac{1}{3}c(P'_{\text{blue}}) = c(P') - \frac{2}{3}c(P'_{\text{blue}}) \le  c(P') - \frac{2}{3}c(F_P)$.
Since $\detour(E(P)\cap E_1(\Jeuscr))=c(P') - c(s,v)$, this yields the claimed inequality.
\end{proof}

There is a one-to-one correspondence between the connected components of $(V(P)\setminus\{s,t\},F_P)$ and the blue intervals of $P$.
For every blue interval $I$, we choose one arbitrary vertex $v\in V(P)\setminus \{s,t\}$ with $c(s,v)\in I$.
We call these vertices the \emph{sentinels} of $P$.
We have exactly one sentinel for each connected component of $(V(P)\setminus\{s,t\},F_P)$.
The sentinels of $P$ are visited by $P$ in an order of increasing distance from $s$.
Hence, shortcutting $P$ such that it visits only the sentinels yields a monotone path of length at most $c(P)$.

The forest $F$ is the union of all $F_P$ for $P\in \Pscr$ and the monotone solution $\Pscr_U$ arises from $\Pscr$ by shortcutting $\Pscr$ such that it only visits the sentinels.

Next we describe a linear programming relaxation for this forest $F$ and monotone solution $\Pscr_U$.
For $u\in V$ denote $\Iscr_u:=\{I\in\Iscr: c(s,u)\in I\}$.
Recall that for every vertex $v\in V$, there is a path $P\in \Pscr$ that contains $v$, 
and there is a blue interval $I\in \Iscr_v$ of $P$ and a sentinel $u$ such that $I\in \Iscr_u$.
We introduce variables 
\[
 a_{v,(I,u)} \ge 0 \qquad (u,v\in V, I \in \Iscr_u \cap \Iscr_v)
\]
that model the assignment of $v$ to the interval $I$ with sentinel $u$.
We remark that if $v$ is a sentinel of interval $I$ we have $a_{v,(I,v)}=1$ in the corresponding LP solution.
These variables satisfy the following constraints:
\begin{align*}
a_{v,(I,u)} &\le a_{u, (I,u)} &(u,v\in V,\, I\in\Iscr_u\cap\Iscr_v) \\
\sum_{u\in V,\, I\in\Iscr_u\cap \Iscr_v} a_{v,(I,u)} &= 1 &(v\in V).
\end{align*}
To model the forest $F$, we introduce variables $z_e \in \mathbb{R}_{\ge 0}$ for all $e\in {V\choose 2}$.
If an edge $e\in {V\choose 2}$ is contained in $F$, we have $z_e=1$ in the corresponding LP solution (and we have $z_e=0$ otherwise).
If $v\in V(P)\setminus\{s,t\}$ for $P\in \Pscr$, then there is a blue interval $I\in \Iscr_v$ of $P$, 
and $v$ must be connected to the sentinel $u$ of $I$ in the forest $(V,F)$. Hence we can introduce the constraints
\begin{align*}
z(\delta(S)) &\ge \sum_{u\in V\setminus S,\, I\in\Iscr_u}a_{v,(I,u)} &(v\in S\subseteq V).
\end{align*}
Finally, we introduce flow constraints modelling the monotone solution $\Pscr_U$.
Define a digraph $B$ with vertex set $\{(I,u): u\in V,\, I\in\Iscr_u\}\cup \{s\} \cup \bar T$ and four types of arcs:
\begin{itemize}\itemsep=0pt
\item arcs $(s,(I,u))$ for $u\in V$ and $I \in \Iscr_u$,
\item arcs $((I,u),(J,v))$ whenever $c(s,u)<c(s,v)$ and $I \in \Iscr_u$ and $J\in \Iscr_v$ are disjoint,
\item arcs  $((I,u),t)$ for $t\in \bar T\cup \{s\}$, $u\in V$, and $I \in \Iscr_u$, and
\item arcs $(s,t)$ for $t\in \bar T$.
\end{itemize}
We remark that the arc set is constructed such that $B-s$ is acyclic.
The arcs in $B$ inherit their cost from $c$, i.e., $c(s,(I,u))=c(s,u)$, $c((I,u),(J,v))=c(u,v)$, etc.

We introduce flow variables $f_{e} > 0$ for every arc $e$ in $B$.
For our solution $\Pscr_U$, we define the corresponding flow $f$ in $B$ as follows.
\begin{itemize}\itemsep=0pt
\item For an edge $(s,u)$ in  $\Pscr_U$, where $u$ is the sentinel of the interval $I$, we set $f_{(s,(I,u))} :=1$.
\item For an edge $(u,u')$ in $\Pscr_U$, where $u$ is the sentinel of the interval $I$ and $I'$ is the interval of the sentinel $u'$ we set $f_{((I,u),(I',u'))} :=1$.
\item For an edge $(u,t)$ in $\Pscr_U$, where $u$ is the sentinel of the interval $I$ and $t\in \bar T \cup \{s\}$, we set $f_{((I,u), t)} :=1$.
\item For an edge $(s,t)$ in $\Pscr_U$, where $t\in \bar T$, we set $f_{(s,t)}:=1$.
\end{itemize}
If an edge appears multiple times, the flow values add up.
All other flow variables are zero.
Then the flow $f$ fulfills the following constraints:
\begin{align*}
\sum_{u\in V} f(\delta^-(\Tg  )) &= b(\Tg  ) &(\Tg  \in\Tscr) \\
f(\delta^-((I,u))) = f(\delta^+((I,u))) &=a_{u,(I,u)} &(u\in V,\, I\in\Iscr_u).
\end{align*}
To prove Theorem~\ref{thm:mainVRPwTargetGroups2}, we want to construct a feasible LP solution from a weak fractional solution to $\Jeuscr$. 
To this end, we will consider the polyhedron $\Pi$ below, 
where we replaced the flow constraints $\sum_{u\in V} f(\delta^-(\Tg  )) = b(\Tg  )$ by $\sum_{u\in V} f(\delta^-(\Tg  )) = (1-\tau)\cdot b(\Tg  )$ 
and relaxed $\sum_{u\in V,\, I\in\Iscr_u\cap \Iscr_v} a_{v,(I,u)} = 1$ to $\sum_{u\in V,\, I\in\Iscr_u\cap \Iscr_v} a_{v,(I,u)} \ge 1$.
We define
\begin{equation}\label{eq:FS_LP}
\begin{aligned}
\Pi =\Big\{& a,f,z : & \\
& &\sum_{u\in V} f(\delta^-(\Tg  )) &= (1-\tau) \cdot b(\Tg  ) &(\Tg  \in\Tscr) \\
& & f(\delta^-((I,u))) = f(\delta^+((I,u))) &=a_{u,(I,u)} &(u\in V,\, I\in\Iscr_u) \\
& & a_{v,(I,u)} &\le a_{u,(I,u)} &(u,v\in V,\, I\in\Iscr_u\cap\Iscr_v) \\
& & \sum_{u\in V,\, I\in\Iscr_u\cap \Iscr_v} a_{v,(I,u)} &\ge 1 &(v\in V) \\
& & z(\delta(S)) &\ge \sum_{u\in V\setminus S,\, I\in\Iscr_u\cap\Iscr_v} a_{v,(I,u)} &(v\in S\subseteq V) \\
& & a,f,z &\ge 0 \\
\Big\}&.
\end{aligned}
\end{equation}
Note that the number of constraints can be reduced to a polynomial number 
by introducing flow variables $g^v_e$ for each $v \in V$ and $e\in{V\choose 2}$. 
Require that the flow $g^v$ sends $\sum_{I \in \Iscr_u} a_{v,(I,u)}$ units from $v$ to $u$ and $z \ge g^v$. 
Then the cut constraints for $z$ can be dropped without changing $\Pi$.

Later, in Section~\ref{sec:rounding_lp_solution}, we show how we can construct a cheap solution to $\Jeuscr$ 
from any element $(a,f,z)\in \Pi$ of the polyhedron $\Pi$.
We write $c(z):= \sum_{e =\{v,w\}\in{V\choose 2}} c(v,w) z_e$ 
and $c(f):= \sum_{e\in E(B)} c(e)f_e$.
We will show:

\begin{lemma}\label{lem:rounding_lp_solution}
Let $0 < \theta \le 1-\tau$.
Given $(a,f,z)\in \Pi$ we can compute a solution $\Pscr$ of the instance $\Jeuscr$ of \textsc{Vehicle Routing with Target groups}  with
\[
c(\Pscr) \ \le \ \frac{4}{1-\theta} c(z) +  \frac{1}{\theta}  \cdot c(f) 
+  \frac{1-\tau - \theta}{\theta \cdot (1-\tau)} \cdot
\sum_{t\in \bar T}\sum_{u\in V,I\in\Iscr_u} c(s,u) \cdot f((I,u),t)
\]
in polynomial time.
\end{lemma}
In Section~\ref{sec:existence_cheap_lp_solution} we show how we can obtain a vector $(a,f,z)\in \Pi$ from a weak fractional solution $x$ of $\Jeuscr$ and prove the following lemma. 
Note that we gain a factor $2$ in the bound on $z$, compared to Lemma~\ref{lem:bound_forest_length_lp}.
\begin{lemma}\label{lem:cheap_lp_solution_exists}
For any weak fractional solution $x$ of  the instance $\Jeuscr$ of \textsc{Vehicle Routing with Target groups} there exists a vector $(a,f,z)\in \Pi$ such that
\begin{align*}
 c(z)  \le&\ \sfrac{3}{4}\detour(x|_{E_1(\Jeuscr)}), \\
c(f) \le&\  c(x), \\
\sum_{t\in \bar T}\sum_{u\in V,I\in\Iscr_u} c(s,u) \cdot f((I,u),t) \le&\ c(x|_{E_1(\Jeuscr)}).
\end{align*}
\end{lemma}
From Lemma~\ref{lem:rounding_lp_solution} and Lemma~\ref{lem:cheap_lp_solution_exists} we can now conclude Theorem~\ref{thm:mainVRPwTargetGroups2}.

\begin{proof}[Proof of Theorem~\ref{thm:mainVRPwTargetGroups2}]
We compute an optimum solution $(a,f,z)$ to the LP
\[
\min \Big\{ \frac{4}{1-\theta} c(z) +  \frac{1}{\theta}  \cdot c(f) 
+  \frac{1-\tau - \theta}{\theta \cdot (1-\tau)} \cdot   \sum_{t\in \bar T}\sum_{u\in V,I\in\Iscr_u} c(s,u) \cdot f((I,u),t) \ :\ (a,f,z) \in \Pi \Big\}.
\]
By Lemma~\ref{lem:cheap_lp_solution_exists} we have
\begin{align*}
&\ \frac{4}{1-\theta} c(z) +  \frac{1}{\theta}  \cdot c(f) 
+  \frac{1-\tau - \theta}{\theta \cdot (1-\tau)} \cdot  \sum_{t\in \bar T}\sum_{u\in V,I\in\Iscr_u} c(s,u) \cdot f((I,u),t)
\\
 \le&\ \frac{3}{1-\theta} \detour(x|_{E_1(\Jeuscr)}) + \frac{1}{\theta}  \cdot c(x)
+  \frac{1-\tau - \theta}{\theta \cdot (1-\tau)} \cdot  c(x|_{E_1(\Jeuscr)}). 
\end{align*}
Hence, applying Lemma~\ref{lem:rounding_lp_solution} completes the proof.
\end{proof}

\subsection{Existence of a cheap LP solution (Proof of Lemma~\ref{lem:cheap_lp_solution_exists})}\label{sec:existence_cheap_lp_solution}

In this section we prove Lemma~\ref{lem:cheap_lp_solution_exists}.
Given a weak fractional solution $x$ of $\Jeuscr$ we will construct a vector $(a,f,z)\in \Pi$.
By the definition of weak fractional solutions (Definition~\ref{def:weak_fractional_solution}) we can write
\begin{equation*}
    x = \sum_{P\in \Pscr} \lambda_P \chi^{E(P)},
\end{equation*}
where $\Pscr$ is a set of walks in $G(\Jeuscr)$, each starting in $s$ and ending in some $t\in \bar T\cup \{s\}$, and $\lambda_P \ge 0$ for all $P\in \Pscr$.
For every walk $P\in \Pscr$, we consider the forest $(V,F_P)$.
Recall that every connected component of $(V,F_P)$ corresponds to a blue interval $I\in\Iscr$ of $P$.
For every connected component $Z$ of $(V,F_P)$ let $u^Z_1$ and $u^Z_2$ be the vertex in $Z$ visited first and visited last by $\Pscr$, respectively.

For each $v\in Z$ we set
\[
    a^P_{v,(I,u^Z_1)} = a^P_{v,(I,u^Z_2)} = \sfrac{1}{2}
\] 
if $u^Z_1\not=u^Z_2$, and $a^P_{v,(I,u^Z_1)} = 1$ if $u^Z_1=u^Z_2$.
We set $a^P_{v,(I,u)}=0$ for all other $(I,u)$ with $u\in V$ and $I\in \Iscr_u \cap \Iscr_v$. 
We aggregate
\[
 a_{v,(I,u)} =\sum_{P\in \Pscr: v\in V(P)} \lambda_P \cdot  a^P_{v,(I,u)} .
\]
Then $a_{v,(I,u)} \le a_{u,(I,u)}$ for all $u,v\in V,\, I\in\Iscr_u\cap\Iscr_v$ and 
$\sum_{u\in V,\, I\in\Iscr_u\cap \Iscr_v} a_{v,(I,u)} = \sum_{P\in \Pscr: v\in V(P)} \lambda_P \ge 1$ for all $v\in V$.

\begin{figure}[ht]
\begin{center}
\begin{tikzpicture}[scale=1]
\tikzstyle{vertex}=[blue,circle,fill,minimum size=5,inner sep=0pt]
\tikzstyle{egreen}=[darkgreen,line width=1.2, densely dotted]
\tikzstyle{eblue}=[blue,line width=1.2]

\begin{scope}[shift={(0,-5)}]

\begin{scope}[line width=1.2]

 \draw[blue] (1,1) -- (2,1);
 \draw[blue] (1,0.8) -- (1,1.2); 
 \draw[blue] (2,0.8) -- (2,1.2);

 \draw[blue] (1,2) -- (2,2);
 \draw[blue] (1,1.8) -- (1,2.2); 
 \draw[blue] (2,1.8) -- (2,2.2);

 \draw[blue] (4,3) -- (7,3);
  \draw[blue] (4,2.8) -- (4,3.2); 
 \draw[blue] (7,2.8) -- (7,3.2); 

 \draw[blue] (4,1.5) -- (7,1.5);
  \draw[blue] (4,1.3) -- (4,1.7); 
 \draw[blue] (7,1.3) -- (7,1.7); 

  \draw[blue] (7.5,2.3) -- (7.5,2.7); 

 \draw[blue] (8,3) -- (8.5,3);
 \draw[blue] (8,2.8) -- (8,3.2); 
 \draw[blue] (8.5,2.8) -- (8.5,3.2);

 \draw[blue] (8,2.5) -- (8.5,2.5);
 \draw[blue] (8, 2.3) -- (8,2.7); 
 \draw[blue] (8.5, 2.3) -- (8.5,2.7);

 \draw[blue] (9,4) -- (11,4);
  \draw[blue] (9,3.8) -- (9,4.2); 
 \draw[blue] (11,3.8) -- (11,4.2); 

 \draw[blue] (9,2) -- (11,2);
  \draw[blue] (9,1.8) -- (9,2.2); 
 \draw[blue] (11,1.8) -- (11,2.2); 

\end{scope}
 
\begin{scope}[every node/.style={vertex}, line width=1.2]
 \node[black] (s) at (0,1) {};
 \node[red] (v1) at (2,2) {};
 \node[brown] (v2) at (1,1) {};
 \node[red] (v3) at (5,3) {};
 \node[brown] (v8) at (5.5,1.5) {};
 \node[orange] (v9) at (7.5,2.5) {};
 \node[red] (v10) at (8.5,3) {};
 \node[brown] (v11) at (8,2.5) {};
 \node[red] (v12) at (11,4) {};
 \node[brown] (v13) at (9,2) {};
 \node[black] (t) at (10,2.5) {};
\end{scope}

\begin{scope}[thick, >=latex]
\draw[red,->] (s) to (v1);
\draw[red,->, bend right=10] (v1) to (v3);
\draw[red,->] (v3)  to (v9);
\draw[red,->] (v9) to  (v10);
\draw[red,->] (v10) to (v12);
\draw[red,->] (v12) to (t);
\draw[brown,->] (s) to  (v2);
\draw[brown,->, bend right=30] (v2) to (v8);
\draw[brown,->] (v8) to (v9);
\draw[brown,->] (v9) to (v11);
\draw[brown,->] (v11) to (v13);
\draw[brown,->] (v13) to (t);
\end{scope}

 \node at (-0.25,1) {\small $s$};
 \node[black] at (10.25,2.5) {\small $t$};
\end{scope}


\begin{scope}[line width=1.2]
 \draw[green] (0,0) -- (1,0);
 \draw[blue] (1,0) -- (2,0);
 \draw[green] (2,0) -- (4,0);
 \draw[blue] (4,0) -- (7,0);
 \draw[green] (7,0) -- (8,0);
 \draw[blue] (8,0) -- (8.5,0);
 \draw[green] (8.5,0) -- (9,0);
 \draw[blue] (9,0) -- (11,0);
 \draw[blue] (1,-0.2) -- (1,0.2); 
 \draw[blue] (2,-0.2) -- (2,0.2); 
 \draw[blue] (4,-0.2) -- (4,0.2); 
 \draw[blue] (7,-0.2) -- (7,0.2); 
 \draw[blue] (7.5,-0.2) -- (7.5,0.2); 
 \draw[blue] (8,-0.2) -- (8,0.2); 
 \draw[blue] (8.5,-0.2) -- (8.5,0.2); 
 \draw[blue] (9,-0.2) -- (9,0.2); 
 \draw[blue] (11,-0.2) -- (11,0.2); 
 \draw[gray,->] (11,0) -- (12,0);
  \draw[gray] (0,-0.2) -- (0,0.0); 
 \node at (0,-0.4) {\small 0};
 \node at (13.5,0) {\small distance from $s$};
 \node at (-0.25,1) {\small $s$};
 \node[black] at (10.25,2.5) {\small $t$};
 \end{scope}
 
\begin{scope}[every node/.style={vertex}, line width=1.2]
 \node[black] (s) at (0,1) {};
 \node[red] (v1) at (2,2) {};
 \node[brown] (v2) at (1,1) {};
 \node[red] (v3) at (5,3) {};
 \node (v4) at (6,3.5) {};
 \node (v5) at (4,1.5) {};
 \node (v6) at (5,2) {};
 \node (v7) at (7,3) {};
 \node[brown] (v8) at (5.5,1.5) {};
 \node[orange] (v9) at (7.5,2.5) {};
 \node[red] (v10) at (8.5,3) {};
 \node[brown] (v11) at (8,2.5) {};
 \node[red] (v12) at (11,4) {};
 \node[brown] (v13) at (9,2) {};
 \node[black] (t) at (10,2.5) {};
\end{scope}
\draw[egreen] (s)--(v1);
\draw[eblue] (v1)--(v2);
\draw[egreen] (v2)--(v3);
\draw[eblue] (v3)--(v4)--(v5)--(v6)--(v7)--(v8);
\draw[egreen] (v8)--(v9)--(v10);
\draw[eblue] (v10)--(v11);
\draw[egreen] (v11)--(v12);
\draw[eblue] (v12)--(v13);
\draw[egreen] (v13)--(t);

\end{tikzpicture}
\end{center}
\caption{\label{fig:blue_intervals_and_sentinels} 
A path $P$ (at the top) with vertices $u^Z_1$ in red and orange and vertices $u^Z_2$ in brown and orange.
The bottom part of the picture shows the corresponding flow $f$ in the graph $B$ with vertices $s,t$ and
$(I,u^Z_i)$.
Both the red and the brown edges have flow value $\frac{1}{2}$.
\label{fig:construct_flow_in_B}
}
\end{figure}
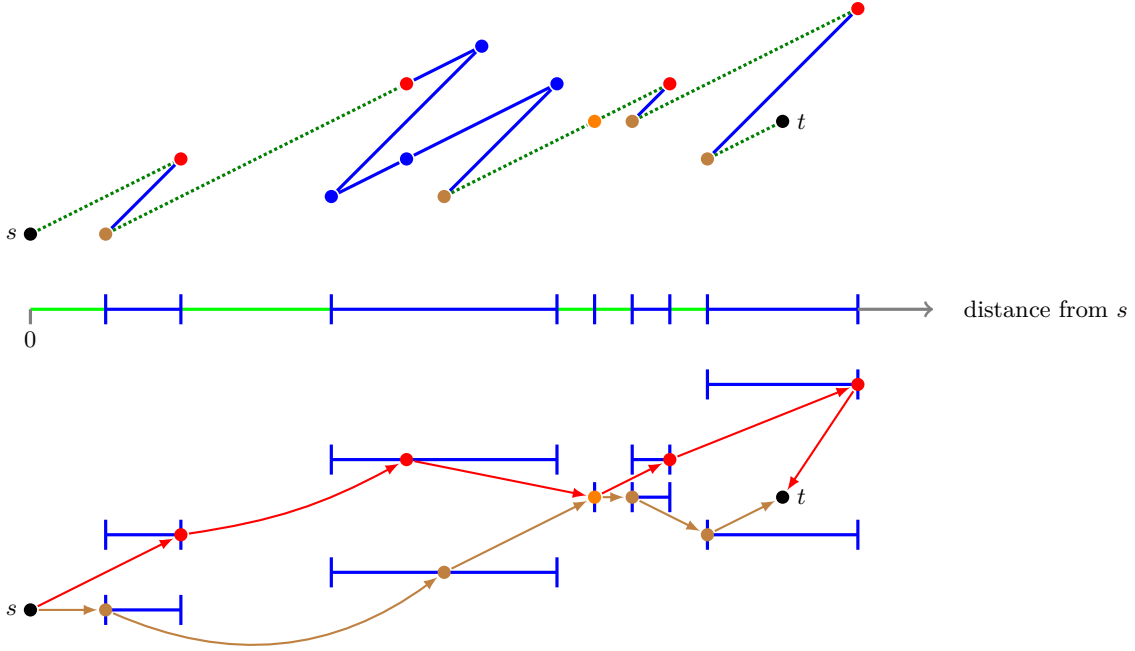

For each $P\in \Pscr$ we define two monotone walks $P_1$ and $P_2$.
The walk $P_i$ results from the $s$-$t$-walk $P$ by shortcutting such that it visits only the vertices $s$, $t$, and 
$u^Z_i$ for the connected components $Z$ of $(V,F_P)$.
From these two paths we obtain $s$-$t$-walks $P'_1$ and $P'_2$ in the graph $B$ by replacing every vertex $u^Z_i$ by $(I,u^Z_i)$, where $I$ is the blue interval corresponding to the component $Z$.
See Figure~\ref{fig:construct_flow_in_B}.
We set 
\[
    f := \sum_{P\in \Pscr} \lambda_P \cdot \frac{1}{2} \cdot \left( \chi^{E(P'_1)} +  \chi^{E(P'_2)}\right).
\]
Then we have $f(\delta^-((I,u))) = f(\delta^+((I,u))) =a_{u,(I,u)}$ for all $u\in V,\, I\in\Iscr_u$
 and because $x$ is a weak fractional solution of $\Jeuscr$ we also have $\sum_{u\in V} f(\delta^-(\Tg  )) = (1-\tau) \cdot b(\Tg  )$ for all $\Tg  \in\Tscr$.
We have $c(f) \le c(x)$ and
\[
\sum_{t\in \bar T} \sum_{u\in V,I\in\Iscr_u} c(s,u) \cdot f((I,u),t) \ \le \ 
     \sum_{P\in\Pscr} \lambda_P \cdot c(E(P)\cap E_1(\Jeuscr)) 
\ \le \ c(x|_{E_1(\Jeuscr)}).
\]

Finally, we define 
\[
 z := \frac{1}{2} \sum_{P\in \Pscr} \lambda_P \cdot \chi^{F_P}.
\]

\begin{lemma}
We have $z(\delta(S)) \ge \sum_{u\in V\setminus S,\, I\in\Iscr_u\cap\Iscr_v}a_{v,(I,u)}$ for all $v\in S\subseteq V$.
\end{lemma}
\begin{proof}
Let $v\in V$ and let $P\in \Pscr$ be a walk containing $v$.
Moreover,  let $Z$ be the vertex set of the connected component of $(V,F_P)$ that contains $v$
and let $I$ be the blue interval corresponding to $Z$.
If $F_P$ contains only $v$, then the inequality holds trivially because its right hand-side is always $0$ in this case.
Otherwise $u^Z_1\not=u^Z_2$ and
$a_{v,(I,u^Z_1)}= a_{v,(I,u^Z_2)} = \frac{1}{2}$ and $a_{v,(I,u)}=0$ for all other $(I,u)$.
The connected component $(Z, F[Z])$ of $(V,F)$ that contains $v$ is the subpath of $P$ from $u^Z_1$ to $u^Z_2$.
Therefore, the $v$-$u^Z_1$-path  and the $v$-$u^Z_2$-path in $(V,F)$ are edge-disjoint.
This implies  $z(\delta(S)) \ge \sum_{u\in V\setminus S,\, I\in\Iscr_u\cap\Iscr_v}a_{v,(I,u)}$ for all $S$.
\end{proof}

By Lemma~\ref{lem:bound_forest_length_lp}, we have
\begin{align*}
    c(z) =\  \frac{1}{2} \cdot \sum_{P\in \Pscr} \lambda_P \cdot c(F_P)  
\  \le\  \frac{1}{2} \cdot \sum_{P\in \Pscr} \lambda_P \cdot \frac{3}{2}\detour(E(P)\cap E_1(\Jeuscr)) \
         \ =\ \frac{3}{4}\detour(x|_{E_1(\Jeuscr)}).
\end{align*}

This completes the proof of Lemma~\ref{lem:cheap_lp_solution_exists}.

\subsection{Rounding LP solutions (Proof of Lemma~\ref{lem:rounding_lp_solution})}\label{sec:rounding_lp_solution}

Let $(a,f,z)\in \Pi$. 
First, we explain how we round $z$ to obtain a forest $(V,F)$.
For $\emptyset\not=S\subsetneq V$, let
\[
r(S) := \begin{cases}
1 & \text{ if } \sum_{u\in S, I\in\Iscr_u\cap\Iscr_v} a_{v,(I,u)} < \theta\text{ for all }v\in S\\
0 & \text{ otherwise.}
\end{cases}
\]  
\begin{lemma}[\cite{friggstad2017compact}]\label{lem:forest_rounding_GW}
We can compute a forest $(V,F)$ with $|F\cap \delta(S)|\ge r(S)$ for all $\emptyset \ne S\subsetneq V$
and 
\[c(F)\le \frac{2}{1-\theta} \cdot c(z)\] 
in polynomial time.
\end{lemma}
\begin{proof}
Consider a nonempty set $S\subsetneq V$.
Then $S$ contains at least one vertex $v$ and by the LP constraints
\[ z(\delta(S))\ \ge\ \sum_{u\in V\setminus S,I\in \Iscr_u\cap\Iscr_v} a_{v,(I,u)}\ =\ \sum_{u\in V,I\in \Iscr_u\cap\Iscr_v} a_{v,(I,u)} - \sum_{u\in S,I\in \Iscr_u\cap\Iscr_v} a_{v,(I,u)}\ \ge\ 1 - \sum_{u\in S,I\in \Iscr_u\cap\Iscr_v} a_{v,(I,u)}.\]
This implies $z(\delta(S))\ge 1- \theta$ for all $S$ with $r(S)=1$.
In other words, $\frac{1}{1-\theta}z(\delta(S))\ge r(S)$ for all $S$.
Moreover, we observe that $r$ is a downwards-monotone cut-requirement function, i.e.\ we have
$r(X) \ge r(Y)$ for any $X,Y$ with $\emptyset\not=X \subseteq Y$.
Hence, an algorithm by Goemans and Williamson \cite{GoemansWilliamson} can be used to compute a forest $(V,F)$ as required.
\end{proof}

Let $\Zscr$ denote the set of (vertex sets of) connected components of $(V,F)$.
For every $Z\in \Zscr$ we have $r(Z)=0$ and hence there is a vertex $w\in Z$ with $\sum_{u\in Z, I\in\Iscr_u\cap\Iscr_w} a_{w,(I,u)} \ge \theta$.
We choose such a vertex $w$ for each connected component $Z$ and call it the \emph{witness node} of $Z$.
Then we define
\[
  \sigma(Z) := \{ (I,u) : u\in Z, I\in \Iscr_u\cap\Iscr_w \}.
\]
Note that $\sigma(Z)\subset V(B)$. 

\begin{lemma}[\cite{friggstad2017compact}]\label{lem:components_have_sufficient_flow}
For every $Z\in \Zscr$ we have $f(\delta^-(\sigma(Z))) =  f(\delta^+(\sigma(Z)))\ge \theta$.
\end{lemma}
\begin{proof}
Let $w$ be the witness node of $Z$.
Because $w\in I$ for every $(I,u)\in \sigma(Z)$ and because the graph $B$ contains edges 
$((I,u),(I',u'))$ only if $I$ and $I'$ are disjoint, $B$ contains no edge with both endpoint in $\sigma(Z)$.
Hence, for the witness node $w$ of $Z$ the LP constraints imply
\[
    f(\delta^-(\sigma(Z))) =  f(\delta^+(\sigma(Z))) = \sum_{(I,u)\in \sigma(Z)} f(\delta^-((I,u))) = \sum _{(I,u)\in \sigma(Z)} a_{u,(I,u)} \ge\sum_{u\in Z, I\in\Iscr_u\cap\Iscr_w} a_{w,(I,u)} \ge \theta. 
\]
\end{proof}

Next, we define an order $\prec$ on $\Zscr$.
For $Z,Z'\in \Zscr$ with witness nodes $w$ and $w'$ respectively, we say that $Z \prec Z'$ if and only if
$c(s,w) < c(s,w')$.

We define a directed graph $D$ with vertex set $\{s\}\cup \Zscr \cup \bar T$ and the following types of arcs:
\begin{itemize}\itemsep=0pt
\item an arc $(s,Z)$, for $Z\in \Zscr$
\item an arc $(Z,Z')$ for $Z,Z'\in \Zscr$ with $Z \prec Z'$,
\item an arc $(Z,t)$ for $Z\in \Zscr$ and $t\in \bar T\cup \{s\}$, and
\item an arc $(s,t)$ for all $t\in \bar T$.
\end{itemize}
We define $c(s,Z) := \min_{v\in Z} c(s,v)$, $c(Z,Z') := \min_{v\in Z,v'\in Z'} c(v,v')$, and
$c(Z,t) := \min_{v\in Z} c(v,t)$.

\begin{lemma}\label{lem:solution_from_flow_in_component_graph}
Given the forest $(V,F)$ and a flow $g$ in $D$ with
\begin{itemize}\itemsep0pt
\item $g(\delta^-_D(Z)) = g(\delta^+_D(Z)) = 1$ for every $Z\in \Zscr$, and
\item $g(\delta^-_D(\Tg  ))=b(\Tg  )$ for every target group $\Tg  \in \Tscr$,
\end{itemize}
we can compute a solution to the instance $\Jeuscr$ of \textsc{Vehicle Routing with Target Groups}
with cost at most $2 \cdot c(F) + c(g)$ in polynomial time.
\end{lemma}
\begin{proof}
It is well-known that there is an optimum flow $\bar g$ with the same properties that is integral.
Back in $G(\Jeuscr)$, this corresponds to a multi-set $H\subseteq E(\Jeuscr)$
with $|\delta^-_H(Z)|=|\delta^+_H(Z)|=1$ for every $Z\in\Zscr$ and $|\delta^-_H(\Tg  )|=b(\Tg  )$ for every target group $\Tg  \in\Tscr$.
Let $Z\in\Zscr$.
If $\delta^-_H(Z)=\{(v,z)\}$ and $\delta^+_H(Z)=\{(z',v')\}$, we let $P_Z$ be the path in $(V,F)$ that connects $z$ and $z'$,
oriented from $z$ to $z'$.
For each $e=\{v,w\}\in F$ such that $(v,w)$ belongs to $P_Z$ for its $Z\in\Zscr$, we add $(v,w)$ to $H$.
For each $e=\{v,w\}\in F$ such that neither $(v,w)$ nor $(w,v)$ belongs to any $P_Z$, we add 
both orientations $(v,w)$ and $(w,v)$ to $H$.
Then $H$ is a walk solution of cost at most $c(\bar g)+2c(F) \le c(g)+2c(F)$.
Now apply Proposition~\ref{prop:walk_solution_suffices}.
\end{proof}

We will construct a flow $g$ as in Lemma~\ref{lem:solution_from_flow_in_component_graph} with
\begin{equation}\label{eq:cost_of_flow_g}
      c(g) \ \le \ \frac{1}{\theta}c(f) + \frac{1-\tau - \theta}{\theta \cdot (1-\tau)} \cdot \sum_{t\in \bar T} \sum_{u\in V, I\in \Iscr_u}  c(s,u) \cdot f((I,u),t).
\end{equation}
Together with Lemma~\ref{lem:forest_rounding_GW}, this will complete the proof of Lemma~\ref{lem:rounding_lp_solution}.

First, we shortcut the flow $f$ as follows.
Because the graph $B-s$ is acyclic we can decompose $f$ into $s$-$t$-walks with $t\in \bar T \cup \{s\}$, i.e.\
we write $f = \sum_{H\in \Hscr} \lambda_H \cdot \chi^H$ where $\Hscr$ is a set of  $s$-$t$-walks 
with $t\in \bar T \cup \{s\}$ and $\lambda_H > 0$ for all $H\in \Hscr$.
Note that $\sum_{H\in \Hscr: V(H)\cap\sigma(Z)\not=\emptyset} \lambda_H =f(\delta^-(\sigma(Z)))\ge \theta$ for all $Z\in\Zscr$
by Lemma~\ref{lem:components_have_sufficient_flow}.

Then we shortcut each walk $H\in \Hscr$ to a walk $H'$ 
by skipping vertices that do not belong to $\{s\} \cup \bar T \cup \bigcup_{Z\in \Zscr} \sigma(Z)$.
If $\sum_{H\in \Hscr: V(H)\cap\sigma(Z)\not=\emptyset} \lambda_H > \theta$ for some $Z\in\Zscr$,
we will also skip $\sigma(Z)$ in some of the walks so that
$\sum_{H\in \Hscr: V(H')\cap\sigma(Z)\not=\emptyset} \lambda_H = \theta$;
this may require to replace a walk by two copies.
Shortcutting (skipping vertices) is possible in the graph $B$, although $B$ is not a complete digraph, because if $B$ contains two edges
$(b_1,b_2)$ and $(b_2,b_3)$ for $b_2\not=s$, it also contains the edge $(b_1,b_3)$. 
We define a flow $f'$ in $B$ to be $f' = \sum_{H\in \Hscr} \lambda_H \cdot \chi^{H'}$.
Because of the triangle inequality, we have $c(f') \le c(f)$. 

Next we contract every set $\sigma(Z)$ for all $Z\in \Zscr$.
This yields a flow $g'$ in the graph $D$ by identifying the vertex resulting from the contraction of $\sigma(Z)$ with the vertex $Z\in V(D)$.
Here we use that for an arc $((I,u),(J,v))$ in $B$ with $u\in Z_u\in \Zscr$ and $v\in Z_v\in \Zscr$ we have
$c(s,u) < c(s,v)$  and $I\cap J = \emptyset$, implying that $d_1 < d_2$ for every $d_1\in I$ and $d_2\in J$.
In particular, $c(s,w_u) < c(s,w_v)$ where $w_u$ and $w_v$ are the witness nodes of $Z_u$ and $Z_v$ respectively.
We have $c(g') \le c(f') \le c(f)$.

Finally, to define $g$, we start with the flow $\frac{1}{\theta}g'$ in $D$.
The transformation from $f$ to $f'$ ensured that
$\frac{1}{\theta}g'(\delta^-(Z)) = \frac{1}{\theta}g'(\delta^+(Z)) = 1$ for every $Z\in \Zscr$.
However, if $\theta<1-\tau$, there is too much flow arriving in the target groups.
More precisely, $\frac{1}{\theta}g'(\delta^-(\Tg  ))= \frac{1-\tau}{\theta}\cdot b(\Tg  )$ for every target group $\Tg  \in \Tscr$, but we need exactly $b(\Tg)$.
For every edge $(Z,t)$ with $t\in \bar T$ and $Z\in \Zscr$ we decrease the flow on
$(Z,t)$ from $\frac{1}{\theta}g'(Z,t)$ to $\frac{1}{1-\tau}g'(Z,t)$ and increase the flow on $(Z,s)$ by 
$\frac{1}{\theta}g'(Z,t) - \frac{1}{1-\tau}g'(Z,t) = \frac{1-\tau - \theta}{\theta \cdot (1-\tau)}g'(Z,t)$.
For every edge $(s,t)$ with $t\in \bar T$ we decrease the flow on $(s,t)$ from $\frac{1}{\theta}g'(Z,t)$ to $\frac{1}{1-\tau}g'(Z,t)$.
The resulting flow is the desired flow $g$.

The cost of $g$ is 
\begin{align*}
  c(g) \ \le&\ \frac{1}{\theta} c(g') +  \frac{1-\tau - \theta}{\theta \cdot (1-\tau)} \cdot \sum_{Z \in \Zscr} \sum_{t\in \bar T} c(s,Z) \cdot g'(Z,t) \\
        \le&\  \frac{1}{\theta}c(f) +   \frac{1-\tau - \theta}{\theta \cdot (1-\tau)} \cdot  \sum_{u\in V,I\in\Iscr_u} \sum_{t\in \bar T} c(s,u) \cdot f'((I,u),t).
\end{align*}
Because the graph $B$ visits pairs $(I,u)$ in an order of increasing $c(s,u)$, we have $c(s,u) < c(s,v)$ whenever
we shortcut $((I,u),(J,v))$ and $((J,v),t)$ to an edge $((I,u),t)$ in the construction of $f'$ from $f$.
Hence,
\[
  \sum_{u\in V, I\in \Iscr_u} \sum_{t\in \bar T} c(s,u) \cdot f'((I,u),t) \ \le\   \sum_{u\in V, I\in \Iscr_u} \sum_{t\in \bar T} c(s,u) \cdot f((I,u),t),
\]
implying \eqref{eq:cost_of_flow_g}.
This completes the proof of Lemma~\ref{lem:rounding_lp_solution}.

\subsection{Improving the approximation ratio for regret-bounded vehicle routing}\label{sect:improve_FS}

In this subsection we remark that our construction of a fractional LP solution (Section~\ref{sec:existence_cheap_lp_solution}) also leads
to an improved approximation ratio for regret-bounded vehicle routing, from 15 \cite{friggstad2017compact} to 10. 
We will not need this for our main result, but it might be worth noting.
Recall that an instance of (additive) regret-bounded vehicle routing 
consists of a metric space $(\{s\}\cup V,c)$ and a regret bound $R$;
the task is to compute a minimum number of paths, each beginning at $s$ and having total detour at most $R$, such that every element of $V$ belongs to one path.

\begin{theorem}
There is a $10$-approximation algorithm for regret-bounded vehicle routing.
\end{theorem}

\begin{proof}
For an instance $(V,s,c,R)$
let $\bar T=V$, $\Tscr=\{\bar T\}$, $\tau=\theta=\frac{1}{2}$, and let $b(t)$ be twice the number of paths that end in $t$ in an optimum solution.
We do not know $b(t)$ for any $t\in \bar T$, but we can guess $b(\bar T)=2\cdot\OPT(V,s,c,R)$ by trying all possible values $2,4,\ldots,2|V|$.
Then let $\Pi$ be the polyhedron \eqref{eq:FS_LP} for this $\bar T$, $\Tscr$, $\tau$, and $b(\bar T)$.
We set $x$ to be the incidence vector of an optimum solution to $(V,s,c,R)$ and apply the construction in Section~\ref{sec:existence_cheap_lp_solution},
but with the cost of $f$ induced by $\detour$ instead of $c$ (like in the work by Friggstad and Swamy \cite{friggstad2017compact}).
We get a vector $(a,f,z)\in\Pi$ such that $\detour(f) \le \theta \cdot b(\bar T) \cdot R$ and
$c(z)\le \frac{3}{4} \theta \cdot b(\bar T) \cdot R$.
So after adding these constraints the polyhedron $\Pi'$ is still nonempty, 
and we can find a vector $(a,f,z)\in\Pi'$ in polynomial time.
By the construction from the proof of Lemma~\ref{lem:rounding_lp_solution} (where the last part simplifies to setting $g=\frac{1}{\theta}g'$ because $\theta=1-\tau$), 
we get a solution $\Pscr$ with $b(\bar T)$ paths and
\begin{align*}
\detour(\Pscr) &\ \le \ \frac{4}{1-\theta} \cdot c(z) +  \frac{1}{\theta}  \cdot \detour(f) 
\ \le \ \left( \frac{3}{1-\theta} + \frac{1}{\theta} \right) \cdot R \cdot \OPT(V,s,c,R) 
\ = \ 8\cdot R \cdot \OPT(V,s,c,R).
\end{align*}
After splitting each path whenever it would exceed the detour bound $R$ 
(cf.\ Lemma~2.2 in \cite{friggstad}), we get a set of paths that all have detour at most $R$,
and the number of these paths is at most  
$b(\bar T)+ 8\cdot\OPT(V,s,c,R) = 10 \cdot \OPT(V,s,c,R)$.
\end{proof}

By the simple reduction from \cite{bock},
this immediately implies an 11-approximation algorithm for the school bus problem.

\section{Final result\label{sect:final_result}}

We now prove Theorems~\ref{thm:main_CVR} and \ref{thm:main_UD_and_S} simultaneously.
Let $\alpha>1$ such that there exists an $\alpha$-approximation algorithm for the traveling salesman problem.
Choose $\epsilon>0$ such that $2\epsilon+f(\epsilon) \le \alpha-1$,   
where $f$ is the function from Theorem~\ref{thm:CVR_3plus}.

Compute a traveling salesman tour by the given $\alpha$-approximation algorithm
and apply the tour partitioning of Theorem~\ref{thm:tour_splitting} to obtain a feasible solution.
If the instance is not difficult, this solution does the job: its cost is at most $(\alpha+2 \cdot (1-\epsilon))\OPT$
or $(\alpha+1-\epsilon)\OPT$, respectively. 

If the instance is difficult, we compute a solution as in Theorem~\ref{thm:CVR_3plus},  
which has cost at most $(3+f(\epsilon))\OPT \le (\alpha+ 2 \cdot (1 - \epsilon))\OPT$.

We now set constants that yield the final claims in Theorems~\ref{thm:main_CVR} and \ref{thm:main_UD_and_S} (for $\alpha=\frac{3}{2}$).
Applying Theorem~\ref{thm:mainVRPwTargetGroups2} for $\theta=1-\tau$ yields
$c(\Pscr) < \frac{1}{1-\tau} c(x) + \frac{3}{\tau} \detour(x |_{E_1(\Jeuscr)})$.
If the instance is difficult, plugging in the bounds from Lemma~\ref{lem:good_sol_exists} yields
$c(\Pscr) < ( \frac{1+\zeta}{1-\tau} + \frac{3\epsilon}{\tau} )\OPT$.
By Theorem~\ref{thm:good_tour_difficult_instances} we obtain a traveling salesman tour of cost 
at most
$$ \left( \frac{1+ \zeta}{1-\tau} + \frac{3\epsilon}{\tau} + \frac{3\rho}{(1-\rho)(1-\tau)} \right)\OPT. $$
Let $\tau=0.071$, $\rho=0.027$, and $\epsilon=\frac{1}{3000}$.
Using the definition of $\zeta$ in \eqref{eq:def_zeta}, our tour has length less than $\frac{3}{2}-\frac{2}{3000}$.
Together with Theorem~\ref{thm:tour_splitting}, this implies the final claims of Theorems~\ref{thm:main_CVR} and \ref{thm:main_UD_and_S}.

Using the combinatorial algorithm from Section~\ref{sec:combinatorial},
we get 
$c(\Pscr) < \frac{1}{1-\tau} c(x) + \frac{8}{\gamma-2} c(x|_{E_1(\Jeuscr)})
+(4+\frac{8}{\gamma-2}+2\gamma-\frac{\tau}{\tau-1})\detour(x |_{E_1(\Jeuscr)})$.
If the instance is difficult, plugging in the bounds from Lemma~\ref{lem:good_sol_exists} yields
$c(\Pscr) < ( \frac{1+\zeta}{1-\tau} + \frac{8}{\gamma-2} + \epsilon(4+\frac{8}{\gamma-2}+2\gamma-\frac{\tau}{\tau-1}) ) \OPT$.
Setting $\gamma=148$, $\tau=0.054$, $\rho=0.022$, and $\epsilon=\frac{1}{6000}$
yields a tour of length less than $\frac{3}{2}-\frac{2}{6000}$ in $O(n^3)$ time.
Combining this with Theorem~\ref{thm:tour_splitting} and the Christofides--Serdjukov algorithm,
we get an approximation ratio of $\frac{3}{2}-\frac{1}{3000}$ for \textsc{Capacitated Vehicle Routing}
and $\frac{3}{2}-\frac{1}{6000}$ for the unit-demand and splittable variants in $O(n^3)$ time.

\section{Comments on integrality ratios}

The classical tour partitioning algorithm (Section~\ref{sect:classical_algo}) yields also an upper bound on the integrality ratio of the following 
well-known linear programming relaxation for \textsc{Capacitated Vehicle Routing}:
\begin{equation}\label{eq:cvrp_lp}
\begin{aligned}
&\min \sum\limits_{v,w \in V \cup \{s\}} x_{\{v,w\} } \cdot c(v,w)  \\ &\text{ s.t. } &  x(\delta(A)) &\ge 2 \cdot  d(A)  & (A \subseteq V) \\
&& x(\delta(A)) &\ge 2 & (\emptyset\not= A \subseteq V) \\
&& x(\delta(v)) &= 2 & (v \in V) \\
&& x &\ge 0
\end{aligned}
\end{equation}
Let $\text{LP}$ denote the value of this linear program (for a given instance), 
and let $\beta$ denote the integrality ratio of the subtour relaxation of the traveling 
salesman problem, which results from \eqref{eq:cvrp_lp} by omitting the constraints
$x(\delta(A)) \ge 2 \cdot  d(A)$.
It is easy to see that $\sum_{v\in V}2d(v)c(s,v)\le \text{LP}$.
By Theorem~\ref{thm:tour_splitting}, this implies that the integrality ratio of
\eqref{eq:cvrp_lp} is at most $\beta+2$.
We know $\frac{4}{3}\le\beta\le\frac{3}{2}$ \cite{wolsey}.

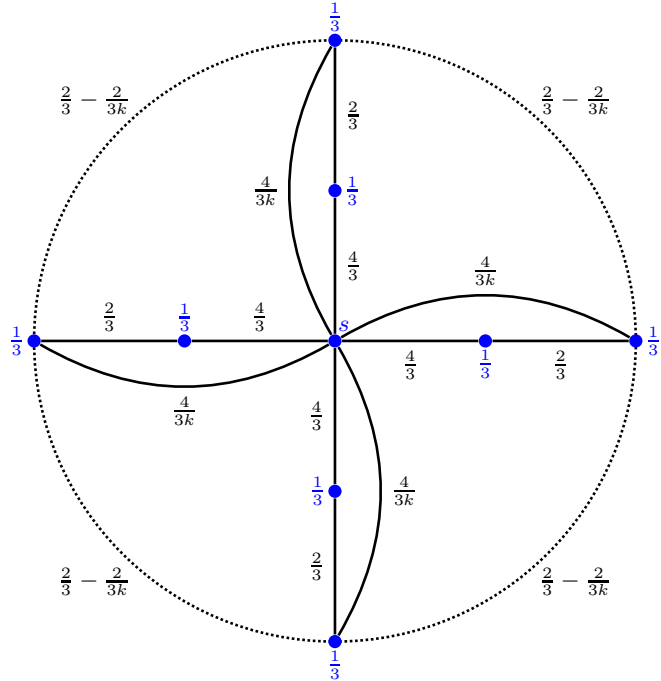
\begin{figure}[b!]
\begin{center}
\begin{tikzpicture}[scale=1]
\tikzstyle{vertex}=[blue,circle,fill,minimum size=5,inner sep=0pt]
\tikzstyle{edge}=[line width=1]
\tikzstyle{short}=[line width=1, densely dotted]

\begin{scope}[every node/.style={vertex}]
\draw[short] (4,4) circle (4);
 \node (v1) at (8,4) {};
 \node (v2) at (4,8) {};
 \node (v3) at (0,4) {};
 \node (v4) at (4,0) {};
 \node (w1) at (6,4) {};
 \node (w2) at (4,6) {};
 \node (w3) at (2,4) {};
 \node (w4) at (4,2) {};
 \node (s) at (4,4) {};
\end{scope}
\draw[edge] (s) to[bend left]  node[midway, above] {\small $\frac{4}{3k}$} (v1);
\draw[edge] (s) to[bend left]  node[midway, left] {\small $\frac{4}{3k}$} (v2);
\draw[edge] (s) to[bend left]  node[midway, below] {\small $\frac{4}{3k}$} (v3);
\draw[edge] (s) to[bend left]  node[midway, right] {\small $\frac{4}{3k}$} (v4);
\draw[edge] (s) -- (w1) node[midway, below] {\small $\frac{4}{3}$};
\draw[edge] (s) -- (w2) node[midway, right] {\small $\frac{4}{3}$};
\draw[edge] (s) -- (w3) node[midway, above] {\small $\frac{4}{3}$};
\draw[edge] (s) -- (w4) node[midway, left] {\small $\frac{4}{3}$};
\draw[edge] (w1) -- (v1) node[midway, below] {\small $\frac{2}{3}$};
\draw[edge] (w2) -- (v2) node[midway, right] {\small $\frac{2}{3}$};
\draw[edge] (w3) -- (v3) node[midway, above] {\small $\frac{2}{3}$};
\draw[edge] (w4) -- (v4) node[midway, left] {\small $\frac{2}{3}$};
\begin{scope}[blue]
\node at (4.12,4.2) {\small $s$};
\node[right] at (v1) {\small $\frac{1}{3}$};
\node[above] at (v2) {\small $\frac{1}{3}$};
\node[left] at (v3) {\small $\frac{1}{3}$};
\node[below] at (v4) {\small $\frac{1}{3}$};
\node[below] at (w1) {\small $\frac{1}{3}$};
\node[right] at (w2) {\small $\frac{1}{3}$};
\node[above] at (w3) {\small $\frac{1}{3}$};
\node[left] at (w4) {\small $\frac{1}{3}$};
\end{scope}
\node at (7.2,7.2) {\small $\frac{2}{3}-\frac{2}{3k}$};
\node at (0.8,7.2) {\small $\frac{2}{3}-\frac{2}{3k}$};
\node at (7.2,0.8) {\small $\frac{2}{3}-\frac{2}{3k}$};
\node at (0.8,0.8) {\small $\frac{2}{3}-\frac{2}{3k}$};Lemma~\ref{lem:rounding_lp_solution}.

\end{tikzpicture}
\end{center}
\caption{\label{fig:bad_example_for_integrality_gap} 
An instance of \textsc{Unit-Demand Capacitated Vehicle Routing} with $|V|=2k$ (here $k=4$). 
The capacity is 3, so $d(v)=\frac{1}{3}$ for all $v\in V$.
Every dotted edge has length 0, every solid straight edge has length 1, 
other distances are given by the metric closure of this graph.
(In particular, the solid bent edges have length 2.)
A feasible LP solution, even to the stronger linear program, is given by the numbers next to the edges; thus $\text{LP}'\le 2k+\frac{8}{3}$.
Note that $\sum_{v\in V}2d(v)c(s,v)=2k$.
For large enough $k$, we have 
$\sum_{v\in V}2d(v)c(s,v) > (1-\epsilon) \text{LP}'$. Nevertheless, only half of the customers are clustered.
}
\end{figure}

A natural question is whether our techniques lead to an improved bound on the 
integrality ratio. However, this seems to be difficult even for the
strengthened linear programming relaxation where we replace
$x(\delta(A)) \ge 2 \cdot  d(A)$ by $x(\delta(A)) \ge 2 \cdot  \lceil d(A)\rceil$.
Let $\text{LP}'$ denote the value of this stronger linear program.
A crucial fact in our proof was that difficult instances are ``clustered''.
One might hope that even all instances with $\sum_{v\in V}2d(v)c(s,v) > (1-\epsilon) \text{LP}'$
are ``clustered'', but this is not always the case as
the example in Figure~\ref{fig:bad_example_for_integrality_gap} shows.

We remark that solving the stronger linear program is NP-hard \cite{diarrassouba}, but it
can be approximated arbitrarily well because an approximate separation oracle can be
obtained by enumerating all cuts for which $x(\delta(A))$ is bounded by a constant (these can be enumerated in polynomial time by \cite{nagamochi}).

Proving stronger upper bounds on the integrality ratios remains an open question.

 \bibliographystyle{acm}
 \bibliography{references.bib}

\begin{thebibliography}{10}

\bibitem{Adamszek}
{\sc Adamaszek, A., Czumaj, A., and Lingas, A.}
\newblock {PTAS} for $k$-tour cover problem on the plane for moderately large
  values of $k$.
\newblock {\em International Journal of Foundations of Computer Science 21\/}
  (2010), 893--904.

\bibitem{altinkemer}
{\sc Altinkemer, K., and Gavish, B.}
\newblock Heuristics for unequal weight delivery problems with a fixed error
  guarantee.
\newblock {\em Operations Research Letters 6\/} (1987), 149--158.

\bibitem{Asano1}
{\sc Asano, T., Katoh, N., and Kawashima, K.}
\newblock A new approximation algorithm for the capacitated vehicle routing
  problem on a tree.
\newblock {\em Journal of Combinatorial Optimization 5\/} (2001), 213--231.

\bibitem{Asano2}
{\sc Asano, T., Katoh, N., Tamaki, H., and Tokuyama, T.}
\newblock Covering points in the plane by $k$-tours: towards a polynomial time
  approximation scheme for general $k$.
\newblock In {\em Proceedings of the Annual ACM Symposium on Theory of
  Computing (STOC)\/} (1997), pp.~275--283.

\bibitem{becker}
{\sc Becker, A.}
\newblock A tight 4/3 approximation for capacitated vehicle routing in trees.
\newblock In {\em Approximation, Randomization, and Combinatorial Optimization.
  Algorithms and Techniques (APPROX/RANDOM)\/} (2018), pp.~3:1--3:15.

\bibitem{becker2}
{\sc Becker, A., Klein, P.~N., and Saulpic, D.}
\newblock Polynomial-time approximation schemes for $k$-center, $k$-median, and
  capacitated vehicle routing in bounded highway dimension.
\newblock In {\em 26th Annual European Symposium on Algorithms (ESA)\/} (2018),
  pp.~8:1--8:15.

\bibitem{becker3}
{\sc Becker, A., Klein, P.~N., and Schild, A.}
\newblock A {PTAS} for bounded-capacity vehicle routing in planar graphs.
\newblock In {\em Algorithms and Data Structures (WADS)\/} (2019), pp.~99--111.

\bibitem{blum2007approximation}
{\sc Blum, A., Chawla, S., Karger, D.~R., Lane, T., Meyerson, A., and Minkoff,
  M.}
\newblock Approximation algorithms for orienteering and discounted-reward
  {TSP}.
\newblock {\em SIAM Journal on Computing 37\/} (2007), 653--670.

\bibitem{bock}
{\sc Bock, A., Grant, E., Könemann, J., and Sanità, L.}
\newblock The school bus problem on trees.
\newblock {\em Algorithmica 67\/} (2011), 10--19.

\bibitem{bompadre}
{\sc Bompadre, A., Dror, M., and Orlin, J.~B.}
\newblock Improved bounds for vehicle routing solutions.
\newblock {\em Discrete Optimization 3\/} (2006), 299--316.

\bibitem{christofides}
{\sc Christofides, N.}
\newblock Worst-case analysis of a new heuristic for the traveling salesman
  problem.
\newblock Tech. rep., {Carnegie-Mellon University}, 1976.

\bibitem{das}
{\sc Das, A., and Mathieu, C.}
\newblock A quasi-polynomial time approximation scheme for {E}uclidean
  capacitated vehicle routing.
\newblock {\em Algorithmica 73\/} (2015), 115--142.

\bibitem{diarrassouba}
{\sc Diarrassouba, I.}
\newblock On the complexity of the separation problem for rounded capacity
  inequalities.
\newblock {\em Discrete Optimization 25\/} (2017), 86--104.

\bibitem{EdmondsKarp}
{\sc Edmonds, J., and Karp, R.~M.}
\newblock Theoretical improvements in algorithmic efficiency for network flow
  problems.
\newblock {\em Journal of the ACM 19\/} (1972), 248--264.

\bibitem{friggstad}
{\sc Friggstad, Z., and Swamy, C.}
\newblock Approximation algorithms for regret-bounded vehicle routing and
  applications to distance-constrained vehicle routing.
\newblock In {\em Proceedings of the Annual ACM Symposium on Theory of
  Computing (STOC)\/} (2014), pp.~744--753.

\bibitem{friggstad2017compact}
{\sc Friggstad, Z., and Swamy, C.}
\newblock Compact, provably-good {LPs} for orienteering and regret-bounded
  vehicle routing.
\newblock In {\em International Conference on Integer Programming and
  Combinatorial Optimization (IPCO)\/} (2017), pp.~199--211.

\bibitem{GoemansWilliamson}
{\sc Goemans, M.~X., and Williamson, D.~P.}
\newblock Approximating minimum-cost graph problems with spanning tree edges.
\newblock {\em Operations Research Letters 16\/} (1994), 183--189.

\bibitem{haimovich}
{\sc Haimovich, M., and Rinnooy~Kan, A. H.~G.}
\newblock Bounds and heuristics for capacitated routing problems.
\newblock {\em Mathematics of Operations Research 10\/} (1985), 527--542.

\bibitem{Hamaguchi}
{\sc Hamaguchi, S., and Katoh, N.}
\newblock A capacitated vehicle routing problem on a tree.
\newblock In {\em Algorithms and Computation (ISAAC)\/} (1998), pp.~399--407.

\bibitem{karlin}
{\sc Karlin, A.~R., Klein, N., and Gharan, S.~O.}
\newblock A (slightly) improved approximation algorithm for metric {TSP}.
\newblock arXiv:2007.01409, 2020.

\bibitem{khachay}
{\sc Khachay, M., and Dubinin, R.}
\newblock {PTAS} for the {E}uclidean capacitated vehicle routing problem in
  $\mathbb{R}^d$.
\newblock In {\em 9th International Conference on Discrete Optimization and
  Operations Research (DOOR)\/} (2016), pp.~193--205.

\bibitem{Labbe}
{\sc Labbé, M., Laporte, G., and Mercure, H.}
\newblock Capacitated vehicle routing on trees.
\newblock {\em Operations Research 39\/} (1991), 616--622.

\bibitem{nagamochi}
{\sc Nagamochi, H., Nishimura, K., and Ibaraki, T.}
\newblock Computing all small cuts in an undirected network.
\newblock {\em SIAM Journal on Discrete Mathematics 10\/} (1997), 469--481.

\bibitem{serdjukov}
{\sc Serdjukov, A.}
\newblock Some extremal bypasses in graphs [in {R}ussian].
\newblock {\em Upravlyaemye Sistemy 17\/} (1978), 76--79.

\bibitem{Tomizawa}
{\sc Tomizawa, N.}
\newblock On some techniques useful for solution of transportation network
  problems.
\newblock {\em Networks 1\/} (1971), 173--194.

\bibitem{wolsey}
{\sc Wolsey, L.}
\newblock Heuristic analysis, linear programming and branch and bound.
\newblock {\em Mathematical Programming Study 13\/} (1980), 121--134.

\bibitem{wu}
{\sc Wu, Y., and Lu, X.}
\newblock Capacitated vehicle routing problem on line with unsplittable
  demands.
\newblock {\em Journal of Combinatorial Optimization, \textnormal{to appear}\/}
  (2020).

\end{thebibliography}
 
\end{document}